\documentclass[a4paper,american,10pt,DIV=12,includehead]{scrartcl}
\pdfoutput=1
\usepackage{lmodern}
\usepackage[utf8]{inputenc}
\usepackage[T1]{fontenc}

\usepackage{fouriernc}
\usepackage{newcent}
\usepackage{helvet}

\usepackage[notindex,nottoc,notlot,notlof]{tocbibind}
\usepackage{amsmath}
\usepackage{amsthm}
\usepackage{lastpage}
\usepackage[square,sort]{natbib}

\let\cite\citep
\let\citeN\citet
\let\citeNP\citealt
\let\citeyear\citeyearpar

\usepackage{float}
\floatstyle{ruled}
\newfloat{algorithm}{tbp}{loa}
\floatname{algorithm}{Algorithm}

\newcommand\tbl[2]{%
	\begin{small}\mbox{}\hfill{#2}\hfill\mbox{}\end{small}
	\caption{#1}%
}

\addtokomafont{caption}{\sffamily\small}
\addtokomafont{captionlabel}{\sffamily\textbf}
\setcapmargin{2em}

\date{\small\today}

\newenvironment{longitem}{\begin{itemize}}{\end{itemize}}
\newenvironment{longenum}{\begin{enumerate}}{\end{enumerate}}

\newcommand\received[3]{}

\theoremstyle{plain}
\newtheorem{theorem}{Theorem}[section]
\newtheorem{proposition}[theorem]{Proposition}
\newtheorem{lemma}[theorem]{Lemma}
\newtheorem{conjecture}[theorem]{Conjecture}

\renewcommand\proof[1][]{%
    \ifthenelse{\equal{#1}{}}{%
        \textit{Proof: }%
    }{%
        \textit{Proof #1:}%
    }%
}

\theoremstyle{definition}
\newtheorem{property}[theorem]{Property}
\newtheorem{fact}[theorem]{Fact}
\newtheorem{corollary}[theorem]{Corollary}

\usepackage{graphicx}
\usepackage{array}
\usepackage{booktabs}
\usepackage{multirow}
\usepackage{multicol}
\usepackage{ifthen}
\usepackage{pbox}
\usepackage{placeins}
\usepackage{url}
\usepackage{amsmath}
\usepackage{amssymb}
\usepackage{mathtools}
\usepackage{tikz}
\usepackage{clrscode3e}
\usepackage{babel}
\usepackage{xfrac}
\usepackage{colonequals}
\usepackage{enumitem}
\usepackage{xspace}
\usepackage[final]{listings}
\usepackage{pbox}
\usepackage{textcomp}
\usepackage{microtype}
\usepackage{needspace}

\usepackage[referable]{threeparttablex}

\usetikzlibrary{shapes.multipart,positioning,calc,external}

\tikzsetexternalprefix{plots/externalized/}
\tikzset{external/export=false} 

\usepackage{pgfplots}
\pgfplotsset{compat=1.5}
\pgfplotsset{
  every axis/.append style={
  font=\footnotesize,
  line width=0.5pt,
  tick style={thin}}
}
\pgfplotscreateplotcyclelist{bw}{%
mark=x\\%
mark=square*,every mark/.append style={fill=white}\\%
densely dashed,mark=*,every mark/.append style={solid,fill=white}\\%
densely dotted,mark=diamond*,every mark/.append style={solid,fill=white}\\%
mark=square*,every mark/.append style={fill=black!60}\\%
}

\pgfplotscreateplotcyclelist{asymptotics}{%
mark=x\\%
only marks,mark=*,every mark/.append style={fill=white}\\%
}

\usepackage[final,unicode=true,bookmarks=true,
            bookmarksnumbered=true,
            bookmarksopen=true,
            breaklinks=true,
            pdfborder={0 0 0},
            colorlinks=false]{hyperref}
\hypersetup{
 pdftitle={Average Case and Distributional Analysis of Dual-Pivot
	Quicksort},
 pdfauthor={Sebastian Wild, Markus E. Nebel and Ralph Neininger},
 pdfsubject={},
 pdfkeywords={Quicksort, Dual Pivot, Average Case Analysis}
}

\newcounter{basicblocknumber}
\newcounter{basicblocknumberIS}

\newcommand\weakemph[1]{\textsl{#1}}

\newdimen\makeboxdimen
\newcommand\makeboxlike[3][l]{%
\setbox0=\hbox{#2}%
\global\makeboxdimen=\wd0%
\setbox1=\hbox{\makebox[\makeboxdimen][#1]{%
\makebox[0pt][#1]{#3}%
}}%
\ht1=\ht0%
\dp1=\dp0%
\box1%
}

\allowdisplaybreaks[3]

\numberwithin{equation}{section}

\def\EndFor{\End\li\kw{end for} }
\def\EndIf{\End\li\kw{end if} }
\def\EndWhile{\End\li\kw{end while} }

\newcommand\E{\mathop{\mbox{$\mathbb{E}$}}\nolimits}
\newcommand\given{\;|\;}

\def\.{\mskip1mu}

\renewcommand\given{\mathbin{\mid}}
\newcommand\harm[1]{\ensuremath{\mathcal{H}_{#1}}}

\newcommand\ce{\colonequals}

\newcommand\pc{T\!}
\newcommand\insertsortcost{C^{\mathrm{IS}}}
\newcommand\sumpq{\sum_{1\le p<q\le n}}
\newcommand\probsumpq{%
	1 \big/ \tbinom n2 \!\!\!\!\!\! \sum_{1\le p<q\le n}\!\!\!
}

\newcommand\toll[2][]{%
	\ensuremath{%
	\ifthenelse{\equal{#1}{}}{%
		T_{\!#2}%
	}{%
		T_{\!#2}({#1})%
	}}%
}
\newcommand\comparisonmarker[2][]{\toll[#1]{\totalcomparisonmarker[]{#2}}}
\newcommand\swapmarker[2][]{\toll[#1]{\totalswapmarker[]{#2}}}
\newcommand\totalcomparisonmarker[2][]{%
	\ensuremath{%
	\ifthenelse{\equal{#2}{}}{%
		\ifthenelse{\equal{#1}{}}{%
			C%
		}{%
			C_{#1}%
		}%
	}{%
		\ifthenelse{\equal{#1}{}}{%
			\ui C{#2}%
		}{%
			\ui C{#2}_{#1}%
		}%
	}}%
}
\newcommand\totalswapmarker[2][]{%
	\ensuremath{%
	\ifthenelse{\equal{#2}{}}{%
		\ifthenelse{\equal{#1}{}}{%
			S%
		}{%
			S_{#1}%
		}%
	}{%
		\ifthenelse{\equal{#1}{}}{%
			\ui S{#2}%
		}{%
			\ui S{#2}_{#1}%
		}%
	}}%
}
\newcommand\totalwritesmarker[2][]{%
	\ensuremath{%
	\ifthenelse{\equal{#2}{}}{%
		\ifthenelse{\equal{#1}{}}{%
			W%
		}{%
			W_{#1}%
		}%
	}{%
		\ifthenelse{\equal{#1}{}}{%
			\ui W{#2}%
		}{%
			\ui W{#2}_{#1}%
		}%
	}}%
}
\newcommand\bytecodes[2][]{%
	\ensuremath{%
	\ifthenelse{\equal{#2}{}}{%
		\ifthenelse{\equal{#1}{}}{%
			\mathit{BC}%
		}{%
			\mathit{BC}_{#1}%
		}%
	}{%
		\ifthenelse{\equal{#1}{}}{%
			\ui{\mathit{BC}}{#2}%
		}{%
			\ui{\mathit{BC}}{#2}_{#1}%
		}%
	}}%
}

\newcommand\arrayA{\smash{\raisebox{-.2pt}{\scalebox{1.25}[1.18]{$\mathtt{A}$}}}\xspace}

\newcommand\values[1]{#1}

\newcommand\positionsets[1]{\mathcal{#1}}

\newcommand\numberat[2]{\values{#1}\.\mbox{\emph{@}}\.\positionsets{#2}}
\newcommand\satK{\numberat s{K}}

\newcommand\indicator[1]{\mathbb 1_{\{#1\}}}

\newcommand\rel[1]{\mathrel{\:{#1}\:}}
\newcommand\wrel[1]{\mathrel{\;{#1}\;}}
\newcommand\wwrel[1]{\mathrel{\;\;{#1}\;\;}}
\newcommand\bin[1]{\mathbin{\:{#1}\:}}
\newcommand\wbin[1]{\mathbin{\;{#1}\;}}
\newcommand\wwbin[1]{\mathbin{\;\;{#1}\;\;}}

\newcommand{\V}{\ensuremath{\mathrm{Var}}}
\newcommand{\Cov}{\ensuremath{\mathrm{Cov}}}
\newcommand{\Prob}{\ensuremath{\mathbb{P}}}
\newcommand{\Rset}{\ensuremath{\mathbb{R}}}
\newcommand{\Nset}{\ensuremath{\mathbb{N}}}
\newcommand{\bo}{\ensuremath{\mathrm{O}}}

\newcommand\bernoulli{\mathrm B}
\newcommand\hypergeometric{\mathrm{HypG}}
\newcommand\multinomial{\mathrm M}
\newcommand\binomial{\mathrm{Bin}}

\newcommand\ui[2]{#1^{\smash{(#2)}}}

\newcommand\eqdist{	
	\mathchoice{
		\mathrel{\overset{\raisebox{0ex}{$\scriptstyle \cal D$}}=}%
	}{
		\mathrel{\like{=}{%
			\overset{\raisebox{-1ex}{$\scriptscriptstyle \cal D$}}=%
		}}%
	}{
		\mathrel{\overset{\cal D}=}%
	}{
		\mathrel{\overset{\cal D}=}%
	}%
}

\newcommand\person[1]{#1}
\newcommand\pe{\phantom{{}={}}}
\newcommand\ppe{\phantom{=}}

\newcommand\tCpj{C_j{}^{\mkern-10mu \prime}\mkern3mu}
\newcommand\tCppj{C_{\!j\,}{}^{\mkern-10mu \prime\prime}\mkern1mu}
\newcommand\tCpppj{C_{\!j\,}{}^{\mkern-10mu \prime\prime\prime}\mkern1mu}
\newcommand\tCsn{C^*{}_{\mkern-13mu n}\mkern1mu}
\newcommand\tSsn{S^*{}_{\mkern-13mu n}\mkern1mu}
\newcommand\tBCsn{\mathit{BC}^*{}_{\mkern-13mu n}\mkern1mu}
\newcommand\tYsn{Y^*{}_{\mkern-17mu n}\mkern2mu}

\newcommand\like[3][c]{\makeboxlike[#1]{\ensuremath{#2}}{\ensuremath{#3}}}

\let\affil\thanks

\title{Average Case and Distributional Analysis of Dual-Pivot Quicksort%
}

\author{%
	Sebastian Wild${}^*$
	\and
	Markus E. Nebel%
	\affil{Computer Science Department, University of Kaiserslautern}
	$\,{}^{\text{\scriptsize,}}$%
	\affil{Department of Mathematics and Computer Science,
	University of Southern Denmark}
	\and Ralph Neininger%
	\affil{Institute for Mathematics, J.\,W.~Goethe University}
}

\begin{document}

\maketitle

\begin{abstract}
\rm
In 2009, Oracle replaced the long-serving sorting algorithm in its
Java~7 runtime library by a new dual-pivot Quicksort variant due to
Vladimir Yaroslavskiy.
The decision was based on the strikingly good performance of Yaroslavskiy's
implementation in running time experiments.
At that time, no precise investigations of the algorithm were available to
explain its superior performance\,---\,on the contrary:
Previous theoretical studies of other dual-pivot Quicksort variants even discouraged the
use of two pivots.
Only in 2012, two of the authors gave an average case analysis of a simplified
version of Yaroslavskiy's algorithm, proving that savings in the number of
comparisons are possible.
However, Yaroslavskiy's algorithm needs more swaps, which renders the analysis inconclusive.

To force the issue, we herein extend our analysis to the fully detailed
style of Knuth:
We determine the exact number of executed Java Bytecode instructions.
Surprisingly, Yaroslavskiy's algorithm needs sightly
\emph{more} Bytecode instructions than a simple implementation of classic
Quicksort\,---\,contradicting observed running times. Like in Oracle's library
implementation we incorporate the use of Insertionsort on small subproblems
and show that it indeed speeds up Yaroslavskiy's Quicksort in terms of
Bytecodes; but even with optimal Insertionsort thresholds
the new Quicksort variant needs slightly more Bytecode instructions on average.

Finally, we show that the (suitably normalized) costs of Yaroslavskiy's
algorithm converge to a random variable whose distribution is characterized by
a fixed-point equation.
From that, we compute variances of costs and show that for large $n$, costs are
concentrated around their mean.

\end{abstract}

\section{Introduction}
\label{sec:introduction}

Quicksort is a divide and conquer sorting algorithm originally proposed by
Hoare~\citeyear{Hoare1961b,Hoare1961a}.
The procedure starts by selecting an arbitrary element from the list to be
sorted as \emph{pivot}.
Then, Quicksort \emph{partitions} the elements into two groups:
those smaller than the pivot and those larger than the pivot.
After partitioning, we know the exact \emph{rank} of the pivot element in the
sorted list, so we can put it at its final landing position
between the groups of smaller and larger elements.
Afterwards, Quicksort proceeds by recursively sorting the two parts,
until it reaches lists of length zero or one, which are already sorted by
definition.

We will in the following always assume random access to the data, i.\,e.,
the elements are given as entries of an \emph{array}.
Then, the partitioning process can work \emph{in~place} by directly
manipulating the array.
This makes Quicksort convenient to use and avoids the need for extra 
space (except for the recursion stack).
Hoare's initial implementation \cite{Hoare1961b} works in place and
\citeN{Sedgewick1975} studies several variants thereof.

In the worst case, Quicksort has quadratic complexity, namely if in every
partitioning step, the pivot is the smallest or largest element of the
current subarray.
However, this behavior occurs very infrequently, such that the expected
complexity is $\Theta(n\log n)$.
\citeN{Hoare1962} already gives a precise average case analysis of his
algorithm, which is nowadays contained in most algorithms textbooks, [e.\,g.\
\citeNP{Cormen2009}].
Sedgewick~\citeyear{Sedgewick1975,Sedgewick1977} refines this analysis to count
the exact number of executed primitive instructions of a low level
implementation.
This detailed breakdown reveals that Quicksort has the asymptotically fastest
average running time on MIX among all the sorting algorithm studied by
\citeN{Knuth1998}.

Only considering average results can be misleading.
To increase our confidence in a sorting method, we also require that it is
\emph{likely} to observe costs close to the expectation.
The standard deviation of the Quicksort complexity
grows linearly with $n$ \cite{Knuth1998,hennequin1989combinatorial},
which implies that the costs are concentrated around their mean
for large~$n$.
Precise tail bounds which ensure tight concentration around the mean were
derived by \citeN{hamc95}; see also \citeN{fija02}.

Much more information is available on the full distribution of the number of
key comparisons.
When suitably normalized, the number of comparisons converges in law
\cite{Regnier}, with a certain unknown limit distribution.
\citeN{hennequin1989combinatorial} computed its first cumulants and proved that
it is not a normal distribution.
The limiting distribution can be implicitly characterized by a stochastic
fixed-point equation \cite{roesler1991limit} and it is known to have a smooth
density \cite{Tan1995,fija00}.

\smallskip
Due to its efficiency in the average, Quicksort has been used as
general purpose sorting method for decades, for example in the C/C++ standard
library and the Java runtime library.
As sorting is a widely used elementary task, even small speedups of such
library implementations can be worthwhile.
This caused a run on variations and modifications to the basic algorithm.
One very successful optimization is based on the observation that Quicksort's
performance on tiny subarrays is comparatively poor.
Therefore, we should switch to some special purpose sorting method for these
cases \cite{Hoare1962}.
\citeN{Singleton1969} proposed using Insertionsort for this task, which indeed
``for small $n$ [\,\dots] is about the best sorting method known''
according to \citeN[p.\,22]{Sedgewick1975}.
He also gives a precise analysis of Quicksort where Insertionsort is used for
subproblems of size less than~$M$ \cite{Sedgewick1977}.
For his MIX implementation, the optimal choice is $M=9$, which leads to a
speedup of 14\,\% for $n=10\,000$.

Another very successful optimization is to improve the choice of the pivot
element by selecting the \emph{median} of a small
sample of the current subarray.
This idea has been studied extensively
\cite{Hoare1962,Singleton1969,VanEmden1970,Sedgewick1977,hennequin1989combinatorial,Chern2001a,Martinez2001,Durand2003pseudonine},
and real world implementations make heavy use of it \cite{Bentley1993}.%

Precise analysis of the impact of a modification often helped in understanding
and assessing its usefulness,
and in fact, many proposed variations turned out detrimental in the end
(many examples are exposed by \citeN{Sedgewick1975}).
Partitioning with more than one pivot used to be counted among those.
\citeN[p.\,150ff]{Sedgewick1975} studies a dual-pivot Quicksort variant in
detail, but finds that it uses more swaps and comparisons than classic
Quicksort.%
\footnote{%
	Interestingly, tiny changes make Sedgewick's dual-pivot Quicksort
	competitive w.\,r.\,t.\ the number of comparisons; in fact it even needs only
	$28/15 n\ln n + \bo(n)$ comparisons \cite[Chapter\,5]{wild2012thesis}, which
	is less than Yaroslavskiy's algorithm!
	Yet, the many swaps dominate overall performance.	
}
Later \citeN{hennequin1991analyse} considers the general case of partitioning
into $s\ge 2$ partitions. For $s=3$, his Quicksort uses asymptotically the same
number of comparisons as classic Quicksort; for $s>3$, he attests minor
savings which, however, will not compensate for the much more complicated
partitioning process in practice.
These negative results may have discouraged further research along these lines
in the following two decades.

In 2009, however, Vladimir Yaroslavskiy presented his new dual-pivot Quicksort
variant at the Java core library mailing list.%
\footnote{%
	see e.\,g.\ the archive on
	\url{http://permalink.gmane.org/gmane.comp.java.openjdk.core-libs.devel/2628}
}
After promising running time benchmarks, Oracle decided to use Yaroslavskiy's
algorithm as default sorting method for arrays of primitive types%
\footnote{%
	Primitive types are all integer types as well as Boolean, character and
	floating point types.
	For arrays of objects, the library specification prescribes a stable sorting
	method, which Quicksort does not provide.
	Instead a variant of Mergesort is used, there.
}
in the Java~7 runtime library, even though literature did not offer an
explanation for the algorithm's good performance.

Only in 2012, \citeN{Wild2012} made a first step towards closing this gap by
giving exact expected numbers of swaps and comparisons for a simple version of
Yaroslavskiy's algorithm. 
We will re-derive these results here as a special case.

The surprising finding is that Yaroslavskiy's algorithm uses only $1.9n\ln n +
\bo(n)$ comparisons on average\,---\,asymptotically 5\,\% less than the $2 n\ln
n + \bo(n)$ comparisons needed by classic Quicksort.

\def\ck{\smash{$\mbox{\raisebox{-.2pt}{\scalebox{1.25}[1.15]{$\mathtt{C}$}}}_k$}\xspace}
\def\cg{\smash{$\mbox{\raisebox{-.2pt}{\scalebox{1.25}[1.15]{$\mathtt{C}$}}}{}_{\!\.g}$}\xspace}

The reason for the savings lies in the clever usage of stochastic dependencies:
Yaroslavskiy's algorithm contains two opposite pairs of locations~\ck
(lines~\ref*{lin:yaroslavskiy-comp-1} and~\ref*{lin:yaroslavskiy-comp-2} of
Algorithm~\ref{alg:yaroslavskiy}) and~\cg
(lines~\ref*{lin:yaroslavskiy-comp-3} and~\ref*{lin:yaroslavskiy-comp-4})
in the code where key comparisons are done: 
At~\ck, elements are first compared with the small pivot $p$ (in
line~\ref*{lin:yaroslavskiy-comp-1}) and then with the large pivot $q$ (in
line~\ref*{lin:yaroslavskiy-comp-2})\,---\,if still needed, i.\,e., only if the
element is larger than $p$.  
This means that we need only \emph{one} comparison to
identify a small element, whereas all other elements cost us a second
comparisons.
For \cg it is vice versa: We first compare with $q$, and thus large elements
are cheap to identify there.

By the way partitioning is organized, it happens that elements which are
initially to the \emph{right} of the final position of $q$ are classified at
\cg; whereas elements to the left are classified at \ck.
This implies that the \emph{number} of elements classified at \cg
\emph{co-varies} with the number of large elements: 
\cg is executed more often if there are more elements larger than $q$ (on
average) and similarly, \ck is visited often if there are many small elements.
Consequently, the probability that one comparison suffices to determine an
element's target partition is strictly larger than $1\./\.3$\,---\,which
would be the probability if \emph{all} elements are first compared to $p$ 
(or all first to $q$).
The asymmetric treatment of elements is the novelty that makes Yaroslavskiy's
algorithm superior to the dual-pivot partitioning schemes studied earlier.%
\footnote{%
	For details consider \cite{Wild2012} or the corresponding talk	at\\
	\url{http://www.slideshare.net/sebawild/average-case-analysis-of-java-7s-dual-pivot-quicksort}.
}

\smallskip
While the lower number of comparisons seems promising, Yaroslavskiy's dual-pivot
Quicksort needs more swaps than classic Quicksort,
so the high level analysis remains inconclusive.
In this paper, we extend our analysis to detailed instruction counts,
complementing previous work on classic Quicksort \cite{Sedgewick1977}.
However, instead of Knuth's slightly dated mythical machine MIX, we
consider the Java Virtual Machine \cite{Lindholm1999JVMSpec} and count the
number of executed Java Bytecode instructions.
\citeN{wild2012thesis} gives similar results for Knuth's MMIX
\cite{Knuth2005FascicleMMIX}, the successor of MIX.

The number of executed Bytecode instructions has been shown to resemble
actual running time \cite{Camesi2006BytecodeRuntime}, even though just-in-time
compilation can have a tremendous influence \cite{Wild2013Alenex} and some
aspects of modern processor architectures are neglected.

\smallskip
Extending the results of \citeN{Wild2012}, the analysis in this paper includes
sorting short subproblems with Insertionsort. 
Moreover, all previous results on Yaroslavskiy's algorithm only concern
expected behavior.
In this article, we show existence and give characterizations of limit
distributions.
A comforting result of these studies is that the standard deviation grows
linearly for Yaroslavskiy's algorithm as well, which implies concentration
around the mean.

This paper does not consider more refined ways to choose pivots, 
like selecting order statistics of a random sample.
We decided to defer a detailed treatment of Yaroslavskiy's algorithm
under this optimization to a separate article \cite{NebelWild2014}.

\medskip
The rest of this paper is organized as follows.
Section~\ref{sec:yaroslavskiys-algorithm} presents our object of
study.
In Section~\ref{sec:preliminaries}, we review basic notions used
in the analysis later. We also define our input model and collect elementary
properties of Yaroslavskiy's algorithm.
In Section~\ref{sec:average-case-analysis}, we derive exact average costs in
terms of comparisons, swaps and executed Bytecode instructions.
These are used in Section~\ref{sec:distributional-analysis} to identify a
limiting distribution of normalized costs in all three measures, from which we
obtain asymptotic variances.
Finally, Section~\ref{sec:conclusion} summarizes our findings and puts them in
context.

\subsection{Yaroslavskiy's Algorithm}
\label{sec:yaroslavskiys-algorithm}
\begin{algorithm}
	\begin{small}
	\begin{codebox}
\Procname{$\proc{QuicksortYaroslavskiy}\,(\arrayA,\id{left},\id{right})$}
\zi \Comment Sort $\arrayA[\id{left},\dots,\id{right}]$ (including end points).
\li \If $\id{right} - \id{left} < M$
	\qquad\qquad\Comment i.\,e.\ the subarray has $n\le M$ elements
\li \Then 
		\proc{InsertionSort}$(\arrayA,\id{left},\id{right})$
\li	\Else
\li		\If $\arrayA[\id{left}] > \arrayA[\id{right}]$
\li		\Then
			$p\gets \arrayA[\id{right}]$; \>\>\> $q\gets \arrayA[\id{left}]$
\li		\Else
\li			$p\gets \arrayA[\id{left}]$;  \>\>\> $q\gets \arrayA[\id{right}]$
		\EndIf
\li 	$\ell\gets \id{left} + 1$; 
 \quad $g\gets \id{right} - 1$; 
 \quad $k\gets \ell$ \label{lin:yaroslavskiy-init-l-g-k} 
\li		\While $k\le g$ 
\li 	\Do
			\If $\arrayA[k] < p$ \label{lin:yaroslavskiy-comp-1}
\li			\Then
				Swap $\arrayA[k]$ and $\arrayA[\ell]$ \label{lin:yaroslavskiy-swap-1}
\li				$\ell\gets \ell+1$ \label{lin:yaroslavskiy-l++-1}
\li			\Else 
\li				\If $\arrayA[k] \ge q$ \label{lin:yaroslavskiy-comp-2}
\li				\Then
					\While $\arrayA[g] > q$ and $k<g$ \kw{do} $g\gets g-1$ \kw{end while} \label{lin:yaroslavskiy-comp-3}
\li					\If $\arrayA[g] \ge p$ \label{lin:yaroslavskiy-comp-4}
\li					\Then
						Swap $\arrayA[k]$ and $\arrayA[g]$ \label{lin:yaroslavskiy-swap-2}
\li					\Else
\li						Swap $\arrayA[k]$ and $\arrayA[g]$;
	\;					Swap $\arrayA[k]$ and $\arrayA[\ell]$ \label{lin:yaroslavskiy-swap-3}
\li						$\ell\gets \ell+1$ \label{lin:yaroslavskiy-l++-2}
					\EndIf
\li					$g\gets g-1$ \label{lin:yaroslavskiy-g--}
				\EndIf
			\EndIf
\li			$k\gets k+1$ \label{lin:yaroslavskiy-k++}
		\EndWhile \label{lin:yaroslavskiy-end-while}
\li		$\ell\gets \ell-1$; \>\>\>$g\gets g+1$
\li		$\arrayA[\id{left}] \gets \arrayA[\ell]$; \>\>\>\> $\arrayA[\ell] \gets p$
			\label{lin:yaroslavskiy-swap-4} 
			\qquad\Comment Swap pivots to final position 
\li		$\arrayA[\id{right}] \gets \arrayA[g]$; \>\>\>\> $\arrayA[g] \gets q$
			\label{lin:yaroslavskiy-swap-5} 
\li		$\proc{QuicksortYaroslavskiy}\,(\arrayA,
				\makeboxlike[c]{$g+1$}{$\id{left}$},
				\makeboxlike[c]{$g+1$}{$\ell-1$}
			)$
		\label{lin:yaroslavskiy-call-1}
\li		$\proc{QuicksortYaroslavskiy}\,(\arrayA,
				\makeboxlike[c]{$g+1$}{$\ell+1$},
				\makeboxlike[c]{$g+1$}{$g-1$}
			)$
		\label{lin:yaroslavskiy-call-2}
\li		$\proc{QuicksortYaroslavskiy}\,(\arrayA,
				\makeboxlike[c]{$g+1$}{$g+1$},
				\makeboxlike[c]{$g+1$}{$\id{right}$}
			)$
		\label{lin:yaroslavskiy-call-3}
	\EndIf
\zi
\end{codebox}

	\end{small}
	\vspace{-3ex}
	\caption{\protect\rule[-.75ex]{0pt}{2.75ex}%
		\protect\person{Yaroslavskiy}'s Dual-Pivot Quicksort with Insertionsort.}
	\label{alg:yaroslavskiy}
\end{algorithm}
\noindent
Yaroslavskiy's dual-pivot Quicksort is shown in
Algorithm~\ref{alg:yaroslavskiy}.
The initial call to the procedure takes the form
\proc{QuicksortYaroslavskiy}$(\arrayA,1,n)$, where \arrayA is an array
containing the elements to be sorted and $n$ is its length.
After selecting the outermost elements as pivots $p$ and $q$ such that $p\le q$,
lines~\ref*{lin:yaroslavskiy-init-l-g-k}\,--\,\ref*{lin:yaroslavskiy-swap-5} of
Algorithm~\ref{alg:yaroslavskiy} comprise the partitioning method.
After that, all small elements, i.\,e., those smaller than $p$ (and $q$), form a
contiguous region at the left end of the array, followed by $p$ and the medium
elements.
Finally $q$ separates the medium and large elements.
After recursively sorting these three regions, the whole array is in order.

Yaroslavskiy's partitioning algorithm is an asymmetric generalization of Hoare's
crossing pointers technique: The index pointers $k$ and $g$ start at the left
and right ends, respectively, and are moved towards each other until they cross.
Additionally, pointer $\ell$ marks the position of the rightmost small element,
such that the array is kept invariably in the following form:

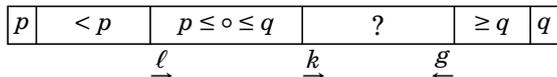
\begin{quote}
	\mbox{}\hfill%
	\begin{tikzpicture}[
		scale=0.5,
		baseline=(ref.south),
		every node/.style={font={}},
		semithick,
	]	
	
	\draw (-.75,0) -- ++(14.5,0) -- ++(0,1) -- ++(-14.5,0) -- cycle;
	\node at (-.375, .5) {$p$} ;
	\node at (13.375, .5) {$q$} ;
	\draw (0,0) -- ++(0,1);
	\draw (13,0) -- ++(0,1);
	
	\node at (1.5,0.5) {$< p$};
	\draw (3,1) -- ++ (0,-1);
	\node at (3.3,-0.4) {$\ell$};
	
	\node at (12,0.5) {$\ge q$};
	\draw (11,1) -- ++ (0,-1);
	\node at (10.7,-0.4) {$g$};
	
	\node at (5,0.5) {$p\le \circ\le q$};
	\draw (7,1) -- ++(0,-1);
	\node at (7.3,-0.4) {$k$};

 	\begin{pgfinterruptboundingbox}
	\node[below] at (10.7,-0.5) {$\leftarrow$};
	\node[below] at (3.3,-0.5) {$\rightarrow$};
	\node[below] at (7.3,-0.5) {$\rightarrow$};
 	\end{pgfinterruptboundingbox}
	
	\node[inner sep=0pt] (ref) at (9,0.5) {?};
	\end{tikzpicture}%
	\hfill\mbox{}
\end{quote}
\medskip

\noindent
Our Algorithm~\ref{alg:yaroslavskiy} differs from Algorithm~3 of
\cite{Wild2012} as follows:
\begin{longitem}
\item
	For lists of length less than $M$, we switch to \proc{InsertionSort}.%
	\footnote{%
		Note that even if \citeN{Sedgewick1977} proposes to use one final
		run of Insertionsort over the entire input array, modern cache hierarchies
		suggest to immediately sort small subarrays as done in our implementation.
	} 
	A possible implementation is given in Appendix~\ref{app:insertionsort}.
	The case $M=1$ corresponds to not using Insertionsort at all.
\item
	The swap of $\arrayA[k]$ and $\arrayA[g]$ has been moved behind the check
	$\arrayA[g] \ge p$. Thereby, we never use array positions in a key comparison
	after we have overwritten their contents in one partitioning step; see
	Fact~\ref{fact:fresh-elements} below.
	(This is just to simplify discussions.)
\item
	The comparison in line~\ref*{lin:yaroslavskiy-comp-2} has been made non-strict.
	For distinct elements this makes no difference, but it drastically improves
	performance in case of many equal keys \cite[p.\,54]{wild2012thesis}.
	The reader might find it instructive to consider the behavior on an array with
	all elements equal.
\end{longitem}

Note that partitioning an array around two pivots is similar in nature to the
\emph{Dutch National Flag Problem (DNFP)} posed by \citeN{Dijkstra1976} as a
programming exercise:
\begin{quote}
\sl Given an array of\/ $n$ red, white and blue pebbles, rearrange them by
swaps, such that the colors form the Dutch national flag: red, white and blue in
contiguous regions.
Each pebble may be inspected only once and only a constant amount of extra
storage may be used.
\end{quote}
Dijkstra assumes an operation ``buck'' that tells us an element's color in
one shot, so any algorithm must use exactly $n$ buck-operations.
Performance differences only concern the number of swaps needed.

Interestingly, Meyer gave an algorithm for the DNFP which
is essentially equivalent to Yaroslavskiy's partitioning method. Indeed, it even
outperforms the algorithm proposed by Dijkstra~\cite{McMaster1978}!
Yet, the real advantage of Yaroslavskiy's partitioning scheme\,---\,the reduced
expected number of key comparisons\,---\,is hidden by the atomic
buck operation; its potential use in Quicksort went unnoticed.

%
%
%

%
\section{Preliminaries}
\label{sec:preliminaries}

In this section, we recall elementary definitions and collect some notation and
basic facts used throughout this paper.

By $\harm n \ce \sum_{i=1}^n 1\mathbin/i$, we denote the $n$th \emph{Harmonic
Number}. We use $\delta_{i\!j}$ for the \emph{\person{Kronecker} delta}, which
is defined to be $1$ if $i=j$ and~$0$ otherwise. 
We define $x\ln(x)=0$ for $x=0$, so that $x\mapsto x\ln(x)$ becomes a continuous 
function on $[0,\infty)$.

The \emph{probability} of an event $E$ is denoted by $\Prob[E]$ and we write
$\indicator{E}$ for its \emph{indicator random variable}, which is~$1$ if the
event occurs and~$0$ otherwise.
For a random variable~$X$, let $\E[X]$, $\V(X)$ and ${\cal L}(X)$ denote its
\emph{expectation}, \emph{variance} and \emph{distribution}, respectively.
$X \mathrel{\smash{\eqdist}} Y$
means that $X$ has the \emph{same distribution} as $Y$.

By $\|X\|_p \ce \E[|X|^p]^{1/p}$, $1\le p<\infty$,
we denote the \emph{$L_p$-norm} of random variable
$X$.
For random variables $X_1,X_2,\ldots$ and $X$, we say $X_n$ \emph{converges
in $L_p$ to $X$}
$$
	X_n \wwrel{\overset{L_p}\longrightarrow} X
	\qquad\qquad\text{iff}\qquad\qquad
	\lim_{n\to\infty} \| X_n - X \|_p \wwrel= 0 \;.
$$

The \emph{Bernoulli distribution} with parameter $p$ is written as
$\bernoulli(p)$.
Provided that \smash{$\sum_{r=1}^b p_r = 1$} and  $b\ge 1$ is a fixed integer,
we denote by $\multinomial(n;p_1,\ldots,p_b)$ the \emph{multinomial
distribution} with $n$ trials and success probabilities $p_1,\ldots,p_b\in[\.0,1]$.
For \emph{random probabilities} $V=(V_1,\ldots,V_b)$, i.\,e., random variables
$0\le V_r\le 1$ ($r=1,\ldots,b$) with $\sum_{r=1}^b V_r =1$ almost surely, we
write
$
	Y \eqdist \multinomial(n;V_1,\ldots,V_b)
$
to denote that
$Y$ \emph{conditional} on $V=v$ (i.\,e., conditional on
$(V_1,\ldots,V_b)=(v_1,\ldots,v_b)$) is multinomially
$\multinomial(n;v_1,\ldots,v_b)$ distributed.

For $k,r,b\in\Nset$ satisfying $k\le r+b$, the \emph{hypergeometric
distribution} with $k$ trials from $r$ red and $b$ black balls is denoted by
$\mathrm{HypG}(k,r,r+b)$.
Given an urn with $r$ red and $b$~black balls, it is the distribution of
the number of red balls drawn when drawing $k$~times without replacement.
The mean and variance of a hypergeometrically $\mathrm{HypG}(k,r,r+b)$ distributed
random variable~$G$ are given by
\cite[p.\,127]{Kendall1945}
\begin{align}
\label{eq:mean-variance-hypergeometric}
	\E[G]	&\wwrel=	k\cdot\frac{r}{r+b}\,,&
	\V(G)	&\wwrel=	\frac{krb(r+b-k)}{(r+b)^2(r+b-1)}\;.
\end{align}
As for the multinomial distribution, given \emph{random} parameters $K$, $R$ in
$\{0,\ldots,n\}$ we use
$
	Y \eqdist \mathrm{HypG}(K,R,n)
$
to denote that
$Y$ \emph{conditional} on $(K,R)=(k,r)$ is hypergeometrically
$\mathrm{HypG}(k,r,n)$ distributed.

%
%

\subsection{Input Model}

We assume the \emph{random permutation model}:
The keys to be sorted are the integers $1,\ldots,n$ and each
permutation of $\{1,\dots,n\}$ has equal probability
$1\./n!$ to become the input.
Note that we implicitly exclude the case of equal keys by that.

As sorting is only concerned with the relative order of elements,
not the key values themselves, we can equivalently assume keys to be
i.\,i.\,d.\ real random variables from any (non-degenerate) continuous
distribution.
Equal keys do not occur almost surely and the \emph{ranks} of the elements form
in fact a random permutation of $\{1,\dots,n\}$ again
[see e.\,g.\ \citeNP{mahmoud2000sorting}]. For the analysis of
Section~\ref{sec:distributional-analysis}, this alternative point of view will
be helpful.

\subsection{Basic Properties of Yaroslavskiy's Algorithm}

As typical for divide and conquer algorithms, the analysis is based on
setting up a recurrence relation for the costs.
For such a recurrence to hold, it is vital that the costs for subproblems of
size~$k$ behave the same as the costs for dealing with an original random input
of initial size~$k$.
For Quicksort, we require the following property:
\begin{property}[Randomness Preservation]
\label{pro:randomness-preservation}
~\\If the whole input is a (uniformly chosen) random permutation of its
elements, so are the subproblems Quicksort is recursively invoked on.
\end{property}
\citeN{hennequin1989combinatorial} showed that
Property~\ref{pro:randomness-preservation}
is implied by the following property.
\begin{property}[Sufficient Condition for Randomness Preservation]
\label{pro:only-comparisons-with-pivots}
~\\Every key comparison involves a pivot element of the current partitioning
step.
\end{property}
Now, it is easy to verify that Yaroslavskiy's algorithm fulfills
Property~\ref{pro:only-comparisons-with-pivots} and hence
Property~\ref{pro:randomness-preservation}.

\medskip\noindent
Since Yaroslavskiy's algorithm is an in-place sorting method,
it modifies the array~$\arrayA$ over time.
This dynamic component makes discussions inconvenient.
Fortunately, a sharp look at the algorithm reveals the following fact,
allowing a more static point of view:
\begin{fact}
\label{fact:fresh-elements}
	The array elements used in key comparisons have not been changed since the
	beginning of the current partitioning step.
	More precisely, if a key comparison involves an array element $\arrayA[i]$,
	then there has not been a write access to $\arrayA[i]$ in the current
	partitioning step.
	\qed
\end{fact}
\section{Average Case Analysis}
\label{sec:average-case-analysis}

Throughout Section~\ref*{sec:average-case-analysis}, we assume that array
\arrayA stores a random permutation of~$\{1,\ldots,n\}$.

\subsection{The Dual-Pivot Quicksort Recurrence}
\label{sec:recurrence-solution}

In this section, we obtain a general solution to the
recurrence relation corresponding to dual-pivot Quicksort.
We denote by 
$\E[C_n]$ 
the \emph{expected costs}\,---\,where different
cost measures will be inserted later\,---\,of Yaroslavskiy's algorithm on a
random permutation of $\{1,\ldots,n\}$.
$\E[C_n]$ decomposes as
\begin{align}
	\E[C_n] &\wwrel= \text{costs of first partitioning step}
						\wbin+ \text{costs for subproblems}
				\;.
\end{align}
As Yaroslavskiy's algorithm satisfies
Property~\ref{pro:randomness-preservation}, the costs for recursively sorting
subarrays can be expressed in terms of $C$ with smaller arguments,
leading to a recurrence relation.
Every (sorted) pair of elements has the same probability $1\big/\binom n2$ of
becoming pivots.
Conditioning on the ranks of the pivots, this gives the following recursive form
for the expected costs $\E[C_n]$ of Yaroslavskiy's algorithm on a random
permutation of size~$n$:
\begin{align}
\label{eq:recurrence}
	\E[C_n] &\wrel= \begin{cases}
			\displaystyle
			\E[\pc_n] \wbin+ 1 \big/ \tbinom n2
			\!\!\! \sum_{1\le p<q\le n}\!\!\! \bigl(
					  \E[C_{p-1}] + \E[C_{q-p-1}] + \E[C_{n-q}]
				\bigr)\,,
			& \text{for } n>M; \\
		\E[\insertsortcost_n]\,,
			& \text{for } n\le M,
	\end{cases}
\end{align}
where $\insertsortcost_n$ denotes the costs of
\proc{InsertionSort}\kern.5pt ing a random permutation of $\{1,\ldots,n\}$
and $\pc_n$ is the cost contribution of the \emph{first}
partitioning step.
This function $\pc_n$ quantifies the ``toll'' we have to pay for unfolding the
recurrence once, therefore we will call $\pc_n$ the \emph{toll function} of the
recurrence.
By adapting the toll function, we can use the same
recurrence to describe different kinds of costs and we only need to derive
a general solution to this single recurrence relation as provided by the
following theorem:

\begin{theorem}
\label{thm:recurrence-solution}
	Let $\E[C_n]$ be recursively defined by \eqref{eq:recurrence}. Then, $\E[C_n]$
	satisfies
	\begin{multline}
	\label{eq:Cn-explicit}
		\E[C_n] \wwrel= \frac1{\tbinom n4}
				\sum_{i=M+4}^{n}\!\!\!\tbinom i4\!\!\! \sum_{j=M+2}^{i-2}
					\biggl(
						      \E[\pc_{j+2}]
						\bin- \tfrac{2j}{j+2} \E[\pc_{j+1}]
						\bin+ \tfrac{\tbinom j2}{\tbinom{j+2}2} \E[\pc_j]
					\biggr) \\*
			 \wbin+ \biggl(
					  \! \tfrac{n+1}5
					+ \frac{\tbinom{M+3}4-\tbinom{M+4}5}{\tbinom n4}
				\biggr) \E[C_{M+3}]
				\wwbin- \tfrac{M-1}{M+3} \biggl(
					  \! \tfrac{n+1}5
					- \frac{\tbinom{M+4}5}{\tbinom n4}
				\biggr) \E[C_{M+2}]\,,
		\qquad \text{for } n\ge M+3.
	\end{multline}
	As an immediate consequence, $\E[C_n]$\,---\,seen as a function of
	$\E[\pc_n]$\,---\,is \emph{linear} in $\E[\pc_n]$.
\end{theorem}
The proof for Theorem~\ref{thm:recurrence-solution} uses several layers
of successive differences of $\E[C_n]$ to finally obtain a telescoping
recurrence. Substituting back in then yields~\eqref{eq:Cn-explicit}. 
The detailed computations are given in Appendix~\ref{app:recurrence-proofs}.
This general solution still involves non-trivial double sums.
For the cost measures we are
interested in, the following proposition gives an
explicit solution for~\eqref{eq:recurrence}.

\begin{proposition}
\label{pro:Cn-explicit-linear-pc}
	Let $\E[C_n]$ be recursively defined by \eqref{eq:recurrence} and let
	$\E[\pc_n] = an + b$ for $n \ge M+1$.
	Then, $\E[C_n]$ satisfies
	\begin{multline}
	\label{eq:Cn-exact-solution-linear-pc}
		\E[C_n] \wwrel=
					\tfrac65 a (n+1) \bigl( \harm{n+1} - \harm{M+2} \bigr)
				\wbin+ \tfrac15 (n+1) \bigl(
					  \tfrac{19}5 a
					+ \tfrac{6(b-a)}{M+2}
				\bigr)
				\wbin+ \tfrac{a-b}2
					\\*
				\wbin+
					\tfrac15 (n+1)
					\sum_{k=0}^M \frac{3M-2k}{\tbinom{M+2}3} \E[\insertsortcost_k]
				\wwbin+
					\frac{\tbinom{M+4}5}{\tbinom n4} R_M
				\,,
		\qquad \text{for } n\ge M+3,
	\end{multline}
	where	
	$$
		R_M \wwrel=    
			  \tfrac65 a + \tfrac{2(a-b)}{M+3}+\tfrac{5b-17a}{2(M+4)}
			- \tfrac{M-1}{M+4} \E[C_{M+3}]
			+ \tfrac{M-1}{M+3} \E[C_{M+2}]\;.
	$$%
	If $\E[\insertsortcost_n = 0]$ for all $n$, $\E[C_n]$ has the following
	asymptotic representation:
	\begin{align}
	\label{eq:Cn-explicit-linear-pc}
		\E[C_n] &\wwrel =
				\tfrac65 a n\ln n
			\wbin+ \bigl(
						  \tfrac{19}{25}a + W
					\bigr)\.n
			\wbin+	 \tfrac65 a \ln n
			\wbin+ \bigl( \tfrac{153}{50}a - \tfrac12b + W \bigr)
			\wwbin+ \bo\bigl(\tfrac1n\bigr)
				\,,\qquad n\to\infty,
	\end{align}
    where
    \begin{align*}
		W &\wwrel= \tfrac65 \bigl(
						   a \, \gamma
						+ \tfrac{b-a}{M+2}
						- a \harm{M+2}
					\bigr)
	\end{align*}
    and $\gamma \approx 0.57721$ is the \person{Euler-Mascheroni} constant.

    \smallskip
    Moreover, if the toll function $\E[\pc_n]$ has essentially the form given
	above, but with $\E[\pc_2] = 0 \ne 2a+b$,
	we get an additional summand
	$-\delta_{M1} \cdot \smash{\tfrac1{10}}	(2a+b) \cdot (n+1)$
	in~\eqref{eq:Cn-exact-solution-linear-pc}.
	Equation~\eqref{eq:Cn-explicit-linear-pc} remains valid if we set
	$W=\tfrac65 \bigl(
						   a \, \gamma
						+ \tfrac{b-a}{M+2}
						- a \harm{M+2}
						- \delta_{M1}(2a+b)\mathbin/10
					\bigr)$.
\end{proposition}
The proof of Proposition~\ref{pro:Cn-explicit-linear-pc} is basically ``by
computing'', the details are again deferred to
Appendix~\ref{app:recurrence-proofs}.

\emph{Remark:}
For constant $M$, i.\,e., $M=\Theta(1)$ as $n\to\infty$, only the linear term of
the expected costs is affected by $M$.
This means that for the leading term of $\E[C_n]$, the ``base case strategy''
for solving small subproblems is totally irrelevant.

\subsection{Basic Block Execution Frequencies}
\label{sec:expected-frequencies}

In this section, we compute for every single instruction of Yaroslavskiy's
algorithm  how often it is executed in expectation.
Based on that, we can easily derive the expected number of key comparisons,
swaps, but also more detailed measures, such as the expected number of
executed Bytecode instructions.
This is the kind of analysis \person{Knuth} popularized through his book series
\weakemph{The Art of Computer Programming} \cite{Knuth1998}.
A~corresponding analysis of classic single pivot Quicksort was done by
\citeN{Sedgewick1977}. Like the Quicksort variant discussed there,
Algorithm~\ref{alg:yaroslavskiy} uses Insertionsort for
sorting small subarrays. 
Our detailed implementation of Insertionsort and
its analysis are given in Appendix~\ref{app:insertionsort}.

Consecutive lines of purely sequential%
\footnote{%
	Purely sequential blocks contain neither (outgoing) jumps, nor targets for
	(incoming) jumps from other locations, except for the last and first
	instructions, respectively.	
}
code always have the same execution
frequencies; 
contracting maximal blocks of such code yields the control flow
graph (CFG).
Figure~\ref{fig:control-flow-graph} shows the resulting CFG for
Yaroslavskiy's algorithm.
Simple flow conservation arguments (a.\,k.\,a.\ Kirchhoff's laws) allow to
express execution frequencies of some blocks by the frequencies of others:
The execution frequencies of the $20$~basic blocks of
Figure~\ref{fig:control-flow-graph} only depend on the following nine
frequencies: $A$, $B$, $R$, $F$, $\totalcomparisonmarker{1}$,
$\totalcomparisonmarker{3}$, $\totalcomparisonmarker{4}$,
$\totalswapmarker{1}$ and~$\totalswapmarker{3}$.
The name $\totalcomparisonmarker{i}$ indicates that this frequency counts
executions of the $i$th location in the code of Algorithm~\ref{alg:yaroslavskiy},
where a key comparison is done.
Similarly, $\totalswapmarker{i}$ corresponds to the $i$th swap
location.

\begin{figure}
	\mbox{}\hfill%
	\begin{tikzpicture}[
	node distance=5mm,
	shorten >=.75pt,
	every node/.style={font={\scriptsize},inner sep=2.5pt},
	basic block/.style={
		fill=black!5,
		draw,
		shape=rectangle split,
		rectangle split	parts=2,
		rounded corners=2pt,
	}
]

\newcommand{\block}[5][1.5cm]{%
	\let\C\totalcomparisonmarker%
	\let\S\totalswapmarker%
	\setcounter{basicblocknumber}{#2}%
	\addtocounter{basicblocknumber}{-1}%
	\refstepcounter{basicblocknumber}%
	\makebox[#1]{\,\textbf{#2}%
	\label{bb:#2}%
	\:\hfill%
		\smash{$#4$}{\tiny$\vphantom{\totalcomparisonmarker{|}}$}\,%
	}%
	\nodepart{two}\parbox{#1}{%
		\centering \pbox{#1}{%
			\setcounter{basicblocknumber}{#2}%
			\addtocounter{basicblocknumber}{-1}%
			\refstepcounter{basicblocknumber}%
			#5%
		}
	}%
}
\def\ak{\arrayA[k]}
\def\Swap#1#2{Swap $\arrayA[#1]$, $\arrayA[#2]$;}
\node[basic block] (b1) at (0,0) 
	{\block[2cm]{1}{?}{R}{%
		$\id{right} - \id{left} < M$
	}} ;
\node[basic block,left=10mm of b1] (b3a)
	{\block[2cm]{3}{}{A}{%
		$\arrayA[\id{left}] > \arrayA[\id{right}]$
		\label{bb:comp-0}
	}} ;
\node[basic block,below left=5mm and -5mm of b3a] (b3b)
	{\block[1.75cm]{4}{}{B}{%
		$p\gets \arrayA[\id{right}]$; \\
		$q\gets \arrayA[\id{left}]$;
		\label{bb:p-right-q-left}
	}} ;
\node[basic block,below right=5mm and -5mm of b3a] (b3c)
	{\block[1.75cm]{5}{}{A-B}{%
		$p\gets \arrayA[\id{left}]$; \\
		$q\gets \arrayA[\id{right}]$;
		\label{bb:p-left-q-right}
	}} ;
\node[basic block,below=21mm of b3a] (b3d) 
	{\block[4cm]{6}{?}{A}{%
		$\like[l]{n\strut}{\ell}\gets \id{left} + 1$; \;
		$\like[l]{n\strut}{g}\gets \id{right} - 1$; \;
		$\like[l]{n\strut}{k}\gets \ell$;
		\label{bb:init-k-l-g}
	}} ; 
\node[basic block,below=of b3d] (b4) 
	{\block{7}{3}{A+\C1}{%
		$k\le g$
		\label{bb:k-le-g}
	}} ;
\node[basic block,below=of b4] (b5)
	{\block{8}{7}{\C1}{%
		$\ak < p$
		\label{bb:comp-1}
	}} ;
\node[basic block,below=of b5] (b7)
	{\block{10}{3}{\C1-\S1}{%
		$\ak\ge q$
		\label{bb:comp-2}
	}} ;
\node[basic block,left=of b7] (b6)
	{\block[2cm]{9}{12}{\S1}{%
		\Swap k\ell \\
		$\ell\gets\ell+1$;
		\label{bb:swap-1}
	}} ;
\node[basic block,right=of b7] (b8)
	{\block{11}{5}{\C3}{%
		$\arrayA[g] > q$
		\label{bb:comp-3}
	}} ;
\node[basic block,right=of b8] (b9) 
	{\block[1.75cm]{12}{3}{\scriptstyle{\C3-\C4+F}}{%
		$k < g$
		\label{bb:inner-k-less-g}
	}} ;
\node[basic block,below right=5mm and 10mm of b9] (b2)
	{\block[3.5cm]{2}{?}{R-A}{%
		\proc{InsertionSort$\,(\arrayA,\id{left},\id{right})$};\strut
		\label{bb:call-insertionsort}
	}} ; 
\node[basic block,above=of b8] (b10) 
	{\block{13}{2}{\C3-\C4}{%
		$g\gets g-1$;
		\label{bb:inner-g--}
	}} ;
\node[basic block,below=of b8] (b11) 
	{\block{14}{5}{\C4}{%
		$\arrayA[g] \ge p$
		\label{bb:comp-4}
	}} ;
\node[basic block,below right=5mm and -5mm of b11] (b13)
	{\block[2cm]{16}{14}{\S3}{%
		\Swap kg \\
		\Swap k\ell \\
		$\ell \gets \ell+1$;
		\label{bb:swap-3}
	}} ;
\node[basic block,left=of b13] (b12)
	{\block[2cm]{15}{6}{\C4-\S3}{%
		\Swap kg
		\label{bb:swap-2}
	}} ;
\node[basic block,below left= 5mm and -5mm of b13] (b14)
	{\block{17}{5}{\C4}{%
		$g\gets g-1$;
		\label{bb:outer-g--}
	}} ;
\node[basic block,left=1.5cm of b14] (b15)
	{\block{18}{2}{\C1}{%
		$k\gets k+1$;
		\label{bb:k++}
	}} ;
\node[basic block,below= 4.75cm of b7] (b16)
	{\block[5.5cm]{19}{?}{A}{%
		\makebox[2cm][l]{$\ell \gets \ell - 1$;}
		$g \gets g + 1$; \\[.25ex]
		\makebox[2cm][l]{$\arrayA[\id{left}] \gets \arrayA[\ell]$;}
		$\arrayA[\ell] \gets p;$ \\
		\makebox[2cm][l]{$\arrayA[\id{right}] \gets \arrayA[g]$;}
		$\arrayA[g] \gets q;$ \\[.25ex]
		$\proc{QuicksortYaroslavskiy}\,(\arrayA,
				\makeboxlike[c]{$g+1$}{$\id{left}$},
				\makeboxlike[c]{$g+1$}{$\ell-1$}
			)$; \\
		$\proc{QuicksortYaroslavskiy}\,(\arrayA,
				\makeboxlike[c]{$g+1$}{$\ell+1$},
				\makeboxlike[c]{$g+1$}{$g-1$}
			)$; \\
		$\proc{QuicksortYaroslavskiy}\,(\arrayA,
				\makeboxlike[c]{$g+1$}{$g+1$},
				\makeboxlike[c]{$g+1$}{$\id{right}$}
			)$;		
		\label{bb:recursive-calls}
		\\[-.75\baselineskip]
	}} ;
\node[basic block,below=of b16] (b20)
	{\block{20}{}{R}{%
		\weakemph{Return};
	}} ;
\coordinate[above= 2.5mm of b3b] (ab3b) {};
\coordinate[above= 2.5mm of b3c] (ab3c) {};
\coordinate[above= 2.5mm of b3d] (ab3d) {};
\coordinate[above= 3cm of b15] (b6t15) {} ;
\coordinate[above= 2.5mm of b11] (b9t11) {};
\coordinate[below left= 3.5mm of b11] (ab12) {};
\coordinate[below right= 3.5mm of b11] (ab13) {};
\coordinate[above= 2.5mm of b14] (ab14) {};
\coordinate[left=5mm of b6] (lb6) {};
\coordinate[right=5mm of b9] (rb9) {};
\coordinate[below=5mm of b16] (bb16) {};
\begin{scope}[->,auto,every node/.style={font={\tiny}}]
\begin{pgfinterruptboundingbox}
	\draw (0,1) -- (b1) ;
\end{pgfinterruptboundingbox}
\draw (b1) -- node[above]{no}  (b3a) ;
\draw (b1) -| node[pos=.125]{yes} (b2) ;
\draw (b3a.south) ++(-.1,0) |- (ab3b) -| node[above,pos=.45] {yes} (b3b) ;
\draw (b3a.south) ++( .1,0) |- (ab3c) -| node[above,pos=.45] {no}  (b3c) ;
\draw (b3b) |- (ab3d) -- (b3d) ;
\draw[-,shorten >=0pt] (b3c) |- (ab3d) ;
\draw (b3d) -- (b4) ;
\draw (b4) -| node[pos=.125]{no} (rb9) |- ($(b16.east) + (0,.5)$);
\draw (b4) -- node{yes} (b5);
\draw (b5) -- node{no} (b7);
\draw (b5) -| node[pos=.25]{yes} (b6);
\draw (b6) |- (b6t15) -- (b15);
\draw[-,shorten >=0pt] (b7) |- node[pos=.125]{no} (b6t15) ;
\draw (b7) -- node{yes} (b8) ;
\draw (b8) -- node{yes} (b9) ;
\draw[thick] (b9) |- node[pos=.25,right]{yes} (b10) ;
\draw (b10) -- (b8) ; 
\draw (b8) -- node[left]{no} (b11) ;
\draw (b9) |- node[pos=.625,below]{no} (b9t11) -- (b11) ;
\draw (b11.south) ++(-.1,0) |-node[pos=1,above]{no} (ab12) -| (b12) ;
\draw (b11.south) ++(.1,0) |-node[pos=1,above]{yes} (ab13) -| (b13) ; 
\draw (b12) |- (ab14) -- (b14) ;
\draw[-,shorten >=0pt] (b13) |- (ab14) ;
\draw (b14) -- (b15) ;
\draw[thick] (b15) -| (lb6) |- (b4) ;
\draw (b2) |- (b20) ;
\draw (b16) -- (b20) ;
\begin{pgfinterruptboundingbox}
	\draw (b20.south) -- ++(0,-.4) ;
	\begin{scope}[dashed,rounded corners=6pt]
		\coordinate (start rec call1) at ($(b16.south) + (2.45,.8)$) ;
		\coordinate (start rec call2) at ($(b16.south) + (2.45,.225)$) ;
		\coordinate (start rec call3) at ($(b16.south) + (2.45,.512)$) ;
		\coordinate (end rec call1) at ($(b16.south) + (-2.45,.525)$) ;
		\coordinate (end rec call2) at ($(b16.south) + (-2.45,.25)$) ;
		\draw (start rec call1) -|	++(1.5,.65) ;
		\draw[-] (start rec call2) -|	++(1.5,.5) ;
		\draw[-] (start rec call3) -|	++(1.5,.5) ;
		\draw (b20.south) ++(-.2,0) |- ++(-3,-.25) |- (end rec call1) ;
		\draw (b20.south) ++(-.2,0)    ++(-3,1.3) |- (end rec call2) ;
	\end{scope}
\end{pgfinterruptboundingbox}
\end{scope}
\end{tikzpicture}%
	\hfill\mbox{}
	\caption{%
		Control flow graph for Algorithm~\ref{alg:yaroslavskiy}.
		The algorithm is decomposed into \emph{basic blocks} of purely sequential
		code.
		Possible transitions from one block to another are indicated by arrows.
		Blocks with two outgoing arrows end with a conditional, the ``yes'' path is
		taken if the condition is fulfilled, otherwise the ``no'' transition is
		chosen.
		We refer to blocks using the number shown in the upper left corner.
		In the upper right corner, a block's symbolic \emph{execution frequency} is
		given.
		For clarity of presentation, the recursive calls in block~\ref{bb:recursive-calls} are
		not explicitly shown, but only sketched by the dashed arrows.
		Block~\ref{bb:call-insertionsort} calls \proc{InsertionSort} which is given in
		Appendix~\ref{app:insertionsort}.
	}
	\label{fig:control-flow-graph}
\end{figure}
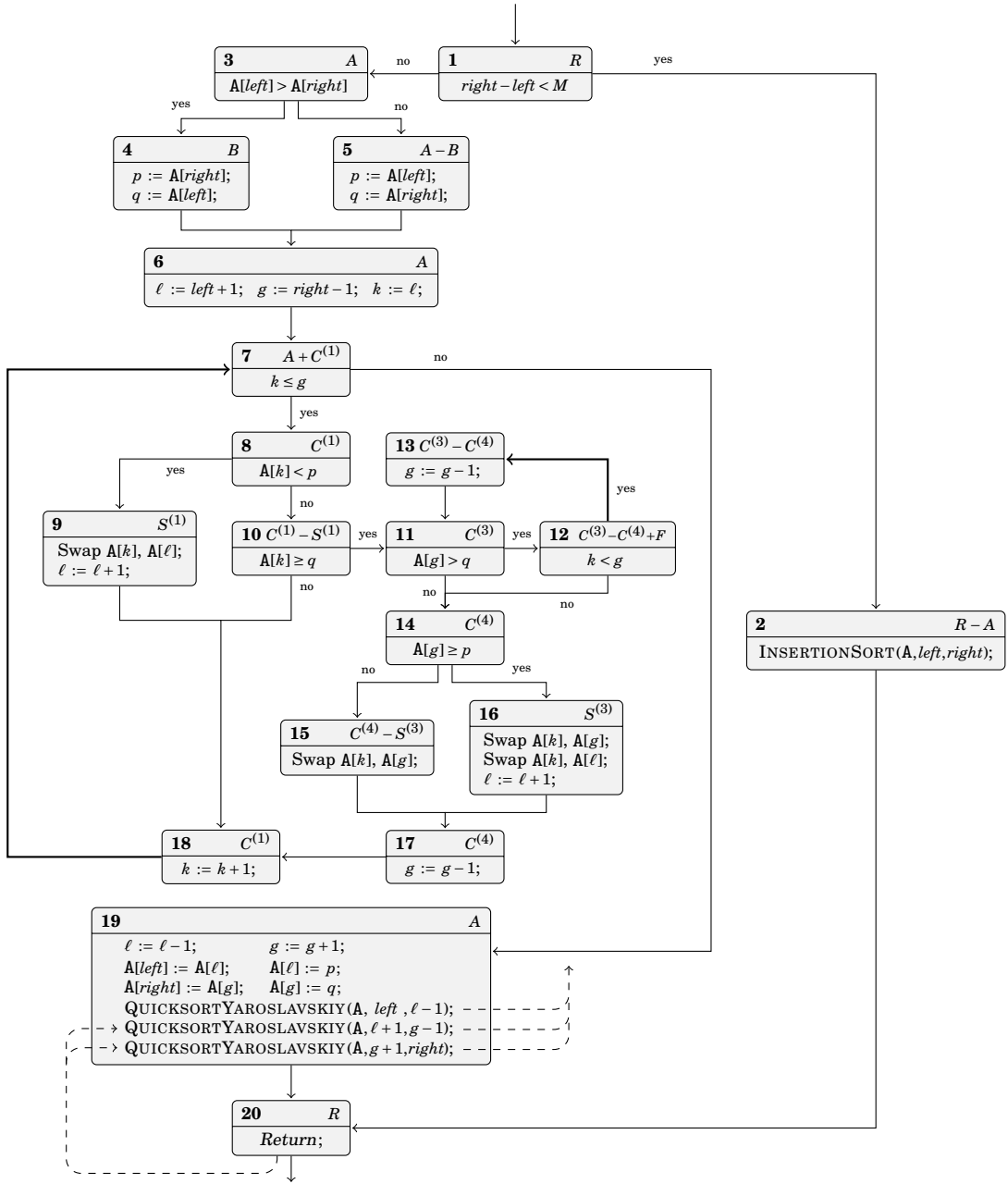

The results are summarized in
Tables~\ref{tab:expectation-of-toll-functions} and~\ref{tab:expected-frequencies}
at the end of the section.

\medskip
The expected execution frequencies allow a recursive representation of the
following form, here using the example of $\totalcomparisonmarker{1}$:
\begin{align}
	\E[\totalcomparisonmarker[n]1]
	&\wwrel= \begin{cases}
		\displaystyle \E[\comparisonmarker[n]{1}] 
		\displaystyle \wwbin+ \probsumpq \bigl(
				  \E[\totalcomparisonmarker[p-1]{1}]
				+ \E[\totalcomparisonmarker[q-p-1]{1}]
				+ \E[\totalcomparisonmarker[n-q]{1}]
			\bigr)\,,
		  & \text{for } n > M; \\
		0, & \text{for } n \le M,
	\end{cases}\,,
\label{eq:recurrence-for-execution-frequencies}
\end{align}
where $\comparisonmarker{1} = \comparisonmarker[n]{1}$ is the
frequency specific toll function\,---\,namely the corresponding frequency during the first
partitioning step only.
For the other frequencies, we similarly denote by $\toll A$, $\toll F$,
$\comparisonmarker{3}$, $\comparisonmarker{4}$, $\swapmarker{1}$
and~$\swapmarker{3}$ the toll functions corresponding to $A$, $F$,
$\totalcomparisonmarker{3}$, $\totalcomparisonmarker{4}$, $\totalswapmarker{1}$
and~$\totalswapmarker{3}$, respectively.

The frequencies $F$, $\totalcomparisonmarker{1}$,
$\totalcomparisonmarker{3}$, $\totalcomparisonmarker{4}$, $\totalswapmarker{1}$
and~$\totalswapmarker{3}$ correspond to basic blocks in the body of the main
partitioning loop, i.\,e., blocks~\ref{bb:8}\,--\,\ref{bb:18}.
All these blocks have in common that they are \emph{not executed at all}
during calls with $\id{right}-\id{left} \le 1$, i.\,e., when $n \le 2$:
In that case we have $k>g$ directly after
block~\ref{bb:init-k-l-g} and hence immediately leave the partitioning loop
from block~\ref{bb:k-le-g} to block~\ref{bb:recursive-calls}.
Therefore, we have $\toll[2]F = \comparisonmarker[2]{1} =
\comparisonmarker[2]{3} = \comparisonmarker[2]{4} = \swapmarker[2]{1} =
\swapmarker[2]{3}=0$. In the subsequent sections, we will determine the toll
functions for~$n\ge3$.

For that, we will first compute their values given \emph{fixed} pivot ranks $P$
and~$Q$, i.\,e., we determine their distribution \emph{conditional} on 
$(P,Q) = (p,q)$. Here, we capitalized $P$ and $Q$ to emphasize the fact that the
pivot ranks are themselves random variables.
Then, we get the unconditional expected frequencies via the \textsl{law of total
expectation}.
Note that for permutations of $\{1,\ldots,n\}$, ranks and values coincide and we
will not always dwell on the difference to keep the presentation concise, but 
unless stated otherwise, $P$ and $Q$ refer to the ranks of the two pivots.

\smallskip
\subsubsection{The Crossing Point Lemma}

The following lemma is the key to the precise analysis of the execution
frequencies that depend on how pointers $k$ and $g$ ``cross''.
As the pointers are moved alternatingly towards each other, one of them will
reach the crossing point \emph{first}\,---\,waiting for the other to
arrive.

\begin{lemma}[Crossing Point Lemma]
\label{lem:k-stops-with-q+delta}
	Let \arrayA store a random permutation of $\{1,\dots,n\}$ with $n\ge2$.
	Then, Algorithm~\ref{alg:yaroslavskiy} leaves
	the outer loop of the first partitioning step with
	\vspace{-1ex}
	\begin{equation}
		k	\wwrel= q+\delta \wwrel= g+1+\delta\,,
			\qquad\text{where } \delta=0 \text{ or } \delta=1.
		\label{eq:crossing-point}
	\end{equation}
	(More precisely, \eqref{eq:crossing-point} holds for the valuations of $k$,
	$g$ and $q$ upon entrance of block~\ref{bb:recursive-calls}).\\[.5ex]
	Moreover, $\delta=1$ \emph{iff} initially $\arrayA[q]>q$ holds, where
	$q = \max\{\arrayA[1],\arrayA[n]\}$ is the large pivot.
\end{lemma}
\begin{proof}[of Lemma~\ref{lem:k-stops-with-q+delta}]
Between two consecutive “$k \le g$”-checks in block~\ref{bb:k-le-g},
we move $k$ and $g$ towards each other by at most one position each;
so we always have $k\le g+2$ and we exit the loop as soon as $k>g$ holds.
Therefore, we always leave the loop with $k=g+1+\delta$ for some
$\delta\in\{0,1\}$.
In the end, $q$ is moved to position~$g$ in block~\ref{bb:recursive-calls}.
Just above in the same block, $g$ has been incremented, so we have $g=q-1$ upon
entrance of block~\ref{bb:recursive-calls}.

For the ``moreover'' part, we show both implications separately.
Assume first that $\delta=1$, i.\,e., the loop is left with a difference
of $\delta+1=2$ between $k$ and $g$.
This difference can only show up when both $k$ is incremented \emph{and} $g$ is
decremented in the last iteration.
Hence, in this last iteration we must have gone from
block~\ref{bb:comp-2} to~\ref{bb:comp-3} and accordingly $\arrayA[k]\ge q$ must have
held there\,---\,and by Fact~\ref{fact:fresh-elements} $\arrayA[k]$ still holds its
initial value.

In case $k<n$, even strict inequality $\arrayA[k]>q$ holds since we then have
$\arrayA[k] \ne \arrayA[n] = q$ by the assumption of distinct elements.
Now assume towards a contradiction, $k=n$ holds in the last execution of
block~\ref{bb:comp-2}.
Since $g$ is initialized in block~\ref{bb:init-k-l-g} to
$\mathit{right}-1=n-1$ and is only decremented in the loop, we have $g\le n-1$.
But this is a contradiction to the loop condition “$k\le g$”:
$n = k\le g\le n-1$.
So, $\arrayA[k] > q$ holds for the last execution of block~\ref{bb:comp-2}.

By assumption, $\delta=1$, so $k=q+1$ upon termination of the loop.
As $k$ has been incremented exactly once since the last test in
block~\ref{bb:comp-2}, we find $\arrayA[q] > q$ there,
as claimed.

\smallskip
Now, assume conversely that initially $\arrayA[q] > q$ holds.
As $g$ stops at $q-1$ and is decremented in block~\ref{bb:outer-g--},
we have $g=q$ for the last execution of block~\ref*{bb:comp-3}.
Using the assumption yields
$\arrayA[g] = \arrayA[q] > q$, since by Fact~\ref{fact:fresh-elements},
$\arrayA[q]$ still holds its initial value.
Thus, we take the transition to block~\ref{bb:inner-k-less-g}.
Execution then proceeds with block~\ref{bb:comp-4}, otherwise we would enter
block~\ref{bb:comp-3} again, contradicting the assumption that we just finished
its \emph{last} execution.
The transition from block~\ref{bb:inner-k-less-g} to~\ref{bb:comp-4} is only
taken if $k \ge g = q$.
With the following decrement of $g$ and increment of $k$, we leave the loop with
$k\ge g+2$, so $\delta=1$.
\end{proof}

\smallskip
\begin{corollary}
	\label{cor:probability-for-delta-1}
	Let $\delta\in\{0,1\}$ be the random variable from
	Lemma~\ref{lem:k-stops-with-q+delta}.\\ 
	It holds $\E[\delta] = \tfrac13$ and 
	$\E[\delta \given (P,Q) = (p,q)] = \tfrac{n-q}{n-2}$.
\end{corollary}
\begin{proof}
We first compute the conditional expectation.
As $\delta\in\{0,1\}$, we have
$\E[\delta\given P,Q] = \Prob[\delta = 1\given P,Q]$, so it suffices to compute
this probability.
Now by Lemma~\ref{lem:k-stops-with-q+delta}, we have
$\Prob[\delta = 1\given P,Q] = \Prob[\arrayA[q] > q\given (P,Q)=(p,q)]$.
We do a case distinction.
\begin{itemize}
\item
	For $q<n$, $\arrayA[q]$ is one of the non-pivot elements.
	(We have $1\le p<q<n$.)
	Any of the $n-2$ non-pivot elements can take
	position $\arrayA[q]$, and among those, $n-q$ elements are strictly greater than $q$.
	This gives a probability of $\frac{n-q}{n-2}$ for $\arrayA[q] > q$.
\item
	For $q=n$, $q$ is the maximum of all elements in the list, so we
	cannot possibly have $\arrayA[q] > q$. This implies a probability
	of $0=\frac{n-q}{n-2}$.
\end{itemize}
By the \weakemph{law of total expectation}, the unconditional expectation is
given by:
\begin{align*}
	\E [\delta]
		&\wwrel=	\sumpq\!\!\!\
						\Prob[(P,Q)=(p,q)]\cdot
						\E[\delta \given (P,Q)=(p,q)]
		 \wwrel=	\probsumpq \tfrac{n-q}{n-2}
\\&\wwrel=	\frac{1}{\tbinom n2 (n-2) }\biggl(
						  		\sumpq\!\!\!\!\!\!	n\;
						\bin- 	\sumpq\!\!\!\!\!\! q\;
					\biggr)
		 \wwrel=			\frac{1}{\tbinom n2 (n-2) } n \tbinom{n}{2}
					\wbin-	\frac{\tfrac{2}{3}(n+1)}{n-2}
		 \wwrel =		\frac{n - \tfrac{2}{3}(n+1)}{n-2}
		 \wwrel=	\tfrac{1}{3}\;.
\end{align*}
\end{proof}

\smallskip\noindent
The following expectations are used several times below, so we collect them
here.

\begin{lemma}
	\label{lem:expectation-p-q}
	$\E[P]=\tfrac13 (n+1)$ and $\E[Q] = \tfrac23 (n+1)$.
\end{lemma}
\begin{proof}
Conditioning on $(P,Q) = (p,q)$, we find
\begin{align*}
	\E[Q]	&\wwrel= \sumpq \frac1{\tbinom n2} \cdot q
			 \wwrel= \frac1{\tbinom n2}\sum_{q=2}^n q \sum_{p=1}^{q-1} 1
			 \wwrel= \tfrac23(n+1) \;.
\end{align*}
A similar calculation for $P$ proves the lemma.
\end{proof}

\subsubsection{Frequency $A$}

The frequency $A=A_n$ equals the number of partitioning steps or
equivalently the number of (recursive) calls with
$\id{right}-\id{left} \ge M$ when initially calling
$\proc{QuicksortYaroslavskiy}(\arrayA,1,n)$ with a random permutation
stored in~\arrayA.
Therefore, the contribution $\toll A$ of \emph{one} partitioning step is
$\toll[n]A=1$.
By Proposition~\ref{pro:Cn-explicit-linear-pc} with $\pc_n = 1$ and
$\insertsortcost_n=0$, we obtain the closed form
\begin{align}
		\E[A_n] 
	&\wwrel= 
		\tfrac{6}{5(M+2)}(n+1) \bin- \tfrac12
		\wwbin+ \tfrac3{10} \tbinom{M+1}4 \bigl/ \tbinom n4
		\;.
\end{align}

\subsubsection{Frequency $R$}

By $R=R_n$, we denote the number of calls to
\proc{QuicksortYaroslavskiy} including those directly passing control to
\proc{InsertionSort} for small subproblems.
Every partitioning step entails three
additional recursive calls on subarrays (see block~\ref{bb:recursive-calls}).
Moreover, we have one additional initial call to the procedure.
Together, this implies
\begin{align}
	R_n &\wwrel= 3\.A_n + 1\;.
\end{align}

\subsubsection{Frequency $B$}

Frequency $B$ counts how often we execute block~\ref{bb:p-right-q-left}.
This block is reached at most once per partitioning step, namely \weakemph{iff}
$\arrayA[\id{left}] > \arrayA[\id{right}]$.
For random permutations, the probability for that is exactly $1\mathbin/2$, so
we find
\begin{align}
	\E[B_n]  &\wwrel=  \tfrac12 \E[A_n]\;.
\end{align}

\subsubsection[Frequency $C^1$]{Frequency $\totalcomparisonmarker1$}

$\totalcomparisonmarker1(n)$ denotes the execution frequency of
block~\ref{bb:comp-1} of Yaroslavskiy's algorithm.
Block~\ref{bb:comp-1} is the first statement in the outer loop and the last
block of this loop (block~\ref{bb:k++}) is the only place where $k$ is
incremented.
Therefore, $\comparisonmarker{1}$ is the number of different values
that $k$ attains during the first partitioning step.
The following corollary quantifies this number as 
$\comparisonmarker{1} = Q-2+\delta$.

\begin{corollary}
\label{cor:positions-K-and-G}
	Let us denote by $\mathcal K$ the set of values that pointer $k$ attains at
	block~\ref{bb:comp-1}.
	Similarly, let $\mathcal G$ be the set of values of $g$ in
	block~\ref{bb:comp-3}.
	We have
	\begin{align*}
		\like{\mathcal G}{\mathcal K}
			&\wwrel= \bigl\{ 2,3,\ldots,Q-1+\delta \bigr\} \,,
			& \like{|\mathcal G|}{|\mathcal K|}
			&\wwrel= Q-2 + \delta\,, \\
		\mathcal G &\wwrel= \bigl\{ n-1,n-2,\ldots,Q+1,Q \bigr\} \,,
			& |\mathcal G| &\wwrel= n-Q\;.
	\end{align*}
\end{corollary}
\begin{proof}
By Lemma~\ref{lem:k-stops-with-q+delta}, we leave the outer loop with
$k=Q+\delta$ and $g=Q-1$.
Since the last execution of block~\ref{bb:comp-1}, $k$ has been incremented
exactly once (in block~\ref{bb:k++}), so the last value of $k$, namely
$Q+\delta$, is not observed at block~\ref{bb:comp-1}.
Similarly, after the last execution of block~\ref{bb:comp-3}, we always pass
block~\ref{bb:outer-g--} where $g$ is decremented. So the last value $Q-1$ for
$g$ is not attained in block~\ref{bb:comp-3}.
\end{proof}

Continuing with frequency $\totalcomparisonmarker{1}$, note that $Q$ and
$\delta$ and hence $\comparisonmarker{1} = Q-2+\delta$ are random
variables.
By linearity of the expectation
$\E[\comparisonmarker{1}] = \E[Q]-2+\E[\delta]$ holds, so with
Lemma~\ref{lem:expectation-p-q} and Corollary~\ref{cor:probability-for-delta-1},
we find
\begin{align}
		\E[\comparisonmarker[n]{1}]
	&\wwrel=
		\tfrac13 (n+1)  \wbin-  2  \wbin+  \tfrac13
	 \wwrel= \tfrac23 n - 1
	 \;.	
\end{align}

\subsubsection[Frequency $S^1$]{Frequency $\totalswapmarker{1}$}

Frequency $\totalswapmarker{1}$ corresponds to block~\ref{bb:swap-1}.
Block~\ref{bb:swap-1} is executed as often as block~\ref{bb:comp-1} is
reached with $\arrayA[k]<p$.
This number depends on the input permutation:
$\swapmarker[n]{1}$ is exactly the number of elements
smaller than $p$ that happen to be located at positions in $\mathcal K$\!,
the range that pointer $k$ scans.
Denote this quantity by $\satK$\!.

\begin{lemma}
\label{lem:s-at-K}
	Conditional on the pivot ranks
	$P$ and $Q$, $\satK$\! is hypergeometrically
	$\hypergeometric(P-1,Q-2,n-2)$ distributed.
\end{lemma}

\begin{proof}
This is seen by considering the following (imaginary) generation process of the
current input permutation:
Assuming fixed pivots $(P,Q) = (p,q)$, we have to generate a random
permutation of the remaining $n-2$ elements $E \ce \{1,\ldots,n\} \setminus \{p,q\}$.
To do so, we first choose a random subset $S$ of the free positions
$F \ce \{2,\ldots,n-1\}$ with $|S|={p-1}$.
Then we put a random permutation of $\{1,\ldots,p-1\}$ into positions $S$ and a
random permutation of $E \setminus \{1,\ldots,p-1\}$ into positions
$F \setminus S$.
It is easily checked that this generates all permutations of $E$ with equal
probability, if all choices are done uniformly.

Then by definition, $\satK = |S \cap \mathcal K|$.
This seemingly innocent equation hides a subtle intricacy not to be
overlooked:
$\mathcal K = \{2,\ldots,q-1+\delta\}$ (Corollary~\ref{cor:positions-K-and-G})
is itself a random variable which depends on the permutation via $\delta$.
Luckily, the characterization of $\delta$
from Lemma~\ref{lem:k-stops-with-q+delta} allows to resolve this
inter-dependence.
$\mathcal K = \{2,\ldots,q\}$ if $\arrayA[q] > q$ and
$\mathcal K = \{2,\ldots,q-1\}$ otherwise.
Stated differently, we get the additional position $q$ in $\mathcal K$
\weakemph{iff} the element at that position is large, which means position $q$
\emph{never} contributes towards \emph{small} elements at positions in $\mathcal
K$\!. 
As a result, $\satK = \numberat s{\!K'} = |S \cap \mathcal K'|$ for
$\mathcal K' = \{2,\ldots,q-1\}$, which is constant for fixed pivot values
$p$ and~$q$.

Drawing positions $S$ for small elements one by one is then equivalent
to choosing $|S|$ balls out of an urn with $n-2$ balls without replacement.
If $|\mathcal K'|$ of the $n-2$ balls are red, then $\satK$ equals the
number of red balls drawn, which is hypergeometrically
$$
	\hypergeometric(|S|,|\mathcal K'|,n-2)
	\wwrel=
	\hypergeometric(p-1,q-2,n-2)
$$
distributed by definition.
\end{proof}

\noindent
The mean of hypergeometric distributions
from~\eqref{eq:mean-variance-hypergeometric} translates into the
\emph{conditional} expectation
$\E[\satK \given P,Q] = (P-1)(Q-2)\big/(n-2)$.
By the \weakemph{law of total expectation}, we can compute the unconditional
expected value:
\begin{align}
\label{eq:exp-of-swapmarker1}
		\E[\swapmarker[n]{1}]
	&\wwrel=	\E[\satK]
	 \wwrel=	\E_{(P,Q)}\bigl[ \E[\satK\given P,Q] \bigr]
	 \wwrel=	\probsumpq \!\!\!\!\! \tfrac{(p-1)(q-2)}{n-2}
	 \wwrel=	\tfrac14 n - \tfrac5{12}\;.
\end{align}

\subsubsection[Frequency $C^3$]{Frequency $\totalcomparisonmarker3$}
\label{sec:freq-C3}

Block~\ref{bb:comp-3}%
\,---\,whose executions are counted in $\totalcomparisonmarker{3}$\,---\,%
compares $\arrayA[g]$ to $q$.
After every execution of block~\ref{bb:comp-3}, pointer $g$ is decremented:
depending on whether we leave the loop or not, either in
block~\ref{bb:inner-g--} or in block~\ref{bb:outer-g--}.
Therefore, we execute block~\ref{bb:comp-3} for every value that $g$ attains
at block~\ref{bb:comp-3}, which by Corollary~\ref{cor:positions-K-and-G} amounts
to
$\comparisonmarker[n]{3} = |\mathcal G| = n-Q$.
Using Lemma~\ref{lem:expectation-p-q}, we find
\begin{align}
	\E[\comparisonmarker[n]{3}] \wwrel = \tfrac13 n -\tfrac23 \;.
\end{align}

\subsubsection{Frequency $F$}
\label{sec:freq-F}
Frequency $F$ counts how often we take the transition from
block~\ref{bb:inner-k-less-g} to block~\ref{bb:comp-4}.
This transition is taken when we exit the inner loop of
Yaroslavskiy's algorithm because the second part  of its loop
condition, “$k<g$”, is violated, which means we had $k\ge g$.

After this has happened, we always execute blocks~\ref{bb:outer-g--}
and~\ref{bb:k++}, where we decrement $g$ \emph{and} increment $k$.
Moreover by Lemma~\ref{lem:k-stops-with-q+delta}, $k$ is at most $g+2$
\emph{after} the loop, and equality holds \weakemph{iff} $\delta=1$.
So at block~\ref{bb:inner-k-less-g}, we always have $k \le g$, which means the
violation of the loop condition occurs for $k=g$ and can only happen in case
$\delta=1$.

We can also show that it \emph{must} happen whenever $\delta=1$:
By Lemma~\ref{lem:k-stops-with-q+delta}, we have $\arrayA[q]>q$, and $k=q+1=g+2$
\emph{after} the loop.
Therefore, \emph{during} the last iteration of the loop, $g=k=q$ and hence
$\arrayA[k]=\arrayA[g]=\arrayA[q]>q$ holds.
As a consequence, execution always proceeds through blocks~\ref{bb:comp-1},
\ref{bb:comp-2} and~\ref{bb:comp-3} to block~\ref{bb:inner-k-less-g}.
There, “$k<g$” is not fulfilled, so we take the transition to
block~\ref{bb:comp-4}.
Together, we obtain $\toll F = \delta$ and
Corollary~\ref{cor:probability-for-delta-1} gives
\begin{align}
	\E[\toll[n]F] &\wwrel= \tfrac13 \;.
\end{align}

\subsubsection[Frequency $C^4$]{Frequency $\totalcomparisonmarker{4}$}

Frequency $\totalcomparisonmarker{4}$ corresponds to block~\ref{bb:comp-4},
which compares $\arrayA[g]$ to~$p$.
From the control flow graph, it is obvious that
$\totalcomparisonmarker{4}$ is the sum of the frequencies of the two
incoming transitions, namely block~\ref{bb:comp-3} to~\ref{bb:comp-4} and
block~\ref{bb:inner-k-less-g} to~\ref{bb:comp-4}.
The latter is exactly~$F$.

For the former, recall from above that block~\ref{bb:comp-3}
is executed once for all values $\mathcal G = \{n-1,n-2,\ldots,Q\}$ that pointer
$g$ attains there.
The transition from block~\ref{bb:comp-3} to block~\ref{bb:comp-4} is taken
\weakemph{iff}~\mbox{$\arrayA[g]\le q$}.
As $1 < g < n$ holds and all elements are distinct, $\arrayA[g]=p$ cannot occur.
Therefore, exactly the small and medium elements that are located at positions
in $\mathcal G$ cause this transition;
denote their number by $\numberat{sm}G$.
A very similar argument as in the proof
of Lemma~\ref{lem:s-at-K} shows that conditional on $P$ and $Q$,
$\numberat{sm}G$ is hypergeometrically $\hypergeometric(Q-2,n-Q,n-2)$
distributed.

Adding both contributions yields
$\comparisonmarker{4} = \delta + (\numberat{sm}G)$.
Using Corollary~\ref{cor:probability-for-delta-1} and
equation~\eqref{eq:mean-variance-hypergeometric} shows
\begin{align}
		\E[\comparisonmarker[n]{4}]
	&\wwrel= 	\tfrac13 \bin+ \probsumpq \!\!\!\!\! \tfrac{(q-2)(n-q)}{n-2}
	 \wwrel=	\tfrac16n-\tfrac16\;.
\end{align}

\subsubsection[Frequency $S^3$]{Frequency $\totalswapmarker{3}$}

Frequency $\totalswapmarker{3}$ counts executions of block~\ref{bb:swap-3}.
Key to its analysis are the following two observations:
\begin{longenum}
\item Block~\ref{bb:swap-3} and block~\ref{bb:swap-1} (with frequency
	$\totalswapmarker{1}$) are the only locations inside the loop where pointer
	$\ell$ is changed.
	Therefore, $\swapmarker{1} + \swapmarker{3} = |\mathcal L|-1$, where $\mathcal
	L$ is the set of values pointer that $\ell$ attains inside the loop (minus one
	as we leave the loop after the last increment of $\ell$ without executing
	blocks~\ref{bb:swap-1} and~\ref{bb:swap-3} again).
\item In block~\ref{bb:recursive-calls}, we move the small pivot to $\arrayA[\ell]$,
	so $\ell=P$ must hold there. Just above the swap, $\ell$ is decremented, so the
	last value of $\ell$ in the loop has been $P+1$.
	Moreover, $\ell$ is initialized to $2$ (block~\ref{bb:init-k-l-g}), so
	$\mathcal L = \{ 2,\ldots,P+1 \}$.
\end{longenum}
Together, this implies
$\swapmarker{3} = P-1 - (\satK)$ and by~\eqref{eq:exp-of-swapmarker1} and
Lemma~\ref{lem:expectation-p-q}:
\begin{align}
		\E[\swapmarker[n]{3}]
	&\wwrel=	\E[P] \bin- 1
				\bin- \E[\satK]
	 \wwrel= 	\tfrac13(n+1)\bin-1
	 			\bin- \bigl( \tfrac14 n -\tfrac5{12} \bigr)
	 \wwrel=	\tfrac1{12} n - \tfrac14
	 \;.
\end{align}

\begin{table}
	\tbl{
		Expected values for the toll functions for execution frequencies that
		characterize the block execution frequencies of all blocks in
		Algorithm~\ref{alg:yaroslavskiy} (top) and the derived toll functions for the
		expected number of comparison, swaps, write accesses and Bytecode instructions
		during the first partitioning step (bottom).
		\label{tab:expectation-of-toll-functions}
	}{
	\newcommand\anminusb[2]{{#1 n}-{#2}}
	\newcommand\anplusb[2]{{#1 n}+{#2}}
	\newcommand\zeronplusb[1]{#1}
	\begin{tabular}{lccccccc}
		\toprule
			Toll function &
			$\toll A$ &
			$\toll F$ &
			$\comparisonmarker1$ &
			$\comparisonmarker3$ &
			$\comparisonmarker4$ &
			$\swapmarker1$ &
			$\swapmarker3$ \\
		\midrule
			Expected value ($n\ge3$)&
			$1$ &
			$\zeronplusb{\tfrac{1}{3}}$ &
			$\anminusb{\tfrac{2}{3}}1$ &
			$\anminusb{\tfrac{1}{3}}{\tfrac{2}{3}}$ &
			$\anminusb{\tfrac{1}{6}}{\frac{1}{6}}$ &
			$\anminusb{\tfrac{1}{4}}{\tfrac{5}{12}}$ &
			$\anminusb{\tfrac{1}{12}}{\tfrac{1}{4}}$ \\[4pt]
			Special value for $n=2$ &
			no &
			0 &
			0 &
			0 &
			0 &
			0 &
			0 \\			
		\bottomrule
		\\[1ex]
		\toprule
			Toll function &
			$\comparisonmarker{\textit{QS}}$ &
			$\swapmarker{\textit{QS}}$ &
			$\toll {\totalwritesmarker{\mathit{QS}}}$ &
			$\toll {\bytecodes{\mathit{QS}}}$
			\\
		\midrule
			Expected value ($n\ge3$)&
			$\anminusb{\tfrac{19}{12}}{\tfrac{17}{12}}$ &
			$\anplusb{\tfrac{1}{2}}{\frac{7}{6}}$ &
			$\anplusb{\tfrac{11}{12}}{\tfrac{31}{12}}$ &
			$\anplusb{\tfrac{217}{12}}{\tfrac{265}{4}}$ \\[4pt]
			Special value for $n=2$ &
			1 &
			2 &
			4 &
			$\tfrac{189}2$ \\			
		\bottomrule
	\end{tabular}
	}
\end{table}

\begin{table}
\tbl{%
	Expected execution frequencies characterizing all
	block execution frequencies of Figure~\ref{fig:control-flow-graph}.
	Those immediately follow from Proposition~\ref{pro:Cn-explicit-linear-pc} and
	the toll functions of Table~\ref{tab:expectation-of-toll-functions}.
	For $M=1$, we give exact expectations (valid for $n \ge 4)$, for $M\ge2$ we
	confine ourselves to (extremely precise) asymptotics.
	Note that exact values can be computed
	using equation~\eqref{eq:Cn-exact-solution-linear-pc} if needed.
	\label{tab:expected-frequencies}
}{
	\let\C\totalcomparisonmarker
	\let\S\totalswapmarker
	\def\annHnnplusbnnplusc#1#2#3{
		  \like{\tfrac1{10}}{#1}(n+1)\bigl(\harm{n+1}-\harm{M+2}\bigr)
		+ \bigl(\like{\tfrac{19}{300}-\tfrac{2}{5(M+2)}}{#2}\bigr)(n+1)
		+ \like{\tfrac56}{#3}
	}
	\def\anHnminusbnplusaHnplusc#1#2#3{
		  \like{\tfrac1{10}}{#1}(n+1)\harm n
		- \like{\tfrac{127}{200}}{#2}n
		+ \like{\tfrac{13}{600}}{#3}
	}
	\def\anHnminusbnplusaHnminusc#1#2#3{
		  \like{\tfrac1{10}}{#1}(n+1)\harm n
		- \like{\tfrac{127}{200}}{#2}n
		- \like{\tfrac{13}{600}}{#3}
	}
	\def\myrowsep{4pt}
	\begin{tabular}{ccc}
		\toprule
		 Frequency & $M=1$ (exact solution for $n \ge 4$) &
		   $M\ge2$ (asymptotic with error term $\bo\bigl(\tfrac{1}{n^{4}}\bigr)$)\\
		\midrule
		$\E[A]$ &
			$\tfrac{2}{5}n-\tfrac{1}{10}$ &
			$\tfrac{6}{5(M+2)}(n+1) - \tfrac12$ \\[\myrowsep]
		$\E[B]$ &
			$\tfrac{1}{5}n-\tfrac{1}{20}$ &
			$\tfrac{3}{5(M+2)}(n+1) - \tfrac14$ \\[\myrowsep]
		$\E[R]$ &
			$\tfrac{6}{5}n+\tfrac{7}{10}$ &
			$\tfrac{18}{5(M+2)}(n+1) - \tfrac12$ \\[\myrowsep]
		$\E[F]$ &
			$\tfrac{1}{10}n-\tfrac{1}{15}$ &
			$\tfrac{2}{5(M+2)}(n+1) - \tfrac16$ \\[\myrowsep]
		$\E[\C1]$ &
			$\anHnminusbnplusaHnminusc{\tfrac{4}{5}}{\tfrac{83}{50}}{\tfrac{2}{75}}$ &
			$\annHnnplusbnnplusc{\tfrac{4}{5}}{\tfrac{38}{75}-\tfrac{2}{M+2}}{\tfrac{5}{6}}$
			\\[\myrowsep]
		$\E[\C3]$ &
			$\anHnminusbnplusaHnplusc{\tfrac{2}{5}}{\tfrac{22}{25}}{\tfrac{1}{50}}$ &
			$\annHnnplusbnnplusc{\tfrac{2}{5}}{\tfrac{19}{75}-\tfrac{6}{5(M+2)}}{\tfrac{1}{2}}$
			\\[\myrowsep]
		$\E[\C4]$ &
			$\anHnminusbnplusaHnminusc{\tfrac{1}{5}}{\tfrac{39}{100}}{\tfrac{7}{300}}$ &
			$\annHnnplusbnnplusc{\tfrac{1}{5}}{\tfrac{19}{150}-\tfrac{2}{5(M+2)}}{\tfrac{1}{6}}$
			\\[\myrowsep]
		$\E[\S1]$ &
			$\anHnminusbnplusaHnminusc{\tfrac{3}{10}}{\tfrac{127}{200}}{\tfrac{1}{600}}$&
			$\annHnnplusbnnplusc{\tfrac{3}{10}}{\tfrac{19}{100}-\tfrac{4}{5(M+2)}}{\tfrac{1}{3}}$
			\\[\myrowsep]
		$\E[\S3]$ &
			$\anHnminusbnplusaHnplusc{\tfrac{1}{10}}{\tfrac{49}{200}}{\tfrac{13}{600}}$ &
			$\annHnnplusbnnplusc{\tfrac{1}{10}}{\tfrac{19}{300}-\tfrac{2}{5(M+2)}}{\tfrac{1}{6}}$
			\\[2pt]
		\midrule
		\pbox{\linewidth}{$\E[\pc_n]=an+b$,\\ $\E[\like{\pc_n}{\pc_2}]=0$} &
			\multicolumn{2}{c}{
				$\tfrac{6}{5}a(n+1)(\harm n-\harm{M+2})
				+ \bigl(
					\tfrac{19}{25}a
					-\tfrac{6}{5}\tfrac{a-b}{M+2}
					-\delta_{M1}\tfrac{2a+b}{10}
				  \bigr)(n+1)
				+ \tfrac{a-b}{2} \wbin+ \bo\bigl(\tfrac1{n^4}\bigr)$
			}\\
		\bottomrule
	\end{tabular}
}
\end{table}

\newcommand\annHnnplusbnnc[4][+]{
		  #2 (n+1)\bigl(\harm{n+1}-\harm{M+2}\bigr)
		\bin+ \bigl( #3 \bigr)(n+1)
		\bin{#1} #4
		\wbin+ \bo\bigl( \tfrac1{n^4} \bigr)\,,
}
\newcommand\annHnnplusbnnclinebreak[4][+]{
	\begin{aligned}
		&\bin{\phantom{+}} #2 (n+1)\bigl(\harm{n+1}-\harm{M+2}\bigr) \\*[1ex]
		&\bin+ \bigl( #3 \bigr)(n+1)
		\bin{#1} #4
		\wbin+ \bo\bigl( \tfrac1{n^4} \bigr)\,,
	\end{aligned}
}
\newcommand\annHnminusbnc[4][+]{
		  #2(n+1)\harm n
		\bin- #3n
		\bin{#1} #4\,,
}
\newcommand\annHnplusbnc[4][+]{
		  #2(n+1)\harm n
		\bin+ #3n
		\bin{#1} #4\,,
}

\subsection{Key Comparisons}
\label{sec:expected-comparisons}

\begin{theorem}
\label{thm:expected-comparisons}
	In expectation, Yaroslavskiy's algorithm (Algorithm~\ref{alg:yaroslavskiy})
	uses
	\begin{align}
		\E[\totalcomparisonmarker[n]{}]
		&\wwrel=
		\begin{cases}
			\annHnnplusbnnclinebreak{\tfrac{19}{10}}%
				{  \tfrac{124}{75}
				 + \tfrac{3}{20}M
				 - \tfrac{9}{5(M+2)}
				 - \tfrac{12}{5(M+2)} \harm{M+1}}%
				{\tfrac32}
					& \text{for } M \ge 2; \\[3ex]
			\annHnminusbnc[-]{\tfrac{19}{10}}%
				{\tfrac{711}{200}}{\tfrac{31}{200}}
					& \text{for } M=1,\; n\ge4,
		\end{cases}
		\label{eq:expected-comparisons-total}
	\end{align}
	key comparisons to sort a random permutation of size~$n$.
\end{theorem}

\begin{proof}
Key comparisons in the partitioning loop happen in basic
blocks~\ref{bb:comp-0}, \ref{bb:comp-1}, \ref{bb:comp-2}, \ref{bb:comp-3}
and~\ref{bb:comp-4}.
Together this amounts to
\begin{align*}
		\totalcomparisonmarker[n]{\mathit{QS}}
	&\wwrel=
				\totalcomparisonmarker[n]{1}
		\wbin+	(\totalcomparisonmarker[n]{1} - \totalswapmarker[n]{1})
		\wbin+ \totalcomparisonmarker[n]{3}
		\wbin+ \totalcomparisonmarker[n]{4}
		\wbin+	A_n 
		\quad\text{and, in expectation,}
\\
		\E[\totalcomparisonmarker[n]{\mathit{QS}}]
	&\wwrel=	
		\begin{cases}
			\annHnnplusbnnc{\tfrac{19}{10}}%
				{\tfrac{361}{300}-\tfrac{18}{5(M+2)}}{\tfrac32}
					& \text{for } M \ge 2; \\[1ex]
			\annHnminusbnc[-]{\tfrac{19}{10}}%
				{\tfrac{711}{200}}{\tfrac{31}{200}}
					& \text{for } M=1,\; n\ge4,
		\end{cases}
\end{align*}
comparisons,
where the second equation follows by summing the results
from Table~\ref{tab:expected-frequencies}.

For $M\ge 2$, we get additional comparisons from
\proc{InsertionSort}, see Appendix~\ref{app:insertionsort} for details:
\begin{align*}
		\E[\totalcomparisonmarker[n]{\mathit{IS}}]
	&\wwrel= \E[E_n] + \E[D_n]
	 \wwrel=	\bigl(
					\tfrac{3}{20}(M+3) + \tfrac{9}{5(M+2)}
					- \tfrac{12}{5(M+2)} \harm{M+1}
				\bigr) (n+1)\;.
\end{align*}
Summing both contributions yields~\eqref{eq:expected-comparisons-total}.
\end{proof}

\subsection{Swaps \& Write Accesses}
\label{sec:expected-swaps}

\begin{theorem}
\label{thm:expected-swaps}
	In expectation, Yaroslavskiy's algorithm
	performs
	\begin{align}
			\E[\totalswapmarker[n]{}] 
		&\wwrel=
		\begin{cases}
			\annHnnplusbnnc[-]{\tfrac35}%
				{\tfrac{19}{50}+\tfrac{4}{5(M+2)}}{\tfrac13}
						& \text{for } M \ge 2; \\[1ex]
			\annHnminusbnc[-]{\tfrac35}%
				{\tfrac{47}{100}}{\tfrac{61}{300}}
						& \text{for } M=1,\; n\ge4,
		\end{cases}
		\label{eq:expected-swaps-partitioning}
	\end{align}
	\emph{swaps} in partitioning steps while sorting a random permutation of
	size~$n$.
	
	Including the ones done in \proc{InsertionSort} on small subproblems,
	Yaroslavskiy's algorithm uses
	\begin{align}
		\E[\totalwritesmarker[n]{}] &\wwrel=
		\begin{cases}
			\annHnnplusbnnclinebreak[-]{\tfrac{11}{10}}%
				{\tfrac{86}{75} + \tfrac3{20} M + \tfrac{18}{5(M+1)}
						- \tfrac{26}{5(M+2)}}{\tfrac56}
					& \text{for } M \ge 2; \\[3ex]
			\annHnminusbnc[-]{\tfrac{11}{10}}%
				{\tfrac{139}{200}}{\tfrac{257}{600}}
					& \text{for } M=1,\; n\ge4,
		\end{cases}
		\label{eq:expected-writes-total}
	\end{align}
	\emph{write accesses} to the array to sort a random permutation.
\end{theorem}
\begin{proof}
We find swaps in the partitioning loop of Yaroslavskiy's algorithm in basic
blocks~\ref{bb:swap-1}, \ref{bb:swap-2}, \ref{bb:swap-3}
and~\ref{bb:recursive-calls},
where blocks~\ref{bb:swap-3} and~\ref{bb:recursive-calls} each contain two
swaps.%
\footnote{%
	Note that \citeN{Wild2012} included an additional contribution of
	$B$ for swapping the pivots if they are out of order.
	However, this swap can be done with local variables only, we do not need to
	write the swapped values back to the array.
	Therefore, this swap is \weakemph{not} counted in this paper.	
}
Hence, the total number of swaps during all partitioning
steps is given by
\begin{align*}
		\totalswapmarker{\mathit{QS}}
	&\wwrel=
				\totalswapmarker{1}
		\wbin+	(\totalcomparisonmarker{4}-\totalswapmarker{3})
		\wbin+	2\totalswapmarker{3}
		\wbin+	2 A
	\;.
\end{align*}
Now, \eqref{eq:expected-swaps-partitioning} follows by inserting the terms from
Table~\ref{tab:expected-frequencies}.

A clever implementation realizes the two consecutive swaps in
block~\ref{bb:swap-3} with only three write operations, see for example
Appendix~\ref{app:low-level}.
This yields an overall number of
\begin{align*}
		\totalwritesmarker[n]{\mathit{QS}}
	&\wwrel=
				2\totalswapmarker[n]{1}
		\wbin+	2(\totalcomparisonmarker[n]{4}-\totalswapmarker[n]{3})
		\wbin+	3\totalswapmarker[n]{3}
		\wbin+	4 A_n
		\quad\text{and, in expectation,}
	\\
	\E[\totalwritesmarker[n]{\mathit{QS}}]
	&\wwrel=
	\begin{cases}
		\annHnnplusbnnc[-]{\tfrac{11}{10}}%
			{\tfrac{209}{300}+\tfrac{2}{M+2}}{\tfrac56}
				& \text{for } M \ge 2; \\[1ex]
		\annHnminusbnc[-]{\tfrac{11}{10}}%
			{\tfrac{139}{200}}{\tfrac{257}{600}}
				& \text{for } M=1,\; n\ge4,
	\end{cases}
\end{align*}
write operations during the partitioning steps.
The contribution from \proc{InsertionSort} is (cf.\
Appendix~\ref{app:insertionsort}):
\begin{align*}
		\E[\totalwritesmarker[n]{\mathit{IS}}]
	&\wwrel= \E[E_n] \bin+ (\E[G_n] - \E[I_n])
	 \wwrel=
		\bigl(
			\tfrac{3}{20}(M+3) + \tfrac{18}{5(M+1)} - \tfrac{36}{5(M+2)}
		\bigr) (n+1)\;.
\end{align*}
Adding both together we obtain~\eqref{eq:expected-writes-total}.
\end{proof}

\subsection{Executed Java Bytecode Instructions}
\label{sec:expected-instruction-count}

\begin{theorem}
\label{thm:expected-bytecodes}
	In expectation, the Java implementation of Yaroslavskiy's algorithm given in
	Appendix~\ref{app:low-level} executes
	\begin{align}
	\label{eq:expected-bytecodes-total}
		\E[\bytecodes[n]{}]
		&\wwrel=	\begin{cases}
					\annHnnplusbnnclinebreak[-]{\tfrac{217}{10}}%
						{\tfrac{4259}{150} + \tfrac{51}{20}M
							+ \tfrac{72}{M+1}
							- \tfrac{317}{5(M+2)}
							- \tfrac{48}{5(M+2)}\harm{M+1}}%
						{\tfrac{181}{12}}
							& M \ge 2; \\[4ex]
					\annHnminusbnc[-]{\tfrac{217}{10}}%
						{\tfrac{1993}{200}}{\tfrac{2009}{600}}
							& \mathllap{n\ge4,\;} M = 1,
				\end{cases}
	\end{align}
	Java Bytecode instructions to sort a random permutation of size $n$.
\end{theorem}
\begin{proof}
By counting the number of Bytecode instructions in each basic block
and multiplying it with this block's frequency,
we obtain:
\begin{align*}
		\bytecodes[n]{}
	&\wwrel=	 71 A
				- 1 B
				+ 6 R
				+ 15 \totalcomparisonmarker{1}
				+ 10 \totalcomparisonmarker{3}
				+ 11 \totalcomparisonmarker{4}
				+  9 \totalswapmarker{1}
				+  8 \totalswapmarker{3}
				+  3 F
				+  4 D
				+ 17 E
				+ 20 G
				-  7 I
	\;.
\end{align*}
For details, see Appendix~\ref{app:low-level}.
Inserting the expectations from Table~\ref{tab:expected-frequencies}
results in~\eqref{eq:expected-bytecodes-total}.

The Bytecode count for $M=1$ corresponds to entirely removing
\proc{InsertionSort} from the code.
\proc{InsertionSort} would execute 13 Bytecodes even on empty and one-element
lists until it finds out that the list is already sorted.
These obviously superfluous instructions are removed for $M=1$.
\end{proof}

\needspace{5cm}
\subsection{Optimal Choice for $M$}
\label{sec:optimal-M}

\begin{figure}
	\def\plotdatafile{plots/bytecodes-small-n-normalized-by-n.tab}
	\mbox{}\hfill%
	\tikzset{external/export=true}\tikzsetnextfilename{bytecodes-insertionsort-small-n}
	\begin{tikzpicture}
		\begin{axis}[
		    xlabel={$n$},
		    ylabel={$\#\mathrm{Bytecodes} \mathbin/ n$},
		    legend style={ at={(0.05,0.95)}, anchor=north west },
		    mark size=1.5pt,
 		    cycle list name=bw,
		   ]
			\addplot plot table[x=n, y=Insertionsort] {\plotdatafile};
			\addlegendentry{Insertionsort}
			\label{plot:bytecodes-small-n-IS}
			\addplot plot table[x=n, y=Yaroslavskiy-noIS] {\plotdatafile};
			\addlegendentry{Yaroslavskiy $M=1$}
			\label{plot:bytecodes-small-n-M1}
			\addplot plot table[x=n, y={Yaroslavskiy-M=7}] {\plotdatafile};
			\addlegendentry{Yaroslavskiy $M=7$}
			\label{plot:bytecodes-small-n-M7}
		\end{axis}
	\end{tikzpicture}%
	\hfill\mbox{}%
	\caption{%
		The expected number of executed Bytecodes for \proc{InsertionSort} and
		Yaroslavskiy's algorithm with different choices for $M$. The numbers of
		Bytecodes shown are normalized by $n$, i.\,e., we show the number
		of executed Bytecode instructions per element to be sorted.
		The data was obtained by naïvely evaluating the recurrence.
		\label{fig:bytecodes-small-n}
	}
\end{figure}
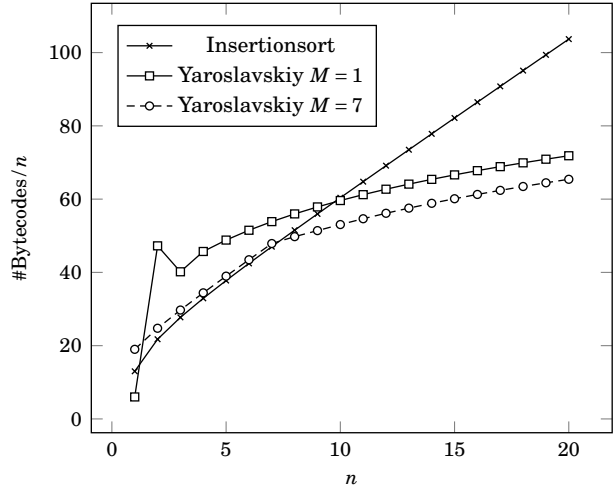

The overall linear term for the expected number of comparisons is
$$
	\bigl(
		  \tfrac{124}{75} + \tfrac{3}{20} M - \tfrac{37}{10(M+2)}
		  - \tfrac{62+19M}{10(M+2)} \harm{M+1}
	\bigr) (n+1)
	\;.
$$
This coefficient has a proper minimum at $M=5$ with value $-3.62024\ldots$\,.
Compared with $-711 \mathbin/ 200 = -3.555$ for $M=1$ this a minor
improvement though.

For the number of write operations, the linear term is
$$
	\bigl(
		  \tfrac{86}{75} + \tfrac{3}{20} M
		  - \tfrac{26}{5(M+2)} + \tfrac{18}{5(M+1)}
		  - \tfrac{11}{10} \harm{M+2}
	\bigr) (n+1)
	\;.
$$
The minimum $-1.098\overline3$ of this coefficient is also located at $M=5$.
The corresponding coefficient for $M=1$ is $-0.695$. This improvement is
more satisfying than the one for comparisons.

The linear term for the number of executed Bytecodes of Yaroslavskiy's algorithm
with $M\ge 2$ attains its minimum $-16.0887\ldots$ at $M=7$. This is a
significant reduction over $-9.965\ldots$, the linear term without
\proc{InsertionSort}\kern.5pt ing.
Figure~\ref{fig:bytecodes-small-n} shows the resulting expected number of
Bytecodes for small lists.
For $n \le 20$, using \proc{InsertionSort} results in an improvement of over $10\,\%$.
For $n=100$ we save $6.3\,\%$, for $n=1000$ it is $4.2\,\%$ and for $n=10\,000$,
we still execute $3.1\,\%$ less Bytecode instructions than the basic version of
Yaroslavskiy's algorithm.

It is interesting to see that both elementary operations favor $M=5$, but
the overall Bytecode count is minimized for ``much'' larger $M=7$.
This shows that focusing on elementary operations can skew the view of an
algorithm's performance.
Only explicitly taking the overhead of partitioning into account reveals that
\proc{InsertionSort} is significantly faster on small subproblems.

\paragraph{Remark}
The actual Java 7 runtime library implementation uses $M=46$, which seems
far from optimal at first sight.
Note however that the implementation uses the more elaborate
pivot selection scheme \weakemph{tertiles of five}
\cite{Wild2013Alenex}, which implies additional constant overhead per
partitioning step.

\needspace{5cm}
\section{Distribution of Costs}
\label{sec:distributional-analysis}

In this section we study the asymptotic \emph{distributions} of our cost
measures.
We derive limit laws after normalization and identify the order of
variances and covariances.
In particular, we find that all costs are asymptotically concentrated around
their mean.

As we confine ourselves to asymptotic statements of first order (leading terms
in the expansions of variances and covariances), it turns out that the choice
of $M$ does not affect the results of this section:
All results hold for any (constant) $M$ (see \cite[proof of
Corollary\,5.5]{neininger2001multivariate} for similar universal behavior of
standard Quicksort).
Appendix~\ref{app:experimental-validation} shows that the asymptotic results are
good approximations for practical input sizes~$n$.
A general survey on distributional analysis of various sorting
algorithms covering many classical results is found in
\cite{mahmoud2000sorting}.

\subsection{The Contraction Method}
\label{sss:contraction-method}

Our tool to identify asymptotic variances, correlations and limit laws is the
\textsl{contraction method}, which is applicable to many divide-and-conquer
algorithms.
Roughly speaking, the idea is to appropriately normalize a recurrence equation
for the \emph{distribution} of costs such that we can hope for convergence
to a limit distribution.
If we then replace all terms that depend on $n$ by their limits for
$n\to\infty$, we obtain a \emph{map} within the space of probability
distributions that approximates the recurrence.

Next a (complete) metric between probability distributions is chosen such that 
this map becomes a contraction; 
then the \textsl{Banach fixed-point theorem} implies the existence of a unique
fixed point for this map.
This fixed point is the candidate for the limit distribution of the normalized
costs and the underlying contraction property is then exploited to also show
convergence of the normalized costs towards the fixed point. 
This convergence is shown within the same complete metric. 
If the metric is sufficiently strong, it may imply more than convergence in law; 
in our case we additionally obtain convergence of the first two moments,
i.\,e., convergence of mean and variance.
This enables us to compute asymptotics for the variance of the cost as well. 
Note that a fixed-point representation for a limit distribution is implicit, but
it is suitable to compute moments of the limit distribution and to identify
further properties such as the existence of a (Lebesgue) density.

\smallskip
For the reader's convenience we formulate a general
convergence theorem from the contraction method that is used repeatedly below
and sufficient for our purpose.
Let $(X_n)_{n\ge 0}$ denote a sequence of centered and square
integrable random variables either in $\mathbb{R}$ or $\mathbb{R}^2$ whose distributions satisfy the recurrence
\begin{align}
\label{eq:general-contraction-recurrence}
	X_n &\wwrel\eqdist
	\sum_{r=1}^K \ui{A}{n}_r \ui{X}{r}_{I^{(n)}_r} + \ui{b}{n},
	\qquad n\ge n_0,
\end{align}
where the random variables
$(\ui{A}{n}_1,\ldots,\ui{A}{n}_K,\ui{b}{n},\ui{I}{n})$ and
$(\ui{X}{1}_n)_{n\ge 0}, \ldots,(\ui{X}{K}_n)_{n\ge 0}$ are independent, and
$\ui{X}{r}_{i}$ is distributed as $X_i$ for all $r=1,\ldots,K$ and $i\ge 0$.
Furthermore, $\ui{I}{n}=(\ui{I}{n}_1,\ldots,\ui{I}{n}_K)$ is a vector of random
integers in $\{0,\ldots,n-1\}$ and $K$ and $n_0$ are fixed integers.

The coefficients $\ui{A}{n}_r$ and $\ui{b}{n}$ are real random variables in the
univariate case, respectively random $2\times 2$ matrices and a $2$-dimensional
random vector in the bivariate case.
We assume also that the coefficients are square integrable and that the
following conditions hold:
\begin{enumerate}[itemsep=1ex,leftmargin=3em,label=(\Alph*)]
\item \label{cond:convergence-of-recurrence-coeffs}
	$(\ui{A}{n}_1,\ldots,\ui{A}{n}_K,\ui{b}{n})
	\wrel{\stackrel{\ell_2}{\longrightarrow}}
	(A_1,\ldots,A_k,b)$,	
\item \label{cond:recurrence-contraction}
	$\sum_{r=1}^K \E\left[ \| A_r^tA_r\|_\mathrm{op}\right] \wrel< 1$,
\item \label{cond:recurrence-small-sublists-seldom}
	$\sum_{r=1}^K \E\bigl[
		{\bf 1}_{\{\ui{I}{n}_r \le\ell\}}
		\|(\ui{A}{n}_r)^t\ui{A}{n}_r\|_\mathrm{op}
	\bigr] \wwrel\to 0 \mbox{ as } n\to\infty \mbox{ for all constants } \ell \ge
	0$.
\end{enumerate}
Here \smash{$\|A\|_\mathrm{op} \ce \sup_{\|x\|=1} \|Ax\|$} denotes the
\textsl{operator norm} of a matrix and $A^t$ the transposed matrix.
Note that in the univariate case we just have $\| A_r^tA_r\|_\mathrm{op}=A_r^2$.
In \ref{cond:convergence-of-recurrence-coeffs} we denote by 
\smash{$\overset{\mbox{\smash{$\scriptscriptstyle\ell_2$}}}{\longrightarrow}$} 
convergence in the Wasserstein-metric of order 2 which here is equivalent to the existence of
vectors $(\ui{\tilde{A}}n_1,\ldots,\ui{\tilde{A}}n_K,\ui{\tilde{b}}n)$
with the distribution of $(\ui{A}{n}_1,\ldots,\ui{A}{n}_K,\ui{b}{n})$ such that
we have the $L_2$ convergence
\begin{align*}
	(\ui{\tilde{A}}n_1,\ldots,\ui{\tilde{A}}n_K,\ui{\tilde{b}}n)
	&\wwrel{\stackrel{L_2}{\longrightarrow}}
	(A_1,\ldots,A_k,b) \;.
\end{align*}
Note in particular that $A_1,\ldots,A_k,b$ are square integrable as well.
Then we consider distributions of $X$ such that
\begin{align}
\label{eq:general-fix-point-eq}
	X &\wwrel\eqdist \sum_{r=1}^K A_r \ui{X}{r} + b\,,
\end{align}
where $(A_1,\ldots,A_K,b)$, $\ui{X}{1},\ldots,\ui{X}{K}$ are independent and
$\ui{X}{r}$ are distributed as $X$ for $r=1,\ldots, K$.
The following two results from the contraction method are used:
\begin{enumerate}[itemsep=.5ex,leftmargin=3em,label=(\Roman*)]
\item \label{res:unique-fixpoint}
	Under \ref{cond:recurrence-contraction}, among all centered, square
	integrable distributions there is a unique solution ${\cal L}(X)$ to 
	\eqref{eq:general-fix-point-eq}.
\item \label{res:convergence-of-coeffs-gives-fixpoint}
	Assuming \ref{cond:convergence-of-recurrence-coeffs},
	\ref{cond:recurrence-contraction}
	and~\ref{cond:recurrence-small-sublists-seldom},
	the sequence $(X_n)_{n\ge 0}$ converges in distribution to the
	solution~${\cal L}(X)$ from \ref{res:unique-fixpoint}.
	The convergence holds as well for the second (mixed) moments of~$X_n$.
\end{enumerate}
These results are given by \citeN[Theorem 3]{rosler2001analysis} for
the univariate case and \citeN[Theorem 4.1]{neininger2001multivariate}
for the multivariate case.

%
%

\subsection{Distributional Analysis of Yaroslavskiy's Algorithm}
\label{sec:distributional-analysis-yaroslavskiy}

We come back to Yaroslavskiy's algorithm (Algorithm~\ref{alg:yaroslavskiy}).
To apply the contraction method, we have to
characterize the full distribution of costs~$\pc_n$ in one partitioning step
and then formulate a \emph{distributional recurrence} for the resulting
distribution of costs~$C_n$ for the complete sorting. 
To obtain a contracting mapping, we rewrite the derived recurrence for~$C_n$ in
terms of suitably normalized costs $\tCsn$; 
here it will suffice to subtract the
expected values computed in the last section and then to divide by $n$.

For the distributional analysis, it proves more convenient to consider
i.\,i.\,d.\ uniformly on $[\.0,1]$ distributed random variables $U_1,\ldots,U_n$ as input.
Note that $U_1,\dots,U_n$ are thus pairwise different almost surely.
As the actual element values do not matter, this input model is the same as
considering a random permutation.

\smallskip
Yaroslavskiy's algorithm chooses $U_1$ and
$U_n$ as pivot elements.
Denote by $D=(D_1,D_2,D_3)$ the \emph{spacings} induced by $U_1$ and $U_n$ on
the interval $[\.0,1]$; formally we have
\begin{align*}
	(D_1,D_2,D_3)  &\wwrel=
	( U_{(1)} ,\, U_{(2)}-U_{(1)} ,\, 1-U_{(2)} ) \,,
\end{align*}
for $U_{(1)} \ce \min\{U_1,U_n\}$ and $U_{(2)} \ce \max\{U_1,U_n\}$.
It is well-known that $D$ is uniformly
distributed in the standard $2$-simplex \cite[p.\,133f]{David2003},
i.\,e., $(D_1,D_2)$ has density
\begin{align*}
	f_D(x_1,x_2) &\wwrel=
		\begin{cases}
			2, & \text{for } x_1,x_2 \ge 0 \rel\wedge x_1+x_2 \le 1; \\
			0, & \text{otherwise},
		\end{cases}
\end{align*}
and $D_3 = 1-D_1-D_2$ is fully determined by $(D_1,D_2)$.
Hence, for a measurable function $g:[\.0,1]^3
\to \Rset$ such that $g(D_1,D_2,D_3)$ is integrable, we have that
\begin{align}
\label{eq:Dirichlet-integral}
	\E[g(D_1,D_2,D_3)] &\wwrel=
		2 \int_0^1 \mkern-6mu \int_0^{1-x_1} \!\!\!\!\!	
			g(x_1,x_2,1-x_1-x_2)
		\; d x_2 \: d x_1 \;.
\end{align}

Further we denote the sizes of the three subproblems generated in the first
partitioning phase by $\ui{I}{n}=(\ui{I}{n}_1,\ui{I}{n}_2,\ui{I}{n}_3)$.
Then we have $\ui{I}{n}_1+\ui{I}{n}_2+\ui{I}{n}_3=n-2$.
Moreover, the spacings $D_1$, $D_2$ and $D_3$ are exactly the probabilities for
an element $U_i$ ($1<i<n$) to be small, medium or large, respectively.
As all these elements are independent, the vector~$\ui{I}{n}$, conditional on
$D$, has a multinomial distribution:
\begin{align*}
	\ui{I}{n} &\wwrel\eqdist \multinomial(n-2;\,D_1,D_2,D_3).
\end{align*}
We will use the short notation $\ui{I}{n}=I=(I_1,I_2,I_3)$ when the dependence
on $n$ is obvious.
From the strong law of large numbers and dominated convergence we have in
particular for $r\in\{1,2,3\}$
\begin{align}
\label{eq:asymptotic-I}
	\frac{\ui{I_r}n}{n}
		&\wwrel{\stackrel{L_p}{\longrightarrow}}  D_r
	\qquad (n\to\infty),
	\quad \; 1\le p<\infty
	\;.
\end{align}

The advantage of this random model is that we can \emph{decouple} values from
ranks of pivots.
With $D_1=U_{(1)}$ and $D_1 + D_2=U_{(2)}$, we choose the \emph{values} of the
two pivots; however, the ranks $P$ and $Q$ are not yet fixed.
Therefore given fixed pivot values, we can still \emph{independently} draw
non-pivot elements 
(with probabilities $D_1$, $D_2$ and $D_3$ to become small, medium and large,
resp.), without having to fuzz with a priori restrictions on the overall number
of small, medium and large elements.
This makes it much easier to compute cost contributions uniformly in pivot
values than in pivot ranks.	
If we operate on random permutations of $\{1,\ldots,n\}$, values and ranks
coincide, so fixing pivot values there implies strict bounds on the number of
small, medium and large elements.

\subsubsection{Distribution of Toll Functions}

\begin{table}
\tbl{%
	Exact distributions of the toll functions introduced in
	Section~\ref{sec:expected-frequencies} or equivalently the distributions
	of block execution frequencies in the first partitioning step.
	\label{tab:distribution-of-tolls}
}{%
	\let\C\comparisonmarker
	\let\S\swapmarker
	\def\rowsep{0pt}
	\def\vv{$\vphantom{\Big|}$}
	\def\ed{\eqdist}
	\def\pd{\vphantom{\ed}}
	\begin{tabular}{ >{\(}c<{\)} >{\(}c<{\)} >{\(}l<{\)} >{\(}c<{\)} >{\(}l<{\)} }
		\toprule
		\multicolumn3{c}{Quantity}	&	\multicolumn2{c}{Distribution given $I=(I_1,I_2,I_3)$}	        \\
		\midrule
 		\vv	P		& = & I_1 + 1                   & \pd &                                                 \\[\rowsep]
		\vv	Q		& = & n - I_3                   & \pd &                                                 \\[\rowsep]
		\vv	\delta	& = & \indicator{\arrayA[Q] > Q}      & \ed &		\bernoulli\bigl(\tfrac{I_3}{n-2}\bigr)          \\[\rowsep] 
		\vv \toll A& = & 1                         & \pd &                                                 \\[\rowsep]
 		\vv	\toll B	& = & \indicator{\arrayA[\id{left}] > \arrayA[\id{right}]}
                                                & \ed & \bernoulli\bigl(\tfrac12\bigr)                  \\[\rowsep]
 		\vv	\toll F	& = & \delta                    & \ed & \bernoulli\bigl(\tfrac{I_3}{n-2}\bigr)          \\[\rowsep]
 		\vv	\C{1}	& = & Q - 2 + \delta            & \ed &  I_1 + I_2
 		                                                + \bernoulli\bigl(\tfrac{I_3}{n-2}\bigr)        \\[\rowsep]
 		\vv	\C{3}	& = & n - Q                     & \ed & I_3                                             \\[\rowsep]
 		\vv	\C{4}	& = & \delta + (\numberat{sm}G)   & \ed & \bernoulli\bigl(\tfrac{I_3}{n-2}\bigr) + \hypergeometric(I_1+I_2,I_3,n-2)              \\[\rowsep]
 		\vv	\S{1}   & = & \satK                     & \ed & \hypergeometric(I_1,I_1+I_2,n-2)                \\[\rowsep]
 		\vv	\S{3}   & = & P-1 - (\satK)               & \ed & I_1 - \hypergeometric(I_1,I_1+I_2,n-2)          \\[\rowsep]
		\bottomrule
	\end{tabular}
}
\end{table}

In Section~\ref{sec:expected-frequencies}, we determined for each basic block of
Yaroslavskiy's algorithm, how often it is executed in one partitioning
step.
There, we only used the expected values in the end, but we already characterized
the full distributions in passing.
They are summarized in
Table~\ref{tab:distribution-of-tolls} for reference.

Most of those distributions are in fact \emph{mixed} distributions, i.\,e.,
their parameters depend on the random variable $I=(I_1,I_2,I_3)$, namely the
sizes of the subproblems for recursive calls.
For example, we find that $\delta$ conditional on the
event $(I_1,I_2,I_3) = (i_1,i_2,i_3)$ is Bernoulli
$\bernoulli(i_3\mathbin/(n-2))$ distributed, which we briefly write as
$\delta \mathrel{\smash\eqdist}
\bernoulli(I_3\mathbin/(n-2))$.
Note that since $I$ has itself a mixed distribution\,---\,namely conditional on
$D$\,---\,we actually have three layers of random variables: spacings,
subproblem sizes and toll functions.
The key technical lemmas for dealing with these three-layered distributions are
given in Section~\ref{sec:asymptotics-lemmas}.

\needspace{2cm}
\subsubsection{Distributional Recurrence}
\label{sec:distributional-recurrence}

Denote by $\pc_n$ the (random) costs of the first partitioning step of
Yaroslavskiy's algorithm.
By Property~\ref{pro:randomness-preservation}, subproblems generated in
the first partitioning phase are, conditional on their sizes,
again uniformly random permutations and independent of each other.
Hence, we obtain the distributional recurrence for the (random) total
costs~$C_n$:
\begin{align}
\label{eq:distributional-recurrence-generic}
		C_n
	\wwrel\eqdist
		C'_{I_1} + C''_{I_2} + C'''_{I_3} 
		\wbin+	\pc_n\,, \qquad (n\ge 3),
\end{align}
where $(I_1,I_2,I_3,\pc_n)$, 
$(\tCpj)_{j\ge 0}$, 
$(\tCppj)_{j\ge 0}$, 
$(\tCpppj)_{j\ge0}$ 
are independent and
$\tCpj$, $\tCppj$, $\tCpppj$ 
are identically distributed as
$C_j$ for $j\ge 0$.
By Theorems~\ref{thm:expected-comparisons}, \ref{thm:expected-swaps}
and~\ref{thm:expected-bytecodes}, we know the expected costs $\E[C_n]$, 
so with \smash{$C^*_0 \ce 0$} and
\begin{align}
\label{eq:def-normalized-costs}
	C^*_n &\wwrel\ce
		\frac{C_n-\E[C_n]}{n}\,,
		\qquad \text{for } n\ge 1,
\end{align}
we have a sequence $(\tCsn)_{n\ge 0}$ of centered, square integrable
random variables. 
Using \eqref{eq:distributional-recurrence-generic} we
find, cf.~\cite[eq.~(27), (28)]{hwne02},
that $(\tCsn)_{n\ge 0}$  satisfies~\eqref{eq:general-contraction-recurrence}
with
\begin{align}
\label{eq:general-coefficients-contraction}
	\ui{A}{n}_r &\wwrel= \frac{I_r}{n}, &
	\ui{b}{n}   &\wwrel=
		\frac{1}{n}\biggl(\pc_n - \E[C_n]
			\wbin+ \sum_{r=1}^3 \E[C_{I_r}\given I_r]
		\biggr)\,,
\end{align}
so we can apply the framework of the contraction method.
It remains to check the conditions 
\ref{cond:convergence-of-recurrence-coeffs},
\ref{cond:recurrence-contraction} and
\ref{cond:recurrence-small-sublists-seldom} to prove that $\tCsn$ indeed
converges to a limit law;
the detailed computations are given in Appendix~\ref{app:contraction-proofs}.
The key results needed therein are presented in the following section
as technical lemmas.

\subsection{Asymptotics of Mixed Distributions}
\label{sec:asymptotics-lemmas}

The following convergence results for mixed distributions are essential for
proving condition~\ref{cond:convergence-of-recurrence-coeffs}.

\begin{lemma}
\label{lem:asymptotics-hypergeometric}
	Let $(V_1,\ldots,V_b)$ be a vector of random probabilities, i.\,e.,
	$0\le V_r\le 1$ for all $r=1,\ldots,b$ and $\sum_{r=1}^b V_r=1$ almost surely.
	Let
	\begin{align}
		(L_1,\ldots,L_b) \wwrel\ce
		(\ui Ln_1,\ldots,\ui Ln_b) \wwrel\eqdist \multinomial(n;V_1,\ldots,V_b),
	\end{align}
	be mixed multinomially distributed.
	Furthermore for $J_1,J_2 \subset \{1,\ldots,b\}$ let
	\begin{align}
		Z_n \wwrel\eqdist
		\hypergeometric\biggl(
			\sum_{j\in J_1} L_j ,\; \sum_{j\in J_2} L_j ,\; n
		\biggr)
	\end{align}
	be mixed hypergeometrically distributed.
	Then we have the $L_2$-convergence, as $n\to\infty$,
	\begin{align}
		\frac{Z_n}{n}
			&\wwrel{\stackrel{L_2}{\longrightarrow}}
		\biggl(
			\sum_{j\in J_1} \!V_j
		\biggr)\cdot\biggl(
			\sum_{j\in J_2} \!V_j
		\biggr).
	\end{align}
\end{lemma}

The proof exploits that the binomial and the hypergeometric
distributions are both strongly concentrated around their means.
The full-detail computations to lift this to conditional expectations
are given in Appendix~\ref{app:contraction-proofs}.
%

\smallskip
\begin{lemma}
\label{lem:asymptotics-xlnx}
	For $L_1,\ldots,L_b$ from Lemma~\ref{lem:asymptotics-hypergeometric}, we
	have for $1\le i\le b$ the $L_2$-convergence
	\begin{align}
			\tfrac{L_i}n \ln \bigl( \tfrac{L_i}n \bigr)
		&\wwrel{\stackrel{L_2}{\longrightarrow}}
			V_i \ln (V_i)\,,
		\qquad\text{as } n \to \infty.
	\end{align}
\end{lemma}

(Recall that we set $x\ln(x)\ce0$ for $x=0$.)\\
The proof is directly obtained by combining the \textsl{law of large
numbers} with the \textsl{dominated convergence theorem}; see
Appendix~\ref{app:contraction-proofs} for details.

\subsection{Key Comparisons}

We have the following asymptotic results on the variance and distribution of
the number of key comparisons of Yaroslavskiy's algorithm:

\begin{theorem}
\label{thm:limiting-dist-comparisons}
	For the number $C_n$ of key comparisons used by Yaroslavskiy's Quicksort when
	operating on a uniformly at random distributed permutation we have
	\begin{align}
			\frac{C_n - \E[C_n]}{n}
		&\wwrel\to
			C^*,
		\qquad (n\to \infty),
	\end{align}
	where the convergence is in distribution and with second moments.
	The distribution of~$C^*$ is determined as the unique fixed point,
	subject to $\E[X] = 0$ and $\E[X^2] < \infty$, of
	\begin{align} \label{eq:fix-point-eq-comparisons}
		X &\wwrel\eqdist
			1 \bin+ (D_1+D_2)(D_2+2D_3)
			  \wbin+ \sum_{j=1}^3 \left(D_j \ui{X}{j}
			  +	\tfrac{19}{10} D_j \ln D_j \right)\,,
	\end{align}
	where $(D_1,D_2,D_3)$, $\ui{X}{1}$, $\ui{X}{2}$ and $\ui{X}{3}$ are
	independent and $\ui{X}{j}$ has the same distribution as $X$ for $j\in\{1,2,3\}$.
	Moreover, we have, as $n\to \infty$,
	\begin{align}
		\V(C_n) &\wwrel\sim  \sigma_C^2 n^2
		\qquad\quad\text{with}\qquad\quad
		\sigma_C^2 \wwrel= \tfrac{2231}{360}-\tfrac{361}{600}\pi^2=
		0.25901\!\ldots
	\end{align}
\end{theorem}

For the proof, we apply the contraction method to $X_n = \tCsn$ as defined by
\eqref{eq:def-normalized-costs}, where the toll function
$\comparisonmarker[n]{}$ is used.
Details on checking conditions 
\ref{cond:convergence-of-recurrence-coeffs},
\ref{cond:recurrence-contraction} and
\ref{cond:recurrence-small-sublists-seldom}, as well as the 
derivation of the resulting fixed-point equation for $C^*$ and the computation
of the variance are given in Appendix~\ref{app:contraction-proofs}.

\subsection{Swaps}

For the number of swaps in Yaroslavskiy's algorithm we have the following
asymptotic behavior of variance and distribution.

\begin{theorem}
\label{thm:limiting-dist-swaps}
	For the number $\totalswapmarker[n]{}$ of swaps used by Yaroslavskiy's
	algorithm when operating on a random permutation we	have
	\begin{align}
			\frac{S_n - \E[S_n]}{n}
		&\wwrel\to
			S^*\,,
			\qquad (n\to \infty)\,,
	\end{align}
	where the convergence is in distribution and with second moments.
	The distribution of $S^*$ is determined as the unique fixed point,
	subject to $\E[X]=0$ and $\E[X^2]<\infty$, of
	\begin{align}
		X &\wwrel\eqdist
			  D_1 \bin+ (D_1 + D_2) \, D_3
			\wbin+ \sum_{j=1}^3 \left(
					D_j \ui{X}{j} + \tfrac{3}{5} D_j \ln D_j
			  \right)\,,
	\end{align}
	where $(D_1,D_2,D_3)$, $\ui{X}{1}$, $\ui{X}{2}$ and $\ui{X}{3}$ are independent
	and $\ui{X}{j}$ has the same distribution as $X$ for $j\in\{1,2,3\}$.
	Moreover, we have, as $n\to \infty$,
	\begin{align}
		\V(S_n) &\wwrel\sim  \sigma_S^2 n^2,
		\qquad\text{with}\qquad
			\sigma_S^2
		\wwrel=
			\tfrac{7}{10} - \tfrac{3}{50}\pi^2
		\wwrel= 0.10782\!\ldots
	\end{align}
\end{theorem}

The proof is similar to the one for Theorem~\ref{thm:limiting-dist-comparisons},
details are given in Appendix~\ref{app:contraction-proofs}.

\needspace{2cm}
\subsection{Executed Bytecode Instructions}

For the number of executed Java Bytecode instructions in Yaroslavskiy's
algorithm we have the following asymptotic variance and
distribution.

\begin{theorem}
\label{thm:limiting-dist-bytecodes}
	For the number $\bytecodes[n]{}$ of executed Java Bytecodes used by
	Yaroslavskiy's algorithm when sorting a random permutation, we have
	\begin{align}
			\frac{\bytecodes[n]{} - \E[\bytecodes[n]{}]}{n}
		&\wwrel\to
			\bytecodes{}^*\,,
			\qquad (n\to \infty)\,,
	\end{align}
	where the convergence is in distribution and with second moments.
	The distribution of $\bytecodes{}^*$ is determined as the unique fixed
	point, subject to $\E[X]=0$ and $\E[X^2]<\infty$, of
	\begin{align}
		X &\wwrel\eqdist
			  24 + (D_3-9)D_2 - 2D_3(5D_3+2)
			\wbin+ \sum_{j=1}^3 \left(
					D_j \ui{X}{j} + \tfrac{217}{10} D_j \ln D_j
			  \right)\,,
	\end{align}
	where $(D_1,D_2,D_3)$, $\ui{X}{1}$, $\ui{X}{2}$ and $\ui{X}{3}$ are independent
	and $\ui{X}{j}$ has the same distribution as $X$ for $j\in\{1,2,3\}$.
	Moreover, we have, as $n\to \infty$,
	\begin{align}
		\V(\bytecodes[n]{}) &\wwrel\sim  \sigma_{\bytecodes{}}^2 n^2,
		\qquad\text{with}\qquad
			\sigma_{\bytecodes{}}^2
		\wwrel=
			\tfrac{1\,469\,983}{1\,800} - \tfrac{47\,089}{600}\pi^2
		\wwrel= 42.0742\!\ldots
	\end{align}
\end{theorem}

Again, the proof is similar to the one for
Theorem~\ref{thm:limiting-dist-comparisons} and details are given in
Appendix~\ref{app:contraction-proofs}.

\subsection{Covariance of Comparisons and Swaps}

In this section we study the asymptotic \emph{covariance} $\Cov(C_n,S_n)$
between the number of key comparisons $C_n$ and the number of swaps $S_n$ in
Yaroslavskiy's algorithm.
\begin{theorem}\label{thm:limit-covariance}
	For the number $C_n$ of key comparisons and the number $S_n$ of swaps used by
	Yaroslavskiy's algorithm on a random permutation, we have for
	$n\to\infty$
	\begin{align}
		\Cov(C_n,S_n) &\wwrel\sim
			\sigma_{C,S} \, n^2
		\qquad\text{with}\qquad
	 	\sigma_{C,S} \wwrel=
	 		\tfrac{28}{15} - \tfrac{19}{100}\pi^2
	 		\wwrel = -0.00855817\!\ldots
	\label{corcs}
	\end{align}
	The \textsl{correlation coefficient} of $C_n$ and $S_n$ consequently is
	$$
			\frac{\Cov(C_n,S_n)}{\sqrt{\V(C_n)}\sqrt{\V(S_n)}}
		\wwrel\sim \rho
		\wwrel\approx	-0.0512112\!\ldots
	$$
\end{theorem}

For the proof of Theorem~\ref{thm:limit-covariance}, we consider the bivariate
random variables
\smash{$Y_n \ce \bigl({C_n \atop S_n}\!\bigr)$} 
and show that their normalized versions
$\tYsn \ce \tfrac1n \bigl( Y_n - \E[Y_n] \bigr)$ 
converge to the (bivariate)
distribution $\mathcal L (\Lambda_1,\Lambda_2)$, which is,
as before, characterized as the unique fixed point of a distributional equation.
The covariance of $C_n$ and $S_n$ is then obtained by multiplying both
components of $\tYsn$, which converges to $\E[\Lambda_1 \cdot \Lambda_2]$.
Full computations are found in Appendix~\ref{app:contraction-proofs}.

\smallskip
Note that Theorem~\ref{thm:limit-covariance} and its proof directly imply the asymptotic
variance and limit distribution of \emph{all} linear combinations
$\alpha C_n + \beta S_n$, for $\alpha,\beta \in \Rset$, which are, for
$\alpha,\beta > 0$, natural cost measures when weighting key comparisons against
swaps.
The reason is that in the proof of Theorem \ref{thm:limit-covariance} we show the bivariate
limit law
\begin{align*}
		\left( \frac{C_n-\E[C_n]}{n} , \frac{S_n-\E[S_n]}{n} \right)
	&\wwrel\to	
		(\Lambda_1,\Lambda_2),
\end{align*}
which holds in distribution and with second mixed moments.
Hence, the \weakemph{continuous mapping theorem} implies, as $n\to\infty$,
\begin{align*}
		\frac{\alpha C_n + \beta S_n \wbin-
		(\alpha\E[C_n] + \beta\E[S_n])}{n}
	&\wwrel\to
		\alpha\Lambda_1 + \beta\Lambda_2,
\end{align*}
in distribution and with second  moments. Thus, we obtain, as $n\to\infty$,
\begin{align*}
		\V(\alpha C_n + \beta S_n)
	&\wwrel=
		\alpha^2 \V(C_n) \bin+ \beta^2 \V(S_n)
		\bin+ 2\alpha\beta \, \Cov(C_n,S_n)\\
	&\wwrel\sim
		( \alpha^2 \sigma^2_C
		+ \beta^2\sigma^2_S
		+ 2\alpha\beta \,	\sigma_{C,S}) \, n^2.
\end{align*}
Note that by this approach also the covariances between all the single
contributions from Table~\ref{tab:distribution-of-tolls} that contribute with
linear order in the first partitioning step to the number of executed Java
Bytecodes used by Yaroslavskiy's algorithm can be identified asymptotically in
first order.

\section{Conclusion}
\label{sec:conclusion}

\noindent
In this paper, we conducted a fully detailed analysis of Yaroslavskiy's
dual-pivot Quicksort\,---\,including the optimization of using Insertionsort on
small subproblems\,---\,in the style of Knuth's book series 
\weakemph{The Art of Computer Programming}.
We give the exact expected number of executed Java Bytecode instructions for
Yaroslavskiy's algorithm.%
\footnote{
	Bytecode instructions serve merely as a sample of one possible detailed cost
	measure; implementations in different low level languages can easily be analyzed
	using the our block execution frequencies.
}
On top of the exact average case results, we establish existence and
fixed-point characterizations of limiting distribution of normalized costs.
From this, we compute moments of the limiting distributions, in
particular the asymptotic variance of the number of executed Bytecodes.
The mere fact that such a detailed average and distributional analysis is
tractable, seems worth noting.
For the reader's convenience, we summarize the main results of this paper in
Table~\ref{tab:executive-summary}, where we also cite corresponding results
on classic single-pivot Quicksort for comparison.

\begin{table}
\tbl{%
	Summary of the results of this paper and comparison with corresponding results
	for classic Quicksort. $M$~is chosen such that the number of executed Bytecodes in
	minimized. 
	(For swaps, results for $M=1$ are given as Insertionsort is not	swap-based.)
	Some asymptotic approximations have been weakened for conciseness of
	presentation, see the main text for full precision.
	The asymptotics use the well-known expansion of Harmonic Numbers [e.\,g.\
	\protect\citeNP[eq.\,(6.66)]{ConcreteMathematics}]
	\label{tab:executive-summary}
}{%
	\newcommand\nlnnminusbn[3][-]{$#2 n \ln n \mathbin{#1} #3 n$}
	\def\an#1{$#1 n$}
	\def\stdev{std.\,dev.}
	\begin{threeparttable}
	\begin{tabular}{r@{\quad}l@{\;\;}lcc}
		\toprule
			\multicolumn{2}{c}{\multirow2*{Cost Measure}}
			&& Yaroslavskiy's Quicksort	& Classic Quicksort \\
			&& error
			 &  with $M=7$ 				& with $M=6$ \\
		\midrule
		\multirow{2}*{Comparisons}
			& expectation & $\bo(\log n)$ 	
				& \nlnnminusbn{1.9}{2.49976}
				& \nlnnminusbn{2}{2.3045}\tnotex{tn:sedgewick77}  \\
			& \stdev & $o(n)$	
				& \an{0.50893}
				& \an{0.648278}\tnotex{tn:hennequin89} \\[1ex]
		\multirow{2}*{Swaps (for $M=1$)}
			& expectation & $\bo(\log n)$
				& \nlnnminusbn{0.6}{0.107004}
				& \nlnnminusbn{0.\overline 3}{0.585373}\tnotex{tn:sedgewick77} \\
			& \stdev & $o(n)$		
				& \an{0.328365}
				& \an{0.0237251}\tnotex{tn:hennequin89} \\[1ex]
		\multirow{1}*{Writes Accesses}
			& expectation & $\bo(\log n)$
				& \nlnnminusbn{1.1}{0.408039}
				& \nlnnminusbn[+]{0.\overline 6}{0.316953}\tnotex{tn:sedgewick77} \\[1ex]
		\multirow{2}*{Executed Bytecodes}
			& expectation & $\bo(\log n)$
				& \nlnnminusbn{21.7}{3.56319}
				& \nlnnminusbn[+]{18}{6.21488}\tnotex{tn:bytecodes} \\
			& \stdev & $o(n)$		
				& \an{6.48646}
				& \an{3.52723}\tnotex{tn:stdevBytecodes} \\[1.5ex]
		\multicolumn{2}{r}{Correlation Coefficient for}
			& \multirow2*{$o(1)$}
				& \multirow2*{$-0.0512112$}
				& \multirow2*{$-0.86404$\tnotex{tn:corr}} \\
		\multicolumn{2}{r}{Comparisons and Swaps} \\
		\bottomrule
	\end{tabular}
	\begin{scriptsize}
	\begin{multicols}2
	\begin{tablenotes}
		\smallskip
		\item[*] \label{tn:sedgewick77}
			see \cite[p.\,334]{Sedgewick1977}.
		\item[\dag] \label{tn:hennequin89}
			see \cite[p.\,330]{hennequin1989combinatorial}.
		\item[\ddag] \label{tn:bytecodes}
			see \cite[p.\,123]{wild2012thesis}.
		\item[\S] \label{tn:stdevBytecodes}
			as in \cite[p.\,515]{neininger2001multivariate} for MIX, \\
			but with Bytecode costs of \cite{wild2012thesis}.
		\item[\$] \label{tn:corr}
			see \cite[Table\,1]{neininger2001multivariate}.
	\end{tablenotes}
	\end{multicols}
	~\\[-3\baselineskip]
	\end{scriptsize}
	\end{threeparttable}
}
\end{table}

As observed by \citeN{Wild2012}, Yaroslavskiy's algorithm uses $5$\,\% less key
comparisons, but $80$\,\% more swaps in the asymptotic average than classic
Quicksort.
Unless comparisons are very expensive, one should expect classic Quicksort to be
more efficient in total.
This intuition is confirmed by our detailed analysis:
In the asymptotic average, the Java implementation of Yaroslavskiy's algorithm
executes $20$\,\% more Java Bytecode instructions than a corresponding
implementation of classic Quicksort.

Strengthening confidence in expectations, we find that
asymptotic standard deviations of all costs remain linear in $n$;
by \textsl{Chebyshev's inequality}, this implies concentration around the mean.
Whereas the number of comparisons in Yaroslavskiy's algorithm shows slightly
less variance than for classic Quicksort, swaps exhibit converse behavior.
In fact, the number of swaps in classic Quicksort is highly
concentrated because it already achieves close to optimal average behavior:
In the partitioning step of classic Quicksort, every swap puts \emph{both}
elements into the correct partition and we never revoke a placement during
one partitioning step.
In contrast, in Yaroslavskiy's algorithm every swap puts only one element into
its final location (for the current partitioning step);
the other element might have to be moved a second time later.

Another facet of this difference is revealed by considering the correlation
coefficient between swaps and comparisons.
In classic Quicksort, swaps and comparisons are almost perfectly
\emph{negatively} correlated.
A ``good'' run w.\,r.\,t.\ comparisons needs balanced partitioning, but the more
balanced partitioning becomes, the higher is the potential for misplaced
elements that need to be moved.
In Yaroslavskiy's partitioning method, such a clear dependency does not exist
for several reasons.
First of all, even if pivots have extreme ranks, sometimes many swaps are done;
e.\,g.\ if $p$ and $q$ are the two largest elements, \emph{all} elements are
swapped in our implementation.
Secondly, for some pivot ranks, comparisons and swaps behave covariantly:
For example if $p$ and $q$ are the two smallest elements, no swap is done and
every element's partition is found with one comparison only.
In the end, the number of comparisons and swaps is almost
uncorrelated in Yaroslavskiy's algorithm.

The asymptotic standard deviation of the total number of executed Bytecode
instructions is about twice as large in Yaroslavskiy's algorithm as in
classic Quicksort.
This might be a consequence of the higher variability in the number of swaps
just described.

Concerning practical performance, asymptotic behavior is not the full story.
Often, inputs in practice are of moderate size and only the massive number of
calls to a procedure makes it a bottleneck of overall execution.
Then, lower order terms are not negligible. For Quicksort, this means in
particular that constant overhead per partitioning step has to be taken into
account.
For tiny $n$, this overhead turns out to be so large, that it pays to switch
to a simpler sorting method instead of Quicksort.
We showed that using \proc{InsertionSort} for subproblems of size at most $M$
speeds up Yaroslavskiy's algorithm significantly for moderate $n$.
The optimal choice for $M$ w.\,r.\,t.\ the number of executed Bytecodes is
$M=7$.

Combining the results for \proc{InsertionSort} from
Appendix~\ref{app:insertionsort} and a corresponding Bytecode count analysis of
a Java implementation of classic Quicksort \cite{wild2012thesis}, we can compare
classic Quicksort and Yaroslavskiy's algorithm exactly.
As striking result we observe that in expectation, Yaroslavskiy's algorithm
needs more Java Bytecodes than classic Quicksort \emph{for all $n$}.
Thus, the efficiency of classic Quicksort in terms of executed Bytecodes
is not just an effect of asymptotic approximations, it holds for realistic input
sizes, as well.

These findings clearly contradict corresponding running time experiments
\cite[Chapter\,8]{wild2012thesis}, where Yaroslavskiy's algorithm was
significantly faster across implementations and programming environments.
One might object that the poor performance of Yaroslavskiy's algorithm is a
peculiarity of counting Bytecode instructions.
\citeN[Section\,7.1]{wild2012thesis} also gives implementations and analyses
thereof in MMIX, the new version of Knuth's imaginary processor architecture.
Every MMIX instruction has well-defined costs, chosen to closely resemble actual
execution time on a simple processor.
The results show the same trend:
Classic Quicksort is more efficient.
Together with the Bytecode results of this paper, we see strong evidence
for the following conjecture:
\begin{conjecture}
The efficiency of Yaroslavskiy's algorithm in practice is caused by advanced
features of modern processors.
In models that assign constant cost contributions to single
instructions\,---\,i.\,e., locality of memory accesses and instruction
pipelining are ignored\,---\,classic Quicksort is more efficient.
\end{conjecture}
It will be the subject of future investigations%
\footnote{%
	Indeed, progress has been made since this article was submitted.
	\citeN{Kushagra2014} analyzed Yaroslavskiy's algorithm in
	the \textsl{external memory model} and show that it needs significantly less
	I/Os than classic Quicksort. 
	Their results indicate that with modern memory hierarchies, using even more
	pivots in Quicksort might be beneficial, since intuitively speaking, more work
	is done in one ``scan'' of the input.
} 
to identify the true reason of
the success of Yaroslavskiy's dual-pivot Quicksort.

\clearpage
\appendix
\section*{APPENDIX}

\section{Solving the Dual-Pivot Quicksort Recurrence}
\label{app:recurrence-proofs}

The proof presented in the following is basically a generalization of the
derivation given by \citeN[p.\,156ff]{Sedgewick1975}. 
\citeN{hennequin1991analyse} gives an alternative approach based on
generating functions that is much more general. 
Even though the authors consider \person{Hennequin}'s
method elegant, we prefer the elementary proof, as it allows
a self-contained presentation.

Two basic identities involving binomials and Harmonic Numbers are used several
times below, so we collect them here. They are found as equations~(6.70)
and~(5.10) in \cite{ConcreteMathematics}.
\begin{align}
	\sum_{0\le k < n} \tbinom km \harm k 
	&\wwrel= \tbinom n{m+1} \bigl( \harm n - \tfrac1{m+1} \bigr)\,,
		&&\hspace*{-1cm}\text{integer }m\ge	0 
	\label{eq:sum-of-harmonics} 
	\,, \\
	\sum_{0\le k \le n} \tbinom km 
	&\wwrel= \tbinom{n+1}{m+1} \,,
		&&\hspace*{-1cm}\text{integers } n,\, m \ge 0\;.  
	\label{eq:summation-in-upper-index}
\end{align}

\begin{proof}[of Theorem~\ref{thm:recurrence-solution}]
The first step is to use symmetries of the sum in \eqref{eq:recurrence}.
\begin{align*}
	\sum_{1\le p<q\le n}\!\!\!
 		\left(C_{p-1}+C_{q-p-1}+C_{n-q}\right)
	\wwrel=&       \sum_{p=1}^{n-1} (n-p)   C_{p-1}
	      \;\;+\;\; \sum_{k=0}^{n-2} (n-1-k) C_{k}
	      \;\;+\;\; \sum_{q=2}^{n}   (q-1)   C_{n-q}\\
	\wwrel=&
	      3\sum_{k=0}^{n-2}(n-k-1)C_{k}\;.
\end{align*}
So, our recurrence to solve is
\begin{align}
	C_n &\wwrel=	\pc_n + \tfrac{6}{n(n-1)}\sum_{k=0}^{n-2} 
				(n-k-1)C_{k}\,,
		\qquad\mbox{for } n>M\;.
\end{align}
We first consider $D_{n}:=\tbinom{n+1}{2}C_{n+1}-\tbinom{n}{2}C_{n}$
to get rid of the factor in front of the sum:
\begin{align*}
	D_n &\wwrel= \overbrace{
				  \tbinom{n+1}2 \pc_{n+1} 
				- \tbinom{n}{2}\pc_n
			}^{d(n)\ce}
		\qquad\qquad(n\ge M+2)\\ 
		& \wwrel\pe {}+ 
			\tfrac{(n+1)n}{2}\tfrac{6}{(n+1)n}\sum_{k=0}^{n-1} (n-k) C_{k} 
		\wwbin- 
			\tfrac{n(n-1)}{2}\tfrac{6}{n(n-1)}\sum_{k=0}^{n-2} (n-k-1) C_{k} \\[-.5ex] 
		& \wwrel= d(n) \;+\; 3\sum_{k=0}^{n-1} C_k
	\;.
\end{align*}
The remaining full history recurrence is eliminated by taking ordinary
differences
\begin{align*} 
	E_{n} & \wrel\ce D_{n+1}-D_{n}
			\wrel= d(n+1)-d(n)+3C_{n}\;.\qquad(n\ge M+2)
\end{align*}
Towards a telescoping recurrence, we consider
\(
	F_n \wrel\ce C_n - \tfrac{n-4}{n}\cdot C_{n-1}\,,
\)
and compute%
\begin{align}
	F_{n+2}-F_{n+1} &\wwrel=	 
						  C_{n+2} 
						- \tfrac{n-2}{n+2}C_{n+1} 
						- \left(C_{n+1}-\tfrac{n-3}{n+1}C_{n}\right)
				\notag\\ 
					&\wwrel=	  
						  C_{n+2}
						- \tfrac{2n}{n+2}C_{n+1}
						+ \tfrac{n-3}{n+1}C_{n}
				\;.
	\label{eq:difference-Fn}
\end{align}
The expression on the right hand side in itself is not helpful. 
However, by expanding the definition of $E_{n}$, we find
\begin{align}
	\left(E_{n}-3C_{n}\right) \bigm/ \tbinom{n+2}{2} 
		&\wwrel=	\left(D_{n+1} - D_{n} - 3C_{n}\right) \bigm/ \tbinom{n+2}{2} \notag\\ 
		&\wwrel=	\biggl(
						\tbinom{n+2}{2} C_{n+2} - \tbinom{n+1}{2} C_{n+1} 
						- \left(\tbinom{n+1}{2}C_{n+1}
						- \tbinom{n}{2} C_{n}\right)
						- 3C_{n}
					\biggr) \Big/ \tbinom{n+2}{2}		\notag\\[-.5ex] 
		&\wwrel	=	  C_{n+2} - \tfrac{2n}{n+2}C_{n+1}
					+ \tfrac{\frac{1}{2}n(n-1)-3}{\frac{1}{2}(n+2)(n-1)} C_{n} \notag\\ 
		&\wwrel=	  C_{n+2}-\tfrac{2n}{n+2}C_{n+1}
					+ \tfrac{\frac{1}{2}(n-3)(n+2)}{\frac{1}{2}(n+2)(n-1)} C_{n} \notag\\ 
		&\wwrel=	C_{n+2} - \tfrac{2n}{n+2}C_{n+1} + \tfrac{n-3}{n+1}C_{n}
	\;.
	\label{eq:expansion-En-minus-three-Cn}
\end{align}
Equating \eqref{eq:difference-Fn} and \eqref{eq:expansion-En-minus-three-Cn}
yields
\[
			F_{n+2}-F_{n+1}
	\wrel=	\left(E_{n}-3C_{n}\right) \bigm/ \tbinom{n+2}{2}
	\wrel=	\underbrace{
				\bigl(d(n+1)-d(n)\bigr) \bigm/ \tbinom{n+2}{2}
			}_{f(n)\ce}
	\;.		\quad (n \ge M+2)
\]
This last equation is now amenable to simple iteration:
\begin{align*}
	F_n & \wrel= 	\underbrace{
						\sum_{i=M+4}^{n} \!\! f(i-2) \wbin+ F_{M+3}
					}_{g(n)}\;.
		\qquad (n \ge M+4)
\end{align*}
Plugging in the definition of $F_{n}=C_{n}-\frac{n-4}{n}\cdot C_{n-1}$
yields
\begin{align}
\label{eq:Cn-in-terms-of-gn}
	C_{n} \wrel= \tfrac{n-4}{n}\cdot C_{n-1} \wbin+ g(n)\;.
	\qquad (n \ge M+4)
\end{align}
Multiplying \eqref{eq:Cn-in-terms-of-gn} by $\tbinom n4$ and using
$\tbinom{n}{4}\cdot\frac{n-4}{n}=\tbinom{n-1}{4}$ 
gives a telescoping
recurrence:
\begin{align} 
	G_{n} \wrel\ce \tbinom{n}{4}C_{n}	
		&\wwrel=	G_{n-1} + \tbinom{n}{4}g(n) \notag\\
 	 	&\wwrel=	\mkern-10mu\sum_{i=M+4}^{n}\mkern-5mu
	 					\tbinom{i}{4}g(i) \bin+ G_{M+3} \notag\\
	 	&\wwrel=	\sum_{i=1}^{n} \tbinom{i}{4}g(i) \wbin+ G_{M+3}
	 				\wbin- \sum_{i=1}^{M+3} \tbinom i4 \biggl(
	 						\:
	 						\smash{\overbrace{
	 							\rule{0pt}{3.1ex}\smash{\sum_{j=M+4}^{i}} f(j-2)
	 						}^{=\.0}} 
	 						\wbin+ F_{M+3}
	 				  \biggr)  \notag\\
	 	&\wwrel=	\sum_{i=1}^{n} \tbinom{i}{4}g(i)
	 				\wbin+ G_{M+3} \wbin- \tbinom{M+4}5 F_{M+3} \,,
	 \label{eq:explicit-Gn}
\end{align}
where the last equation uses~\eqref{eq:summation-in-upper-index}.
Applying definitions, we find
\begin{align}
		\sum_{i=1}^{n} \tbinom i4 g(i)
	&\wwrel=	\sum_{i=1}^{n}\tbinom i4\biggl(
					  F_{M+3}
					\wbin+ \sum_{j=M+2}^{i-2}\frac{d(j+1)-d(j)}{\tbinom{j+2}{2}}
				\biggr)
			\notag \\ 
	&\wwrel=	\tbinom{n+1}5 F_{M+3}
			\wwbin+
				\!\!\!\sum_{i=M+4}^{n}\!\!\!\tbinom i4\!\!\!
				\sum_{j=M+2}^{i-2}\!\!\!\bigg(
					  \pc_{j+2}
					\bin- \tfrac{2j}{j+2} \, \pc_{j+1}
					\bin+ \frac{\tbinom{j}{2}}{\tbinom{j+2}{2}} \, \pc_j \!
				\biggr)
	\;. \label{eq:sum-of-gi}
\end{align}
Using \eqref{eq:sum-of-gi} in \eqref{eq:explicit-Gn}, we finally arrive at the
explicit formula for $C_n$ valid for $n\ge M+3$:
\begin{align*}
	C_{n} 
 	& \wrel=	\frac{\tbinom{n+1}{5}}{\tbinom{n}{4}} F_{M+3}
 		\wbin+ 	\frac1{\tbinom n4}
 				\sum_{i=M+4}^{n}\!\!\!\tbinom i4\!\!\! \sum_{j=M+2}^{i-2}
 				\bigg(
					  	    \pc_{j+2}
					\bin- \tfrac{2j}{j+2} \, \pc_{j+1}
					\bin+ \frac{\tbinom{j}{2}}{\tbinom{j+2}{2}} \, \pc_j \!
				\biggr) \\
	& \wrel\pe \wbin+ \frac{G_{M+3} \bin- \tbinom{M+4}5 F_{M+3}}{\tbinom n4} 
	\;.\\
\intertext{Expanding $F$ and $G$ according to their definition gives} 
	& \wrel=	\frac1{\tbinom n4}
 				\sum_{i=M+4}^{n}\!\!\!\tbinom i4\!\!\! \sum_{j=M+2}^{i-2}
 				\bigg(
					  	    \pc_{j+2}
					\bin- \tfrac{2j}{j+2} \, \pc_{j+1}
					\bin+ \frac{\tbinom{j}{2}}{\tbinom{j+2}{2}} \, \pc_j \!
				\biggr) \\
		& \wwrel\pe \wbin+ \biggl(
			  \! \tfrac{n+1}5
			+ \frac{\tbinom{M+3}4-\tbinom{M+4}5}{\tbinom n4}
		\biggr) C_{M+3}
		\wwbin- \tfrac{M-1}{M+3} \biggl(
			  \! \tfrac{n+1}5
			- \frac{\tbinom{M+4}5}{\tbinom n4}
		\biggr) C_{M+2}\,,
\end{align*}
which concludes the proof.
\end{proof}

\begin{proof}[of Proposition~\ref{pro:Cn-explicit-linear-pc}]
Of course, we start with the closed form \eqref{eq:Cn-explicit} from
Theorem~\ref{thm:recurrence-solution}, which consists of the double sum and two
terms involving “base cases” $C_{M+2}$ and $C_{M+3}$.
\begin{align*}
		C_n &\wwrel= \frac1{\tbinom n4} 
				\sum_{i=M+4}^{n}\!\!\!\tbinom i4\!\!\! \sum_{j=M+2}^{i-2}
					\overbrace{
					\biggl(
						      \pc_{j+2} 
						\bin- \tfrac{2j}{j+2} \pc_{j+1}
						\bin+ \tfrac{\tbinom j2}{\tbinom{j+2}2} \pc_j
					\biggr)}^{\eta_j} 
	\\*
			& \wwrel\pe \wbin+
				\underbrace{
					\biggl(
						  \! \tfrac{n+1}5
						+ \frac{\tbinom{M+3}4-\tbinom{M+4}5}{\tbinom n4}
					\biggr) C_{M+3}
				  \wbin- \tfrac{M-1}{M+3} \biggl(
						  \! \tfrac{n+1}5
						- \frac{\tbinom{M+4}5}{\tbinom n4}
					\biggr) C_{M+2}
				}_{\rho}
\end{align*}
We first focus on the sums. 
Assuming the even more general form
$$\pc_n \wwrel= a\.n + b + \tfrac{c_1}{n-1} + \tfrac{c_2}{n} +
				\tfrac{c_3}{n(n-1)}\,,$$ 
partial fraction decomposition of the innermost term yields  
\begin{align*} 
		\eta_j \wrel\ce \bigl(
			\pc_{j+2}
			\bin- \tfrac{2j}{j+2} \, \pc_{j+1}
			\bin+ \tbinom{j}{2} \big/ \tbinom{j+2}{2}  \, \pc_j
		\bigr)
	&\wwrel= \tfrac{8a-2b}{j+2} - \tfrac{2a-2b}{j+1}\;.
\end{align*}
Note that contributions from $\tfrac{c_1}{n-1}$, $\tfrac{c_2}{n}$ and
$\tfrac{c_3}{n(n-1)}$ cancel out. 
This allows to write the inner sum in terms of Harmonic Numbers:
\begin{align}
	\sum_{j=M+2}^{i-2} \eta_j
	&\wwrel= 			(8a-2b)(\harm{i} - \harm{M+3})
				\bin- 	(2a-2b)(\harm{i-1} - \harm{M+2}) 
				\notag\\[-1ex]
	&\wwrel=			6a \. (\harm{i}-\harm{M+3})
				\bin+	(2a-2b)(\tfrac1i-\tfrac1{M+3})
	\;.
\end{align}
(The second equation uses the basic fact $\harm{k-1} = \harm k - \tfrac1k$.) \\
Using~\eqref{eq:sum-of-harmonics}, \eqref{eq:summation-in-upper-index}
and the absorption property of binomials 
$\tbinom nk = \tbinom{n-1}{k-1} \tfrac nk$,
one obtains
\begin{align}
		\frac1{\tbinom n4}
		\sum_{i=M+4}^{n}\!\!\!\tbinom i4\!\!\! \sum_{j=M+2}^{i-2}	\eta_j 
	&\wwrel=	\frac{6a}{\tbinom n4}	\Bigl(
						  \tbinom{n+1}5 (\harm{n+1} - \tfrac15) 
						- \tbinom{M+4}5 (\harm{M+4} - \tfrac15) 
					\Bigr) \notag\\*
	&\wwrel\pe \wbin+ \frac{2a-2b}{\tbinom n4} \sum_{i=M+3}^{n-1} \Bigl(
						\tfrac14\tbinom{i-3}3 - \tfrac1{M+3}\tbinom i4
					\Bigr) \notag\\*	
	&\wwrel\pe \wbin- \frac{6a}{\tbinom n4}\harm{M+3}\Bigl( 
							\tbinom{n+1}5 - \tbinom{M+4}5 
						\Bigr)
		\notag\\
	&\wwrel=	\tfrac65 a (n+1) (\harm{n+1} - \tfrac15) 
				\bin- 6a\frac{\tbinom{M+4}5}{\tbinom n4} (\harm{M+4} - \tfrac15) 
			\notag\\*
	&\wwrel\pe \wbin+ \frac{2a-2b}{\tbinom n4} \Bigl(
								\tfrac14 \bigl( \tbinom n4 - \tbinom{M+3}4 \bigr)
						\bin- 	\tfrac1{M+3} \bigl( \tbinom{n+1}5 - \tbinom{M+4}5 \bigr)
					\Bigr) \notag\\*	
	&\wwrel\pe \wbin- 6a\,\harm{M+3}\Bigl( 
						\tfrac{n+1}5 - \tbinom{M+4}5 \big/ \tbinom	n4 
					\Bigr)
		\notag\\
	&\wwrel=	\tfrac65 a (n+1) \harm{n+1}
				\wbin-	\tfrac{n+1}5 \bigl(
							6a\harm{M+3} + \tfrac{2a-2b}{M+3} + \tfrac65 a
						\bigr)
				\wbin+	\tfrac{a-b}{2} 
				\notag\\*
	&\wwrel\pe	
				\!\wbin+	\frac{\tbinom{M+4}5}{\tbinom n4} \bigl(
							  \tfrac65 a - 6a (\harm{M+4} -\harm{M+3})
							- \tfrac{a-b}2 \tfrac{5}{M+4} 
							+ \tfrac{2a-2b}{M+3}
						\bigr)
		\notag\\
	&\wwrel=	\tfrac65 a (n+1) \harm{n+1}
				\wbin-	\tfrac{n+1}5 \bigl(
							6a\harm{M+3} + \tfrac{2a-2b}{M+3} + \tfrac65 a
						\bigr)
				\notag\\*
	&\wwrel\pe	\wbin+	\tfrac{a-b}{2}
				\wbin+	\frac{\tbinom{M+4}5}{\tbinom n4} \bigl(
							  \tfrac65 a + \tfrac{2a-2b}{M+3}
							+ \tfrac{5b-17a}{2(M+4)} 
						\bigr)\;.
\label{eq:sum-of-eta}
\end{align}

It remains to consider the second and third summands of \eqref{eq:Cn-explicit}
\begin{align*}
	\rho &\wwrel\ce
		\biggl(
			  \! \tfrac{n+1}5
			+ \frac{\tbinom{M+3}4-\tbinom{M+4}5}{\tbinom n4}
		\biggr) C_{M+3} 
	\wbin- \tfrac{M-1}{M+3} \biggl(
			  \! \tfrac{n+1}5
			- \frac{\tbinom{M+4}5}{\tbinom n4}
		\biggr) C_{M+2}
	\;.
\end{align*}
We start by applying definition~\eqref{eq:recurrence} twice and using 
$C_n=\insertsortcost_n$ for $n \le M$ to expand $C_{M+3}$
\begin{align}
	C_{M+3} &\wwrel=	T_{M+3} \wbin+ \frac3{\tbinom{M+3}2} 
						\sum_{k=0}^{M+1} (M+2-k) C_k 		\nonumber
	\\		&\wwrel=	T_{M+3} \wbin+
						\frac3{\tbinom{M+3}2} \biggl( T_{M+1} + \frac3{\tbinom{M+1}2} \sum_{k=0}^{M-1} (M-k) \insertsortcost_k
						\biggr) 
							\wbin+ \frac3{\tbinom{M+3}2} \sum_{k=0}^{M} (M+2-k) \insertsortcost_k
					\nonumber
	\\		&\wwrel=	T_{M+3} \wbin+ \frac3{\tbinom{M+3}2} T_{M+1} 
						\wbin+ \frac3{\tbinom{M+3}2} \sum_{k=0}^{M}\Bigl( 
							(M+2-k) + \frac{3(M-k)}{\tbinom{M+1}2}
						\Bigr) \insertsortcost_k
\label{eq:C-M+3} 
\end{align}
and $C_{M+2}$
\begin{align}
	C_{M+2}	&\wwrel=	T_{M+2} \wbin+ \frac3{\tbinom{M+2}2} \sum_{k=0}^M 
							(M+1-k) \insertsortcost_k\;.
\label{eq:C-M+2}
\end{align}
Equations~\eqref{eq:C-M+3} and~\eqref{eq:C-M+2} are now inserted into the
second and third summands of \eqref{eq:Cn-explicit}.
With $T_n = a\.n+b$ for $n \ge M+1$, this yields
\begin{align}
	\rho 	&\wwrel=	\tfrac{n+1}5\Bigl(
							  (M+3)a+b
							+ \frac{3}{\tbinom{M+3}2}\,\bigl((M+1)a+b\bigr)
							- \tfrac{M-1}{M+3}\,\bigl((M+2)a+b\bigr)
						\Bigr) 		\nonumber
	\\*		&\wwrel\ppe
						\wbin+ \tfrac{n+1}5 \sum_{k=0}^M
							\Bigl(
								  \frac{3(M+2-k)}{\tbinom{M+3}2} 
								+ \frac{9(M-k)}{\tbinom{M+1}2 \tbinom{M+3}2}
								- \tfrac{M-1}{M+3}\frac{3(M+1-k)}{\tbinom{M+2}2}
							\Bigr) \insertsortcost_k 		\nonumber
	\\*		&\wwrel\ppe
						\wbin+ \frac{\tbinom{M+4}5}{\tbinom n4}\Bigl(
							  (\tfrac5{M+4} - 1) C_{M+3}
							+ \tfrac{M-1}{M+3} C_{M+2}
						\Bigr) 		\nonumber
	\\		&\wwrel=	\tfrac{n+1}5\Bigl(
							  5a + \tfrac{6(b-a)}{M+2} + \tfrac{2(4a-b)}{M+3}
						\Bigr)
						\wbin+ \tfrac{n+1}5 \sum_{k=0}^M
							\frac{3M-2k}{\tbinom{M+2}3} \insertsortcost_k
	\label{eq:rho}
	\\*		&\wwrel\ppe
						\wbin+ \frac{\tbinom{M+4}5}{\tbinom n4}\bigl(
							  \tfrac{M-1}{M+3} C_{M+2}
							- \tfrac{M-1}{M+4} C_{M+3}
						\bigr)		\nonumber
\end{align}
Adding~\eqref{eq:sum-of-eta} and~\eqref{eq:rho} finally yields the claimed
representation
\begin{align*}
	C_n	&\wwrel=	\tfrac65 a (n+1) \harm{n+1}
					\wbin+ \tfrac{n+1}5\bigl(
						\tfrac{19}5 a + \tfrac{6(b-a)}{M+2}- 6 a \harm{M+2}
					\bigr)
					\wbin+ \tfrac{a-b}2
					\wbin+ \tfrac{n+1}5 \sum_{k=0}^M
						\frac{3M-2k}{\tbinom{M+2}3} \insertsortcost_k
	\\*	&\wwrel\ppe
					\wbin+ \frac{\tbinom{M+4}5}{\tbinom n4} \bigl(
						  \tfrac65 a + \tfrac{2(a-b)}{M+3} + \tfrac{5b-17a}{2(M+4)}
						- \tfrac{M-1}{M+4} C_{M+3} + \tfrac{M-1}{M+3} C_{M+2}
					\bigr)
	\;.
\end{align*}

\medskip
For the asymptotic representation~\eqref{eq:Cn-explicit-linear-pc} of $C_n$, the
penultimate summand is $0$ because of the assumption $\insertsortcost_n=0$.
The last summand is in $\Theta(n^{-4})$ and therefore vanishes in the
$\bo(\tfrac1n)$ term (we assume $M=\Theta(1)$ as $n\to\infty$ to be constant).
Now, replacing $\harm{n}$ by its well-known asymptotic estimate
\begin{align*}
	\harm n	 &\wwrel=	\ln(n) + \gamma + \tfrac12 n + \bo\bigl(\tfrac1{n^2}\bigr)
		\qquad \text{ \cite[eq.\,(6.66)]{ConcreteMathematics} }
\end{align*}
and expanding terms in~\eqref{eq:Cn-exact-solution-linear-pc} directly
yields~\eqref{eq:Cn-explicit-linear-pc}.

Finally, the case $T_2=0 \ne 2a+b$ affects the derivation only at a
single point: As $M \ge 1$, the only occurring toll function that can ever equal
$T_2$ is $T_{M+1}$, which occurs only in $C_{M+3}$, see~\eqref{eq:C-M+3}.
In $\rho$, we multiply $C_{M+3}$ by 
$\bigl(\tfrac{n+1}5+\bigl(\tbinom{M+3}4-\tbinom{M+4}5\bigr)\big/ \tbinom
n4\bigr) \rel= \tfrac15(n+1) + \bo(n^{-4})$.
Consequently, we have to subtract
$$
		\delta_{M1}\Bigl(
			\tfrac15(n+1)\cdot \frac3{\tbinom{M+3}2} (2a+b) +
			\bo(n^{-4})
		\Bigr)
	\wwrel=
		\delta_{M1} \, \tfrac{2a+b}{10}(n+1)  + \bo(n^{-4})
	\;.
$$
(The second equation follows by setting $M=1$.)\\*
This concludes the proof of Proposition~\ref{pro:Cn-explicit-linear-pc}.
\end{proof}

\section{Insertionsort}
\label{app:insertionsort}

\begin{algorithm}
	\begin{codebox}
\Procname{$\proc{InsertionSort}(\arrayA,\id{left},\id{right})$}
\li	\For $i=\id{left}+1\,,\dots,\,\id{right}$
\li	\Do
		$j\gets i-1$; \quad
		$v\gets \arrayA[i]$
\li		\While $j \ge \id{left} \wbin\wedge v < \arrayA[j]$
\li		\Do
			$\arrayA[j+1] \gets \arrayA[j]$; \quad	\label{lin:insertionsort-knuth-write-1}
			$j\gets j-1$
		\EndWhile
\li		$\arrayA[j+1] \gets v$ \label{lin:insertionsort-knuth-write-2}
	\EndFor
\end{codebox}
	\caption{\protect\rule[-.75ex]{0pt}{2.75ex}
		Insertionsort as given and analyzed by
		\protect\citeN{Knuth1998}.
	}
	\label{alg:insertionsort-knuth}
\end{algorithm}

\noindent
In this section, we consider in some detail the \proc{InsertionSort} procedure
used for sorting small subproblems. 
Insertionsort is a primitive sorting algorithm with quadratic running time 
in both worst and average case.
On very small arrays however, it is extremely efficient,
which makes it a good choice for our purpose.

Our implementation of \proc{InsertionSort} is given as
Algorithm~\ref{alg:insertionsort-knuth} and its control flow graph is shown in
Figure~\ref{fig:control-flow-graph-insertionsort}.
Algorithm~\ref{alg:insertionsort-knuth} is based on
the implementation by \citeN[Program S]{Knuth1998}. 
Knuth assumes $n\ge 2$ in his code and analysis, but our
Quicksort implementation also calls \proc{InsertionSort} on subproblems of
size~$0$ or~$1$. 
Therefore, Algorithm~\ref{alg:insertionsort-knuth} starts with an index
comparison “$i\le \id{right}$” to handle these cases.

\begin{figure}
 	\mbox{}\hfill%
 	\begin{tikzpicture}[
	node distance=5mm,
	shorten >=.75pt,
	every node/.style={font={\scriptsize},inner sep=2.5pt},
	basic block/.style={
		fill=black!5,
		draw,
		shape=rectangle split,
		rectangle split	parts=2,
		rounded corners=2pt,
	}
]

\newcommand{\block}[4][2cm]{%
	\setcounter{basicblocknumberIS}{#2}%
	\addtocounter{basicblocknumberIS}{-1}%
	\refstepcounter{basicblocknumberIS}%
	\label{bb:insertionsort-#2}%
	\makebox[#1]{\,\textbf{\ref*{bb:insertionsort-#2}}%
	\:\hfill%
	\smash{$#3$}{\tiny$\vphantom{p}$}
	}%
	\nodepart{two}\parbox{#1}{%
		\centering \pbox{#1}{%
			\setcounter{basicblocknumberIS}{#2}%
			\addtocounter{basicblocknumberIS}{-1}%
			\refstepcounter{basicblocknumberIS}%
			#4%
		}
	}%
}
\node[basic block] (b2a) at (0,0) 
	{\block{1}{I}{%
		$i \gets \id{left} + 1$;
		\label{bb:insertionsort-init-i}
	}} ;
\node[basic block,below=of b2a] (b2b)
	{\block{2}{G}{%
		$i \le \id{right}$
		\label{bb:insertionsort-outer-loop}
	}} ;
\node[basic block,below=7.5mm of b2b] (b2c)
	{\block{3}{G-I}{%
		$j \gets i-1$; \\
		$v \gets \arrayA[i]$;
		\label{bb:insertionsort-init-v}
	}} ;
\node[basic block,right=of b2c] (b2d)
	{\block{4}{E+D}{%
		$v < \arrayA[j]$
		\label{bb:insertionsort-comp}
	}} ;
\node[basic block,right=of b2d] (b2e)
	{\block{5}{E}{%
		$\arrayA[j+1] \gets \arrayA[j]$; \\
		$j \gets j-1$; \\
		$j < \id{left}$
		\label{bb:insertionsort-write-1}
	}} ;
\node[basic block,below=of b2e] (b2f)
	{\block{6}{G-I-D}{%
		Goto;
		\label{bb:insertionsort-goto}
	}} ;
\node[basic block,below=20mm of b2d] (b2g)
	{\block{7}{G-I}{%
		$\arrayA[j+1] = v$; \\
		$i \gets i+1$;
		\label{bb:insertionsort-write-2}
	}} ;
\node[basic block,left=of b2g] (b2h)
	{\block{8}{I}{%
		Return
		\label{bb:insertionsort-return}
	}} ;
\coordinate[above left=10mm and 5mm
 of b2h] (alb2h) {};
\coordinate[above=5mm of b2e] (ab2e) {};
\coordinate (ab2g) at (b2f -| b2g) ;
\coordinate[right=5mm of b2e] (rb2e) {};
\begin{scope}[->,auto,every node/.style={font={\tiny}}]
\draw (0,1) -- (b2a) ;
\draw (b2a) -- (b2b) ;
\draw (b2b) -- node {yes} (b2c) ;
\draw (b2b) -| node[pos=.6] {no} (alb2h) -| (b2h) ;
\draw (b2c) -- (b2d) ;
\draw (b2d) -- node {yes} (b2e) ; 
\draw (b2d) -- node[pos=.2] {no} (b2g) ;
\draw (b2e) -- node {yes} (b2f) ;
\draw[thick] (b2e) -- node {no} (ab2e) -| (b2d) ;
\draw[-,shorten >=0pt] (b2f) -- (ab2g) ;
\draw[thick] (b2g) -| (rb2e) |- (b2b) ;
\draw (b2h) -- ++(0,-1) ;
\end{scope}
\end{tikzpicture}%
 	\hfill\mbox{}
 	\caption{%
		Control flow graph of our Java implementation of \proc{InsertionSort}
		(Algorithm~\ref{alg:insertionsort-knuth}). 
		The block names “2a”\,--\,“2h” indicate that these blocks
		replace block~\ref{bb:call-insertionsort} in Figure~\ref{fig:control-flow-graph};
		this figure provides a “close-up view” of block~\ref{bb:call-insertionsort}. 
		Blocks~\ref{bb:insertionsort-goto} and~\ref{bb:insertionsort-return} contain
		control flow statements needed in Java Bytecode, which would normally be
		represented by arrows only.
		They are shown to remind us of their cost contributions.
		\label{fig:control-flow-graph-insertionsort}
 	}%
\end{figure}
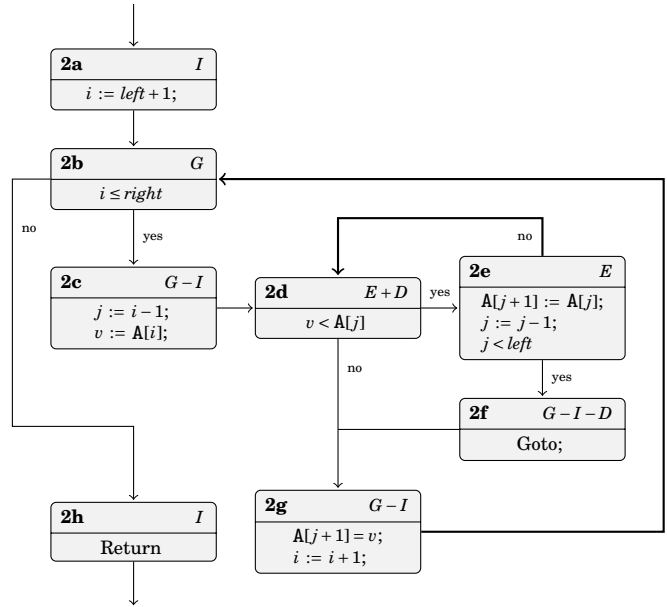

Figure~\ref{fig:control-flow-graph-insertionsort} lists the execution
frequencies of all basic blocks. 
The names are chosen to match the corresponding notation of
\citeN{Sedgewick1977} and denote the \emph{total} execution frequencies across
all invocations of \proc{InsertionSort} on small subproblems caused by one
initial call to \proc{QuicksortYaroslavskiy}.
Let us define $\tilde I$, $\tilde G$, $\tilde D$ and $\tilde E$ to denote the
frequencies when we use \proc{InsertionSort} \emph{in isolation} for sorting a
random permutation.
These frequencies are analyzed by \citeN[p.\,82]{Knuth1998} for $n\ge2$. 
As mentioned above, our implementation has to work for $n\ge0$, so our
analysis must take the special cases $n\in\{0,1\}$ into account.
We find
\begin{align*}
	\tilde I(n)		&\wwrel= 	1	\,, &
	\tilde	G(n)	&\wwrel= 	n + \delta_{n0}		\,, &
	\tilde	D(n)	&\wwrel=	n - \harm n		\qquad\text{and} &
	\tilde	E(n)	&\wwrel=	\tbinom n2 \big/ 2		\;. 
\end{align*}
We can compute $I$, $G$, $D$ and $E$ by inserting $\tilde I$, $\tilde G$,
$\tilde D$ and $\tilde E$ for $\insertsortcost$
in the solution provided by Proposition~\ref{pro:Cn-explicit-linear-pc} on
page~\pageref{pro:Cn-explicit-linear-pc}.
\begin{equation}
	\begin{aligned}	I(n)	&\wwrel= 	\tfrac15 (n+1) \frac1{\tbinom{M+2}3} \sum_{k=0}^M
							(3M-2k) \, \tilde I(k)
			 \wwrel=	\tfrac{12}{5(M+2)} (n+1) \,, \\
	G(n)	&\wwrel=	\bigl(1 + \tfrac{18}{5(M+1)} - \tfrac{6}{M+2} \bigr) (n+1) \,,\\
	D(n)	&\wwrel=	\bigl(1 + \tfrac{3}{5(M+2)} 
								- \tfrac{12}{5(M+2)} \harm{M+1} 
						\bigr) (n+1)\,, \\
	E(n)	&\wwrel=	\bigl( \tfrac3{20}M + \tfrac6{5(M+2)} 
							- \tfrac{11}{20} 
						\bigr)	(n+1)\;.
	\end{aligned}
\end{equation}
Using these frequencies, we can easily express the expected number of
key comparisons, write accesses and executed Bytecode instructions:
The only place where key comparisons occur, is in
block~\ref{bb:insertionsort-comp}, so 
$\totalcomparisonmarker{\mathit{IS}} = D+E$.
Write accesses to array \arrayA happen in blocks~\ref{bb:insertionsort-write-1}
and~\ref{bb:insertionsort-write-2}, giving
$\totalwritesmarker{\mathit{IS}} = E + (G-I)$.
The number of Bytecodes is given in the next section.

\section{Low-Level Implementations and Instruction Counts}
\label{app:low-level}

\begin{figure}
\begin{lstlisting}
void quicksortYaroslavskiy(int[] A, int left, int right) {
	if (right - left < M) {
		insertionsort(A, left, right);
	} else {
		final int p, q;
		if (A[left] > A[right]) {
			p = A[right]; q = A[left];
		} else {
			p = A[left]; q = A[right];
		}
		int l = left + 1, g = right - 1, k = l;
		while( k <= g ) {
			final int ak = A[k];
			if (ak < p) {
				A[k] = A[l]; A[l] = ak; ++l;
			} else if (ak >= q) {
				while (A[g] > q && k < g)
					--g;
				if (A[g] < p) {
					A[k] = A[l]; A[l] = A[g]; ++l;
				} else {
					A[k] = A[g];
				}
				A[g] = ak; --g;
			}
			++k;
		}
		--l; ++g;
		A[left] = A[l]; A[l] = p; A[right] = A[g]; A[g] = q;
		quicksortYaroslavskiy(A, left, l - 1);
		quicksortYaroslavskiy(A, g + 1, right);
		quicksortYaroslavskiy(A, l + 1, g - 1);
	}
}

void insertionsort(int[] A, int left, int right) {
	for (int i = left + 1; i <= right; ++i) {
		int j = i-1; final int v = A[i];
		while (v < A[j]) {
			A[j+1] = A[j];	--j;
			if (j < left) break;
		}
		A[j+1] = v;
	}
}
\end{lstlisting}
\caption{%
	Java implementation of Yaroslavskiy's algorithm.%
} 
\label{fig:java-implementation}
\end{figure}

Figure~\ref{fig:java-implementation} shows the Java implementation of
Yaroslavskiy's algorithm whose Bytecode counts are studied in this paper.
The partitioning loop is taken from the
original sources of the Java 7 Runtime Environment library
(see for example
\url{http://www.docjar.com/html/api/java/util/DualPivotQuicksort.java.html}).

The Java code has been compiled using Oracle's Java Compiler 
(javac version 1.7.0\_17). The resulting Java Bytecode was decomposed into
the basic blocks of Figures~\ref{fig:control-flow-graph}
and~\ref{fig:control-flow-graph-insertionsort}.
Then, for each block the number of Bytecode instructions was counted, the
result is given in Table~\ref{tab:block-bytecodes-yaroslavskiy}.
We have automated this process as part of our tool MaLiJAn (Maximum
Likelihood Java Analyzer), which provides a means of automating empirical
studies of algorithms based on their control flow
graphs~\cite{Wild2013Alenex,Laube2010}.

\begin{table}
\tbl{%
	Bytecode counts for the basic blocks of Figures~\ref{fig:control-flow-graph}
	and~\ref{fig:control-flow-graph-insertionsort}.
	\label{tab:block-bytecodes-yaroslavskiy}
}{%
	\def\b#1{\ref{bb:#1}}
	\def\i#1{\ref{bb:insertionsort-#1}}
	\begin{tabular}{r*{11}{r}}
		\toprule
		\textbf{Block}         & \b1 & \b3 & \b4 & \b5 & \b6 & \b7 & \b8 & \b9 & \b{10} & \b{11} \\[1pt]
		\textbf{\#\,Bytecodes} & 5   & 7   &   8 &   9 &  10 &   3 &   7 &  12 &     3  &    5   \\
		\midrule
		\textbf{Block}         & \b{12} & \b{13} & \b{14} & \b{15} & \b{16} & \b{17} & \b{18} & \b{19} & \b{20} \\[1pt]
		\textbf{\#\,Bytecodes} &     3  &     2  &     5  &     6  &    14  &     5  &     2  &    42  &     1  \\
		\midrule
		\addlinespace[1ex]
		\textbf{Block}         & \i1 & \i2 & \i3 & \i4 & \i5 & \i6 & \i7 & \i8  \\[1pt]
		\textbf{\#\,Bytecodes} &   8 &   3 &   8 &   5 &  12 &   1 &   8 &   2  \\
		\bottomrule
	\end{tabular}
}
\end{table}

By multiplying the Bytecodes per block with the block's frequency, we get
the overall number of executed Bytecodes. 
For Yaroslavskiy's Quicksort, we get
{
\let\C\totalcomparisonmarker
\let\S\totalswapmarker
\begin{align}
		\bytecodes{\mathit{QS}}
	&\wwrel=
		5R + 7A + 8B + 9(A-B) + 10A + 3(A+\C1) + 7\C1 + 12\S1 + 3(\C1-\S1)
	\nonumber\\*&\wwrel\ppe\phantom{5R}
		   + 5\C3 + 3(\C3-\C4+F) + 2(\C3-\C4) + 5\C4 + 6(\C4-\S3) + 14\S3 
	\nonumber\\*&\wwrel\ppe\phantom{5R}
		   + 5\C4 + 2\C1 + 42A + 1R
\nonumber
\\*	&\wwrel=	 71 A
				- 1 B 
				+ 6 R
				+ 15 \totalcomparisonmarker{1}
				+ 10 \totalcomparisonmarker{3}
				+ 11 \totalcomparisonmarker{4}
				+  9 \totalswapmarker{1}
				+  8 \totalswapmarker{3}
				+  3 F
	\;.
\end{align}
}
Additionally, we have for \proc{InsertionSort}
\begin{align}
		\bytecodes{\mathit{IS}} 
	&\wwrel= 
		8 I + 3 G + 8 (G-I) + 5(E+D) + 12 E + 1(G-I-D) + 8 (G-I) + 2I
\nonumber	\\
	&\wwrel=	4 D + 17 E + 20 G	- 7 I \;.
\end{align}

\smallskip
Note that \citeN{wild2012thesis} investigated a different (more naïve) Java
implementation of Yaroslavskiy's algorithm and hence reports different Bytecode
counts.

\section{Details on the Distributional Analysis}
\label{app:contraction-proofs}

In this appendix, we give details on the application of the contraction method
to the distributions of costs in Yaroslavskiy's algorithm and we prove the main
technical lemmas from Section~\ref{sec:asymptotics-lemmas}.

\subsection{Proof of Theorem~\ref{thm:limiting-dist-comparisons}}
We consider the first partitioning step of Yaroslavskiy's
algorithm and denote by $\comparisonmarker[n]{}$ the number of key comparisons
of the first partitioning phase.
By Property~\ref{pro:randomness-preservation}, subproblems generated in
the first partitioning phase are, conditional on their sizes,
again uniformly random permutations and independent of each other.
Hence, we obtain the distributional recurrence
\begin{align}
\label{eq:distributional-recurrence-comparisons}
		\totalcomparisonmarker[n]{}
	\wwrel\eqdist
		C'_{I_1} + C''_{I_2} + C'''_{I_3} \wbin+
		\comparisonmarker[n]{} \qquad (n\ge 3),
\end{align}
where $(I_1,I_2,I_3,\comparisonmarker[n]{})$, $(C'_j)_{j\ge 0}$,
$(C''_j)_{j\ge 0}$, $(C'''_j)_{j\ge 0}$ are independent and
$C'_j$, $C''_j$, $C'''_j$ are identically distributed as
$C_j$ for $j\ge 0$.
Note that equation~\eqref{eq:distributional-recurrence-comparisons} is
simply obtained from the generic distributional
recurrence~\eqref{eq:distributional-recurrence-generic} upon inserting the toll
function $\comparisonmarker[n]{}$.

As in Section~\ref{sec:distributional-recurrence}, we now define the normalized
number of comparisons $\tCsn$ as 
\begin{align}
\label{eq:def-normalized-comparisons}
	C^*_0 \wrel\ce 0 
	\qquad\qquad\text{and}\qquad\qquad
	C^*_n &\wwrel\ce
		\frac{C_n-\E[C_n]}{n}\,,
		\qquad (n\ge 1)
\end{align}
Note that $\E[\tCsn] = 0$ and the $\V(\tCsn) < \infty$, i.\,e.,
$(\tCsn)_{n\ge 0}$ is a
sequence of centered, square integrable random variables.
Using  \eqref{eq:distributional-recurrence-comparisons} we
find, cf.~\cite[eq.~(27), (28)]{hwne02},
that $(\tCsn)_{n\ge 0}$  satisfies~\eqref{eq:general-contraction-recurrence}
with
\begin{align*}
	\ui{A}{n}_r &\wwrel= \frac{I_r}{n}, &
	\ui{b}{n}   &\wwrel=
		\frac{1}{n}\biggl(\comparisonmarker[n]{} - \E[C_n]
			\wbin+ \sum_{r=1}^3 \E[C_{I_r}\given I_r]
		\biggr).
\end{align*}

We apply the framework of the contraction method outlined in
Section~\ref{sss:contraction-method}.
To check condition~\ref{cond:convergence-of-recurrence-coeffs} note that from
\eqref{eq:asymptotic-I}, we have $\ui{A}{n}_r \to D_r$ in $L_2$ for $r=1,2,3$
as $n\to\infty$.

To identify the $L_2$-limit of $\ui bn$ we look  at the summands
$\comparisonmarker[n]{} \mathbin/ \! n$ and 
$\bigl(- \E[C_n]\wbin+ \sum_{r=1}^3 \E[C_{I_r}\given I_r]\bigr) \mathbin/ n$
separately. 
By Theorem~\ref{thm:expected-comparisons}, the expectation has the form
$\E[C_n] = 19/10 \: n\ln n + cn +\bo(\log n)$
for some constant $c\in\Rset$, which implies 
$\E[C_{I_r}\given I_r] \rel= 19/10 \: I_r\ln I_r + c I_r + o(n)$ 
since we have $o(I_r)=o(n)$. 
Plugging in these expansions, using $I_1+I_2+I_3=n-2$ and rearranging terms
gives the asymptotic identity, as $n\to\infty$,
\begin{align*}
	\frac{1}{n}\biggl(
			- \E[C_n] + \sum_{r=1}^3 \E[C_{I_r}\given I_r]
		\biggr)
  	&\wwrel= 
			\sum_{r=1}^3 \tfrac{19}{10} \tfrac{I_r}n \ln \tfrac{I_r}n	
				\wwbin- \tfrac2n\, \tfrac{19}{10} \tfrac{n\ln n}{n} 
				\wbin- \tfrac{2c}{n}
				\wbin+ o(1)
\\	&\wwrel=
			\sum_{r=1}^3 \tfrac{19}{10} \tfrac{I_r}n \ln \tfrac{I_r}n	
				\wbin+ o(1)
				\;.
\end{align*}
Hence, Lemma~\ref{lem:asymptotics-xlnx} implies
\begin{align}
\label{eq:convergence-toll-comparisons}
		\frac{1}{n}\biggl(
			- \E[C_n] + \sum_{r=1}^3 \E[C_{I_r}\,|\, I_r]
		\biggr)
	&\wwrel{\stackrel{L_2}{\longrightarrow}}
		\tfrac{19}{10}\sum_{j=1}^3  D_j \ln D_j
		\qquad (n\to \infty)\;.
\end{align}
For the limit behavior of
$\comparisonmarker[n]{}\mathbin/n$ we 
 use the distributions listed in Table~\ref{tab:distribution-of-tolls} and  find
\begin{align*}
		\comparisonmarker[n]{}
	&\wwrel=	  		\comparisonmarker[n]{1}
				\wbin+ \bigl(\comparisonmarker[n]{1} - \swapmarker[n]{1}\bigr)
				\wbin+ \comparisonmarker[n]{3}
				\wbin+ \comparisonmarker[n]{4}
				\wbin+ \toll[n]A
\\	&\wwrel\eqdist
				  		n-1
				\wbin+ I_1
				\wbin+ I_2
				\wbin+ \hypergeometric(I_1+I_2, I_3, n-2)
				\wbin- \hypergeometric(I_1, I_1+I_2, n-2)
				\wbin+ 3\bernoulli(\tfrac{I_3}{n-2})
	\;.
\end{align*}
Using Lemma~\ref{lem:asymptotics-hypergeometric} and~\eqref{eq:asymptotic-I}, we
find for the normalized number of comparisons:
\begin{align*}
		\frac{\comparisonmarker[n]{}}{n}
	&\wwrel\eqdist
				  		\frac{n-2}n
				\wbin+ \frac{I_1}n
				\wbin+ \frac{I_2}n
				\wbin+ \frac{\hypergeometric(I_1+I_2, I_3, n-2)}{n-2}\cdot\frac{n-2}{n}
\\*	&\wwrel\ppe\quad{}
				\wbin- \frac{\hypergeometric(I_1, I_1+I_2, n-2)}{n-2}\cdot\frac{n-2}{n}
				\wbin+ 3\frac{\bernoulli\bigl(\tfrac{I_3}{n-2}\bigr)}{n}
\\	&\wwrel{\overset{L_2}\longrightarrow}
						1
				\wbin+ D_1
				\wbin+ D_2
				\wbin+ (D_1+D_2)D_3
				\wbin- D_1(D_1+D_2)
				\wbin+ 0
\\	&\wwrel=	1 \wbin+ (D_1+D_2)(1 + D_3 - D_1)				
\\	&\wwrel=	1 \wbin+ (D_1+D_2)(D_2 + 2D_3)
	\;.
\end{align*}

Altogether, we obtain that
condition~\ref{cond:convergence-of-recurrence-coeffs} holds with
\begin{align*}
 	(A_1,A_2,A_3,b) &\wwrel=
 		\biggl(
 			D_1,\,
 			D_2,\,
 			D_3,\;
 			1 +	(D_1+D_2)(D_2+2D_3)
 			  + \tfrac{19}{10}\sum_{j=1}^3  D_j \ln D_j
 		\biggr).
\end{align*}
Concerning condition~\ref{cond:recurrence-contraction}, note 
that $D_1$, $D_2$ and $D_3$ are identically distributed with 
density $x\mapsto 2(1-x)$ for $0\le x\le 1$. 
This implies
\begin{align}\label{nein_1}
		\sum_{r=1}^3 \E\bigl[D_r^2\bigr] 
	\wwrel= 
		3\E\bigl[D_1^2\bigr] 
	\wwrel=
		3\int_0^1 \mkern-5mu x^2 \; 2(1-x) \: dx
	\wwrel= 
		\tfrac12
	\wwrel<
		1\;.
\end{align}
Moreover, condition \ref{cond:recurrence-small-sublists-seldom} is fulfilled
since
\begin{align*}
		\sum_{r=1}^K \E\left[
			{\bf1}_{\{\ui{I}{n}_r \le \ell\}} \cdot
			\|(\ui{A}{n}_r)^t \ui{A}{n}_r\|_\mathrm{op}
		\right]
	&\wwrel\le
		\sum_{r=1}^K \Prob(\ui{I}{n}_r \le \ell)
	\wwrel\to 0,
	\qquad(n\to\infty),\quad\text{for any fixed $\ell\in\Nset$}\;.
\end{align*}

\smallskip
Now the conclusions \ref{res:unique-fixpoint} and
\ref{res:convergence-of-coeffs-gives-fixpoint} give the claims
$\tCsn \to C^*$ in distribution with the characterization of
the distribution of $C^*$.
For the asymptotic of the variance note that convergence of the second  moment
$\E[(\tCsn)^2] \to \E[(C^*)^2]$ and the
normalization~\eqref{eq:def-normalized-comparisons}~imply
\begin{align*}
	\V(C_n) &\wwrel\sim \sigma_C^2 n^2
	\qquad\text{with}\qquad
	\sigma_C^2 \wwrel= \E[(C^*)^2]\;.
\end{align*}
To identify $\sigma_C^2$, let $C^{*,(1)}, C^{*,(2)}$ and $C^{*,(3)}$ 
be independent copies of $C^*$ also independent of $(D_1,D_2,D_3)$. 
We abbreviate 
$\tau:=1 + (D_1+D_2)(D_2+2D_3) + \tfrac{19}{10}\sum_{r=1}^{3} D_r \ln D_r$.
Taking squares and expectations in \eqref{eq:fix-point-eq-comparisons}
and noting that $\E[C^*] = \E[C^{*,(r)}] = 0$, we find
	\begin{align*}
			\E[(C^*)^2]
		&\wwrel=
			\E\Bigl[\Bigl(
					1 + (D_1+D_2)(D_2+2D_3)
					  + \tfrac{19}{10}\sum_{r=1}^{3} D_r \ln D_r 
				\wbin+ \sum_{r=1}^3 D_r  C^{*,(r)}
				\Bigr)^2\Bigr]
	\\	&\wwrel= \E[\tau^2] \wbin+ 2\E\Bigl[ \tau \sum_{r=1}^3 D_r C^{*,(r)} \Bigr]
				  \wbin+ \E\Bigl[ \Bigl( \sum_{r=1}^3 D_r  C^{*,(r)} \Bigr)^2 \Bigr]
	\\	&\wwrel= \E[\tau^2]
				\wbin+ 2\sum_{r=1}^3 \E[\tau D_r] \E[C^*]
				\wbin+ \sum_{r=1}^3 \E[D_r^2] \E[(C^*)^2]
				\wbin+ \; \sum_{\substack{\mathclap{1\le,r,s\le 3} \\ r\ne s}} \;
							\E[D_r D_s] \E[C^*]^2
	\\[-1.5ex]	&\wwrel= \E[\tau^2]
				\wbin+ \E[(C^*)^2] \sum_{r=1}^3 \E[D_r^2] \\ &\wwrel=\E[\tau^2] + \tfrac{1}{2}\E[(C^*)^2]
		\;,
	\end{align*}
where for the last equality (\ref{nein_1}) was used. 
Solving for $\E[(C^*)^2]$ implies
\begin{align*}
	\E[(C^*)^2] =
		2 \E \biggl[\Bigl(
			1 + (D_1+D_2)(D_2+2D_3)
			  + \tfrac{19}{10} \sum_{j=1}^3  D_j \ln D_j
		\Bigr)^{2} \biggr].
\end{align*}
Now, the integral representation \eqref{eq:Dirichlet-integral} and the use of a computer
algebra system yields the expression for $\sigma_C^2$.
\qed

\subsection{Proof of Theorem~\ref{thm:limiting-dist-swaps}}
The proof can be done similarly to the one for
Theorem~\ref{thm:limiting-dist-comparisons}. We have the recurrence
\begin{align}
\label{eq:distributional-recurrence-swaps}
		S_n
	&\wwrel\eqdist S'_{I_1} + S''_{I_2} + S'''_{I_3}
						\wbin+ \swapmarker[n]{},
	\qquad (n\ge 3),
\end{align}
with conditions on independence and distributions as in
\eqref{eq:distributional-recurrence-comparisons}, where $\swapmarker[n]{}$
is the number of swaps in the first partitioning step of the algorithm.
We set \smash{$S^*_0 \ce 0$} and
\begin{align}
\label{eq:def-normalized-swaps}
 		S^*_n
	&\wwrel\ce
		\frac{S_n - \E[S_n]}{n}\,,
		\qquad (n\ge 1)\;.
\end{align}
Hence, $(\tSsn)_{n\ge 0}$ is a sequence of centered, square integrable
random variables satisfying~\eqref{eq:general-contraction-recurrence} with
\begin{align*}
	\ui{A}{n}_r = \frac{I_r}{n},\qquad
	\ui{b}{n}   = \frac{1}{n}\biggl(
						\swapmarker[n]{} - \E[S_n] + \sum_{r=1}^3 \E[S_{I_r}\,|\,I_r]
					\biggr),
\end{align*}
where by Theorem~\ref{thm:expected-swaps}, we know 
$\E[S_n]=\frac35 n\ln(n) + c'n + \bo(\log n)$ for
a constant $c'\in\Rset$.
Analogously to~\eqref{eq:convergence-toll-comparisons} we obtain
\begin{align}
\label{eq:convergence-toll-swaps}
		\frac{1}{n}\biggl(
			- \E[S_n] + \sum_{r=1}^3 \E[S_{I_r}\,|\,I_r]
		\biggr)
	&\wwrel\to
		\tfrac{3}{5}\sum_{j=1}^3  D_j \ln D_j ,
		\qquad (n\to \infty),
		\quad\text{in }L_2 \;.
\end{align}

\smallskip
It remains to study the asymptotic behavior of $\swapmarker[n]{}\mathbin/n$.
Again profiting from the spadework of Section~\ref{sec:expected-frequencies},
we find the exact distribution of the number of swaps:
\begin{align*}
		\swapmarker[n]{}
	&\wwrel=			\swapmarker[n]{1}
				\wbin+ 	\bigl( \comparisonmarker[n]{4} - \swapmarker[n]{3} \bigr)
				\wbin+ 2\swapmarker[n]{3}
				\wbin+ 2\toll[n]A
	 \wwrel= I_1 + \numberat{sm}{G} + \delta + 2
\\	&\wwrel\eqdist		I_1
				\wbin+	\hypergeometric(I_1+I_2, I_3, n-2)
				\wbin+	\bernoulli\bigl( \tfrac{I_3}{n-2} \bigr)
				\wbin+ 2
	\;.
\end{align*}
By Lemma~\ref{lem:asymptotics-hypergeometric} and~\eqref{eq:asymptotic-I}, we
find
\begin{align*}
		\frac{\swapmarker[n]{}}n
	&\wwrel{\overset{L_2}\longrightarrow}
		D_1 + (D_1+D_2)D_3
		\;.
\end{align*}
The conditions \ref{cond:convergence-of-recurrence-coeffs},
\ref{cond:recurrence-contraction}
and~\ref{cond:recurrence-small-sublists-seldom}
are now checked as in the proof of Theorem~\ref{thm:limiting-dist-comparisons}.
The assertions of Theorem \ref{thm:limiting-dist-swaps} follow from
\ref{res:unique-fixpoint} and \ref{res:convergence-of-coeffs-gives-fixpoint},
the identification of $\sigma^2_S$ is done as that of
$\sigma^2_C$ in Theorem \ref{thm:limiting-dist-comparisons}.
\qed

\subsection{Proof of Theorem~\ref{thm:limiting-dist-bytecodes}}
The proof can be done very similarly as for
Theorems~\ref{thm:limiting-dist-comparisons} and~\ref{thm:limiting-dist-swaps}.
We only present the key points where changes are needed.
For the distributional recurrence, we here have the toll function
\begin{align*}
		\toll[n]{\bytecodes{}}
	&\wwrel=	
				  15 \comparisonmarker[n]{1}
				+ 10 \comparisonmarker[n]{3}
				+ 11 \comparisonmarker[n]{4}
				+  9 \swapmarker[n]{1}
				+  8 \swapmarker[n]{3}
\\* &\wwrel\ppe\quad {}				
				+ 71 \toll[n]A
				- 1 \toll[n]B
				+ 6 \toll[n]R
				+  3 \toll[n]F \;.
\end{align*}
All contributions from the second line are bounded by $\bo(1)$ (see
Table~\ref{tab:distribution-of-tolls}).
Therefore they vanish in the limit of $\toll[n]{\bytecodes{}} \mathbin/ n$:
\begin{align*}
		\frac{\toll[n]{\bytecodes{}}}n
	&\wwrel{\overset{L_2}\longrightarrow}
				15(D_1+D_2)
		\bin+ 	10D_3
		\bin+	11(D_1+D_2)D_3
		\bin+	 9D_1(D_1+D_2)
		\bin+	 8\bigl( D_1 - D_1(D_1+D_2) \bigr)
\\	&\wwrel=
		24 + (D_3-9)D_2 - 2D_3(5D_3 + 2) \;.
\end{align*}
The rest of the proof is carried out along the same lines as for the
proofs above.
\qed

\subsection{Proof of Theorem~\ref{thm:limit-covariance}}
We define the column vector $Y_n \ce \bigl({C_n \atop S_n}\!\bigr)$. 
Then from
\eqref{eq:distributional-recurrence-comparisons} and
\eqref{eq:distributional-recurrence-swaps}, we obtain
\begin{align*}
	Y_n &\wwrel\eqdist
		Y'_{I_1} + Y''_{I_2} + Y'''_{I_3}
		\wbin+
		\begin{pmatrix}
			\comparisonmarker[n]{} \\
			\swapmarker[n]{}
		\end{pmatrix}
	\,,
\end{align*}
with conditions on independence and distributions as in
\eqref{eq:distributional-recurrence-comparisons}.
We set $Y^*_0 \ce \bigl({0 \atop 0}\bigr)$ and
\begin{align}
\label{eq:def-normalized-CS-vector}
	Y^*_n &\wwrel\ce
		\frac1n
		\bigl( Y_n - \E[Y_n] \bigr) \;.
\end{align}
Hence, $(\tYsn)_{n \ge 0}$ is a sequence of centered, square integrable,
bivariate random variables satisfying~\eqref{eq:general-contraction-recurrence}
with
\begin{align*}
	\ui{A}{n}_r &\wwrel= \frac1n \begin{bmatrix}
							I_r & 0 \\ 0 & I_r
					  \end{bmatrix}
		\,, &
	\ui{b}{n} 	&\wwrel= \frac1n \begin{pmatrix}
							  \comparisonmarker[n]{}
							- \E[\totalcomparisonmarker[n]{}]
							+ \sum_{r=1}^3	\E[\totalcomparisonmarker[I_r]{}\given I_r] \\[.5ex]
							  \swapmarker[n]{}
							- \E[\totalswapmarker[n]{}]
							+ \sum_{r=1}^3	\E[\totalswapmarker[I_r]{} \given I_r]
					  \end{pmatrix}
		\;.
\end{align*}
Using the asymptotic behavior from the proofs of Theorems
\ref{thm:limiting-dist-comparisons}  and \ref{thm:limiting-dist-swaps} we
obtain that condition \ref{cond:convergence-of-recurrence-coeffs} holds with
\begin{align}
\label{eq:covariance-fix-point-coefficients}
	A_r	= 	\begin{bmatrix}
				D_r& 0 \\ 0 & D_r
		  	\end{bmatrix}
	\,, \qquad
	b	=	\begin{pmatrix}
 				1 + (D_1+D_2)(D_2+2D_3)
	 			{} + \tfrac{19}{10}\sum_{j=1}^3  D_j \ln D_j
 			\\[.5ex]
				\like{
					1 + (D_1+D_2)(D_2+2D_3)
				}{
					D_1 + (D_1+D_2)D_3
	 			}
				  + \like{\tfrac{19}{10}}{\tfrac{3}{5}}
				  	\sum_{j=1}^3  D_j \ln D_j
			\end{pmatrix}
	\;.
\end{align}
Condition \ref{cond:recurrence-contraction} is satisfied as
$$
			\sum_{r=1}^3 \E\left[ \| A_r^tA_r\|_\mathrm{op}\right]
	\;=\; 	\sum_{r=1}^3 \E[D_r^2]
	\;=\;	\tfrac{1}{2}
	\;<\;	1 \;.
$$
Condition~\ref{cond:recurrence-small-sublists-seldom} is checked similarly as in
the proof of Theorem~\ref{thm:limiting-dist-comparisons}.

Hence from~\ref{res:unique-fixpoint} we obtain the existence of a centered,
square integrable, bivariate distribution ${\cal L}(\Lambda_1,\Lambda_2)$ that
solves the bivariate fixed-point equation \eqref{eq:general-fix-point-eq} with
the choices for $A_r$ and $b$ given
in~\eqref{eq:covariance-fix-point-coefficients}.
Furthermore~\ref{res:convergence-of-coeffs-gives-fixpoint} implies that the
sequence $(\tYsn)$ defined in~\eqref{eq:def-normalized-CS-vector}
converges in distribution and with mixed second moments towards
$(\Lambda_1,\Lambda_2)$. 
This implies in particular, as $n\to\infty$,
\begin{align*}
			\E\left[\frac{C_n-\E[C_n]}{n} \cdot \frac{S_n-\E[S_n]}{n} \right]
	&\wwrel\to	\E[\Lambda_1 \cdot \Lambda_2].
\end{align*}
Hence, we obtain
\begin{align*}
	\Cov(C_n,S_n) &\wwrel\sim \E[\Lambda_1\Lambda_2] \, n^2.
\end{align*}
The value $\E[\Lambda_1\Lambda_2]$ is obtained from the fixed-point equation
\eqref{eq:general-fix-point-eq} with the choices for $A_r$ and $b$ given in
\eqref{eq:covariance-fix-point-coefficients} by multiplying the components on
left and right hand side, taking expectations and solving for
$\E[\Lambda_1\Lambda_2]$.
The integral representation \eqref{eq:Dirichlet-integral} then leads to the
expression given in \eqref{corcs}.
\qed

\subsection{Proof of Lemma~\ref{lem:asymptotics-hypergeometric}}
We denote by $\binomial(n,p)$ the binomial distribution with $n$ trials
and success probability $p$.

The hypergeometric distribution $\hypergeometric(k,r,r+b)$ has mean and variance
given in \eqref{eq:mean-variance-hypergeometric}.
In particular for sequences $(\alpha_n)_{n\ge 1}$, $(\beta_n)_{n\ge 1}$ with
$\alpha_n = \alpha\,n + r_n$ and 
$ \beta_n = \beta\, n + s_n$ with 
$\alpha, \beta\in (0,1)$ and $|r_n|, |s_n| \le n^{2/3}$, 
we obtain for hypergeometrically $\hypergeometric(\alpha_n,\beta_n,n)$ distributed
random variables $\Upsilon_n$ that $\V(\Upsilon_n)\le n$ and moreover
\begin{align}
	\left\|\frac{\Upsilon_n}{n} - \alpha\beta\right\|_2 
		&\wwrel\le		\left\|\frac{\Upsilon_n}{n} - \frac{\E \Upsilon_n}{n}\right\|_2 + \left|\frac{\E \Upsilon_n}{n} - \alpha\beta\right|\nonumber\\		
		  &\wwrel= 		\frac{\sqrt{\V(\Upsilon_n)}}{n} + \left|\frac{\alpha_n\beta_n}{n^2} - \alpha\beta\right|\nonumber\\
		  &\wwrel\le n^{-1/2} + \frac{\alpha |s_n|}{n} + \frac{\beta |r_n|}{n}+\frac{ |r_n s_n|}{n^2}
				\nonumber\\ 
		& \wwrel\le n^{-1/2} + 2 n^{-1/3} + n^{-2/3} \wwrel\le 4 n^{-1/3} \;.
	\label{eq:eshy}
\end{align}
Now, we first condition on 
$V \ce (V_1,\ldots,V_b) = (v_1,\ldots,v_b) \equalscolon v$, where $v$ satisfies
$v_i\in [\.0,1]$ and $\sum_{r=1}^b v_r = 1$.
Conditionally on $V=v$, the random variables $\sum_{j\in J_i} L_j$ have a
binomial $\binomial(n,w_i)$ distribution,
where $w_i \ce \sum_{j\in J_i} v_j$ for $i=1,2$.
We denote
\begin{align*}
	B_n^v &\wwrel\ce 
	\bigcap_{i=1}^2
	\biggl\{
		\sum_{j\in J_i} L_j \in
		\Bigl[n \, w_i - n^{2/3} ,\; n \, w_i + n^{2/3} \Bigr]
	\biggr\}\;.
\end{align*}
Then by Chernoff's
bound \cite[p.\,195]{McDi98}, we obtain uniformly
in $v$ that 
\begin{align}
			\Prob(B_n^v)
		&\wwrel\ge 1 - 4 \exp\bigl(-2n^{1/3}\bigr) 
		\wwrel\to 1
	\qquad\text{as }n\to \infty
	\;.
\label{eq:prob-Bnv}
\end{align}
We abbreviate $W_i \ce \sum_{j\in J_i}\!\! V_j$ for $i=1,2$ and let
$Z_n^v$ denote $Z_n$ conditional on $V=v$.
By $\Prob_V$ we denote the distribution of $V$.
Then, we obtain with \eqref{eq:eshy}
\begin{align*}
	\E \left[\left| \frac{Z_n}{n} - W_1W_2   \right|^2\right] 
		&\wwrel= 	
				\int \E \left[\left| 
						\frac{Z_n^v}{n} - w_1w_2
					\right|^2\right]
				d\,\Prob_V(v) \\
	 	&\wwrel\le	
	 			\iint_{B_n^v}
	 				\left| \frac{Z_n^v}{n} - w_1w_2 \right|^2
	 			d\,\Prob\; 
	 			\wbin+ 4\exp\bigl( -2n^{1/3} \bigr)
	 				\cdot \max_{0\le z\le n} \bigl|\tfrac zn - w_1 w_2\bigr|^2 d\,\Prob_V(v)
	 			\\ 
	 	&\wwrel\le 	
	 			4 n^{-1/3} 
	 			+ 4\exp\bigl( -2n^{1/3} \bigr) \\* 
	 	&\wwrel\to 	
	 			0
	 		\,, \qquad\text{as } n \to \infty \,,
\end{align*}
which concludes the proof.
\qed

\subsection{Proof of Lemma~\ref{lem:asymptotics-xlnx}}
Let $i\in \{1,\ldots,n\}$ be arbitrary. 
The strong law of large numbers implies that $L_i \mathbin/ n \to V_i$
\weakemph{almost surely} as $n\to\infty$. 
The function $x\mapsto x\ln(x)$ is continuous on~$[0,1]$ 
(with the convention $x\ln(x) = 0$ for $x=0$).
Hence, as $n\to \infty$,
\begin{align*}
		 \Bigl|
			\tfrac{L_i}n \ln \bigl( \tfrac{L_i}n \bigr) - V_i \ln V_i
		\Bigr|^2 
		&\wwrel\to 0 \qquad \mbox{almost surely} \;.
\end{align*}
Since the non-positive function $x\mapsto x\ln(x)$, $x\in [0,1]$ is lower
bounded (e.\,g.\ by $-1/e$) the square in the latter display is uniformly
bounded (e.\,g.\ by $(2/e)^2$). 
Hence the \weakemph{dominated convergence theorem} implies
\begin{align*}
		\E\Bigl[ \Bigl|
			\tfrac{L_i}n \ln \bigl( \tfrac{L_i}n \bigr) - V_i \ln V_i
		\Bigr|^2 \Bigr] 
		&\wwrel\to 0\,,
	\qquad\text{as } n\to\infty\;.
\end{align*}
This concludes the proof.
\qed

\section{Experimental Validation of Asymptotics}
\label{app:experimental-validation}

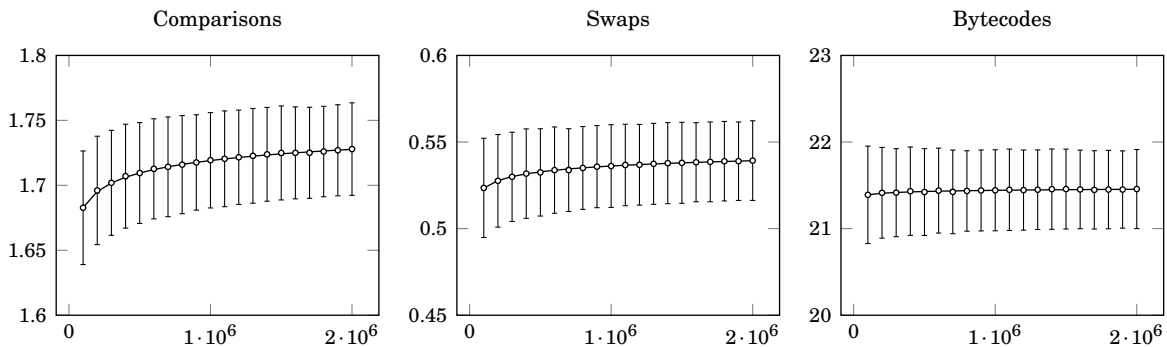
\begin{figure}
	\def\plotdatafile{plots/experimental-cmps-swaps-bytecodes.tab}
	\mbox{}\hfill%
	\tikzset{external/export=true}\tikzsetnextfilename{cmps-mean}%
	\begin{tikzpicture}
		\begin{axis}[
			width=0.37\linewidth,
		    		    title=Comparisons,
		    mark size=1pt,
 		    cycle list name=asymptotics,
 		    restrict x to domain=1000:inf,
 		    scaled ticks=false,
 		    ymin=1.6,ymax=1.8,
		   ]
			\addplot plot table[x=n,y={pred-mean-cmps/nlnn}] {\plotdatafile};
			\label{plot:asymptotic-means}
			\addplot plot[error bars/.cd, y dir=both, y explicit, 
			               error bar style={thin,solid}] 
			          table[x=n,y={mean-cmps/nlnn},y error={stdev-cmps/nlnn}]
			          {\plotdatafile};
			\label{plot:sample-means}
		\end{axis}
	\end{tikzpicture}\hfill%
	\tikzset{external/export=true}\tikzsetnextfilename{swaps-mean}%
	\begin{tikzpicture}
		\begin{axis}[
			width=0.37\linewidth,
		    		    title=Swaps,
		    mark size=1pt,
 		    cycle list name=asymptotics,
 		    restrict x to domain=1000:inf,
 		    scaled ticks=false,
 		    ymin=.45,ymax=.6,
		   ]
			\addplot plot table[x=n,y={pred-mean-swaps/nlnn}] {\plotdatafile};
			\addplot plot[error bars/.cd, y dir=both, y explicit, 
			               error bar style={thin,solid}] 
			          table[x=n,y={mean-swaps/nlnn},y error={stdev-swaps/nlnn}]
			          {\plotdatafile};
		\end{axis}
	\end{tikzpicture}\hfill%
	\tikzset{external/export=true}\tikzsetnextfilename{bytecodes-mean}%
	\begin{tikzpicture}
		\begin{axis}[
			width=0.37\linewidth,
		    		    title=Bytecodes,
		    mark size=1pt,
 		    cycle list name=asymptotics,
 		    restrict x to domain=1000:inf,
 		    scaled ticks=false,
 		    ymin=20,ymax=23,
		   ]
			\addplot plot table[x=n,y={pred-mean-bytecodes/nlnn}] {\plotdatafile};
			\addplot plot[error bars/.cd, y dir=both, y explicit, 
			               error bar style={thin,solid}] 
			          table[x=n,y={mean-bytecodes/nlnn},y error={stdev-bytecodes/nlnn}]
			          {\plotdatafile};
		\end{axis}
	\end{tikzpicture}%
	\hfill\mbox{}%
	\tikzset{external/export=false}
	\caption{%
		Comparison of sample means with the asymptotics
		computed in Section~\ref{sec:average-case-analysis}.
		The asymptotics are shown as solid lines, the
		sample means are indicated by circles. Additionally,
		the error bars show the sample standard deviation.
		On the $x$-axis, the input size $n$ is shown, the $y$-axis shows the
		corresponding counts normalized by $n \ln n$.
		\label{fig:means}
	}
\end{figure}

\begin{figure}
\def\plotdatafile{plots/experimental-cmps-swaps-bytecodes.tab}
	\mbox{}\hfill%
	\tikzset{external/export=true}\tikzsetnextfilename{cmps-var}%
	\begin{tikzpicture}
		\begin{axis}[
			width=0.37\linewidth,
		    		    title=Comparisons,
		    mark size=1.5pt,
 		    cycle list name=asymptotics,
 		    restrict x to domain=1000:inf,
 		    scaled ticks=false,
 		    ymin=.22,ymax=.3,
		   ]
			\addplot plot table[x=n,y={pred-var-cmps/n2}] {\plotdatafile};
			\label{plot:asymptotic-means}
			\addplot plot 
			          table[x=n,y={var-cmps/n2}]
			          {\plotdatafile};
			\label{plot:sample-means}
		\end{axis}
	\end{tikzpicture}\hfill%
	\tikzset{external/export=true}\tikzsetnextfilename{swaps-var}%
	\begin{tikzpicture}
		\begin{axis}[
			width=0.37\linewidth,
		    		    title=Swaps,
		    mark size=1.5pt,
 		    cycle list name=asymptotics,
 		    restrict x to domain=1000:inf,
 		    scaled ticks=false,
 		    yticklabel style={/pgf/number format/.cd,fixed,precision=3},
 		    ymin=0.090,ymax=0.125,
		   ]
			\addplot plot table[x=n,y={pred-var-swaps/n2}] {\plotdatafile};
			\addplot plot 
			          table[x=n,y={var-swaps/n2}]
			          {\plotdatafile};
		\end{axis}
	\end{tikzpicture}\hfill%
	\tikzset{external/export=true}\tikzsetnextfilename{bytecodes-var}%
	\begin{tikzpicture}
		\begin{axis}[
			width=0.37\linewidth,
		    		    title=Bytecodes,
		    mark size=1.5pt,
 		    cycle list name=asymptotics,
 		    restrict x to domain=1000:inf,
 		    scaled ticks=false,
 		    ymin=39,ymax=45,
		   ]
			\addplot plot table[x=n,y={pred-var-bytecodes/n2}] {\plotdatafile};
			\addplot plot
			          table[x=n,y={var-bytecodes/n2}]
			          {\plotdatafile};
		\end{axis}
	\end{tikzpicture}%
	\hfill\mbox{}%
	\tikzset{external/export=false}
	\caption{%
		Comparison of sample variances with the asymptotics
		computed in Section~\ref{sec:distributional-analysis}.
		The asymptotics are shown as solid lines, the
		sample variances are indicated by circles.
		On the $x$-axis, the input size $n$ is shown, the $y$-axis shows the
		corresponding variances normalized by $n^2$.
		\label{fig:variances}
	}
\end{figure}
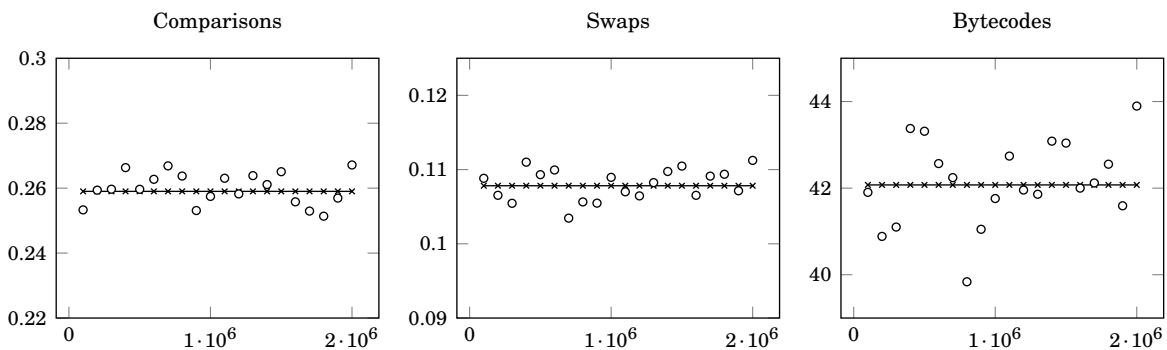

\begin{figure}
	\def\plotdatafile{plots/norm-cmps-1M-hist.tab}
	\tikzset{external/export=true}\tikzsetnextfilename{cmps-dist}%
	\begin{tikzpicture}
		\begin{axis}[
			title=Comparisons,
			width=.37\linewidth,
			ybar,
			bar width=4.4pt,
			yticklabel style={/pgf/number format/.cd,fixed,precision=2},			
		]
		\addplot plot[draw=black,fill=black!10]
				  table[x=norm-cmps-bin,y=prob] {\plotdatafile};		
		\end{axis}
	\end{tikzpicture}\hfill%
	\def\plotdatafile{plots/norm-swaps-1M-hist.tab}%
	\tikzset{external/export=true}\tikzsetnextfilename{swaps-dist}%
	\begin{tikzpicture}
		\begin{axis}[
			title=Swaps,
			width=.37\linewidth,
			ybar,
			bar width=4.4pt,
			yticklabel style={/pgf/number format/.cd,fixed,precision=2},			
		]
		\addplot plot[draw=black,fill=black!10]
				  table[x=norm-swaps-bin,y=prob] {\plotdatafile};		
		\end{axis}
	\end{tikzpicture}\hfill%
	\tikzset{external/export=true}\tikzsetnextfilename{bytecodes-dist}%
	\def\plotdatafile{plots/norm-bytecodes-1M-hist.tab}%
	\begin{tikzpicture}
		\begin{axis}[
			title=Bytecodes,
			width=.37\linewidth,
			ybar,
			bar width=4.4pt,
			yticklabel style={/pgf/number format/.cd,fixed,precision=2},			
		]
		\addplot plot[draw=black,fill=black!10]
				  table[x=norm-bytecodes-bin,y=prob] {\plotdatafile};		
		\end{axis}
	\end{tikzpicture}
	\caption{%
		Histograms for the distributions of the normalized number of comparisons
		$\tCsn$, swaps $\tSsn$ and executed Bytecode instructions $\tBCsn$
		from the sample of $10\,000$ random permutations of size $n=10^6$.
		\label{fig:histograms}
	}
\end{figure}
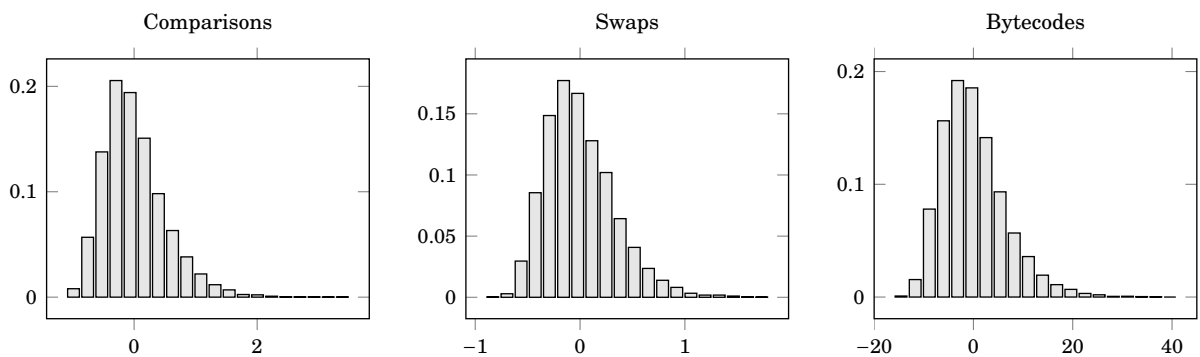

In this paper, we computed asymptotics for mean and variance of the costs
of Yaroslavskiy's algorithm.
Whereas the results for the mean are very precise and indeed can be made exact
with some additional diligence, our contraction arguments only provide leading
term asymptotics. 
In this section, we compare the asymptotic approximations with
experimental sample means and variances. 

We use the Java implementation given in Appendix~\ref{app:low-level} and run it
on $10\,000$ 
pseudo-randomly generated permutations of
$\{1,\ldots,n\}$ for each of the $20$ sizes in 
$\{10^5, \mbox{$2\cdot 10^5$},\ldots,\mbox{$2\cdot 10^6$}\}$. 
Note that for sensible estimates of variances, much larger samples are needed
than for means.
The experiment itself is done using our tool MaLiJAn, which
automatically counts the number of comparisons, swaps and Bytecode instructions
\cite{Wild2013Alenex}.

Figure~\ref{fig:means} shows the results for the expected costs
and Figure~\ref{fig:variances} compares asymptotic and sampled variances.
The histograms in Figure~\ref{fig:histograms} give some impression how the 
limit laws will look like.

It is clearly visible in Figure~\ref{fig:means} that for the given range of
input sizes, the average costs computed in
Section~\ref{sec:average-case-analysis} are extremely precise. 
In fact, hardly any deviation between prediction and measurement is
visible.
The variances in Figure~\ref{fig:variances} show more erratic behavior.
As variances are much harder to estimate than means, this does not come as a
surprise.
From the data we cannot tell whether the true variances show some oscillatory
behavior (in lower order terms) or whether we observe sampling noise.
Nevertheless, Figure~\ref{fig:variances} shows that for the given range of
sizes, the asymptotic is a sensible approximation of the exact variance.

\begin{figure}
	\def\plotsfile#1{plots/exact-norm-dist-n#1.tab}
	\def\myplot#1{%
		\addplot[y filter/.code={\pgfmathadd{##1}{-.015*#1}}]
		 plot[mark=none,color=black] table[x={cmp-mean/n},y=prob]
		 {\plotsfile{#1}}; 
	}
	\mbox{}\hfill%
	\tikzset{external/export=true}\tikzsetnextfilename{cmps-distributions-small-n}%
	\begin{tikzpicture}
		\begin{axis}[
			width=0.8\linewidth,
			height=0.5\linewidth,
			restrict x to domain*=-2.0:2.7,
			hide axis,
		]
 		\foreach \n in {3,...,28} {
 			\myplot{\n}
 		}
 		\begin{scope}[every node/.style={fill=white,inner sep=0pt,font=\tiny}]
			\node at (axis cs:2.7,-0.015*3) {$n=3$};
			\node at (axis cs:2.7,-0.015*5) {$n=5$};
			\node at (axis cs:2.7,-0.015*10) {$n=10$};
			\node at (axis cs:2.7,-0.015*15) {$n=15$};
			\node at (axis cs:2.7,-0.015*20) {$n=20$};
			\node at (axis cs:2.7,-0.015*25) {$n=25$};
		\end{scope}
		\draw (axis cs:0,-26.25*0.015) -- ++(0,-2ex) node[below] {mean} ;
 		\end{axis}
	\end{tikzpicture}
	\hfill\mbox{}
	\caption{%
		This figure shows the distributions of the normalized number of
		comparisons $\tCsn$ for small~$n$. The distributions are computed by
		unfolding the distributional
		recurrence~\eqref{eq:distributional-recurrence-comparisons}.
		This figure serves as visual evidence for the convergence of cost
		distributions to a limit law. 
		It is heavily inspired by a corresponding figure
		for classic Quicksort \protect\cite[Figure\,1.3]{SedgewickFlajolet1996}.
		\label{fig:limiting-process-dists}
	}
\end{figure}
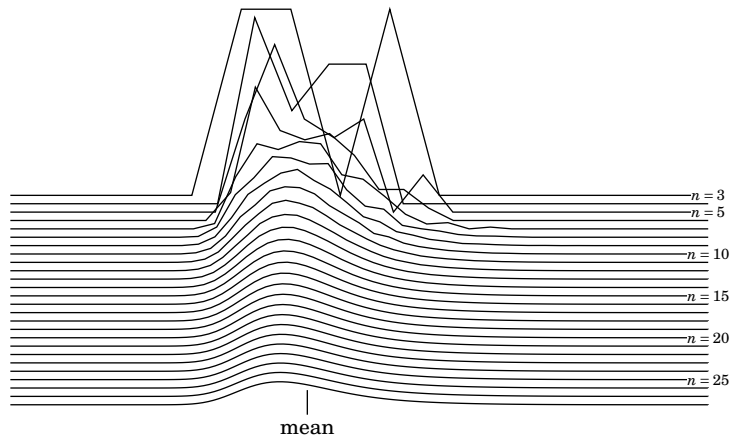

Figure~\ref{fig:limiting-process-dists} shows how fast the exact probability
distribution of the normalized number of comparisons approaches a smooth
limiting shape even for tiny~$n$.
This strengthens the above quantitative arguments that the limiting
distributions computed in this paper are useful approximations of
the true behavior of costs in Yaroslavskiy's algorithm.

\section*{ACKNOWLEDGEMENTS}
We would like to thank two anonymous reviewers for their helpful comments and
suggestions.

\FloatBarrier

\bibliographystyle{plainnat}
\bibliography{quicksort-refs}   

\end{document}